\documentclass{theoretics}

\usepackage{latexsym,
url,
bm,
amsfonts,
mathtools}

\usepackage{nicefrac}
\usepackage{tikz}

\RequirePackage[T1]{fontenc}
\usepackage[english]{babel}
\usepackage{alphabeta}

\usepackage{mathrsfs}
\usepackage{dsfont}
\usepackage{soul}

\usepackage{tikz,caption}
\usetikzlibrary{fadings,backgrounds,positioning,calc,shadows}
\tikzset{terminal/.style={circle, draw=black, fill=black!50, inner sep=0pt, minimum width=4pt}}
\tikzset{terminal2/.style={circle, draw=black, fill=purple!50, inner sep=0pt, minimum width=4pt}}
\tikzset{boundary/.style={circle, draw=red, fill=red!50, inner sep=0pt, minimum width=4pt}}
\tikzset{boundary2/.style={circle, draw=black!50, fill=black!20, inner sep=0pt, minimum width=4pt}}

\definecolor{darkmagenta}{rgb}{0.30, 0.0, 0.30}
\definecolor{MidnightBlack}{rgb}{0.1,0.1,.30}
\definecolor{MidnightBlue}{rgb}{0.1,0.1,0.78}
\definecolor{Black}{rgb}{0,0, 0}
\definecolor{Blue}{rgb}{0, 0 ,1}
\definecolor{Red}{rgb}{1, 0 ,0}
\definecolor{Darkgreen}{rgb}{0, .5 ,0}
\definecolor{White}{rgb}{1, 1, 1}
\definecolor{Grey}{rgb}{.6, .6, .6}
\definecolor{Mygreen}{rgb}{.0, .7, .0}
\definecolor{Yellow}{rgb}{.55,.55,0}
\definecolor{Mustard}{rgb}{1.0, 0.86, 0.35}
\definecolor{applegreen}{rgb}{0.55, 0.71, 0.0}
\definecolor{darkturquoise}{rgb}{0.0, 0.81, 0.82}
\definecolor{celestialblue}{rgb}{0.29, 0.59, 0.82}
\definecolor{green_yellow}{rgb}{0.68, 1.0, 0.18}
\definecolor{crimsonglory}{rgb}{0.75, 0.0, 0.2}
\definecolor{darkmagenta}{rgb}{0.30, 0.0, 0.30}
\definecolor{Magenta}{rgb}{1, 0.0, 1}
\definecolor{internationalorange}{rgb}{1.0, 0.31, 0.0}
\definecolor{darkorange}{rgb}{1.0, 0.55, 0.0}
\definecolor{ao}{rgb}{0.0, 0.5, 0.0}
\definecolor{awesome}{rgb}{1.0, 0.13, 0.32}
\definecolor{medium-blue}{rgb}{0,0,0.5}
\definecolor{dark-red}{rgb}{0.7,0.15,0.15}
\definecolor{dark-blue}{rgb}{0.15,0.15,0.4}

 \SetKwInput{KwInput}{Input}                % Set the Input
 \SetKwInput{KwOutput}{Output}              % Set the Output

\addto\extrasenglish{}
\addto\extrasenglish{}
\addto\extrasenglish{}
\addto\extrasenglish{}
\addto\extrasenglish{}
\addto\extrasenglish{}
\newcounter{func}
\newcommand{\newfun}[1]{f_{\refstepcounter{func}\label{#1}\thefunc}}
\newcommand{\funref}[1]{\hyperref[#1]{f_{\ref*{#1}}}}
\newcommand{\anc}{{\sf Anc}}
\newcommand{\bd}{{\sf bd}}
\newcommand{\bid}{{\sf bid}}
\newcommand{\bigmid}{{\;\big|\;}}
\newcommand{\bN}{{\mathbb{N}}}
\newcommand{\bR}{{\mathbb{R}}}
\newcommand{\cc}{{\sf cc}}
\newcommand{\ch}{{\sf Ch}}
\newcommand{\Char}{{\sf char}}
\newcommand{\compass}{{\sf Compass}}
\newcommand{\cupall}{{\pmb{\pmb{\bigcup}}}}
\newcommand{\crop}{{\sf crop}}
\newcommand{\desc}{{\sf Desc}}
\newcommand{\ed}{{\sf ed}}
\newcommand{\edrec}{{\tt recEd}}
\newcommand{\exc}{{\sf exc}}
\newcommand{\flaps}{{\sf Flaps}}
\newcommand{\FPT}{{\sf FPT}\xspace}
\newcommand{\forget}{{\sf forget}}
\newcommand{\fullchar}{{\sf char^*}}
\newcommand{\height}{{\sf height}}
\newcommand{\im}{{\sf Im}}
\newcommand{\influence}{{\sf influence}}
\newcommand{\Int}{{\sf Int}}
\newcommand{\intr}{{\sf intr}}
\newcommand{\join}{{\sf join}}
\newcommand{\Ker}{{\sf Ker}}
\newcommand{\leaf}{{\sf Leaf}}
\newcommand{\no}{{\sf no}}
\newcommand{\NP}{{\sf NP}}
\newcommand{\obs}{{\sf obs}}
\newcommand{\odd}{{\sf odd}}
\newcommand{\Par}{{\sf Par}}
\newcommand{\poly}{{\sf poly}\xspace}
\newcommand{\pretp}{{\preceq_{\sf tm}}}
\newcommand{\rep}{{\sf rep}}
\newcommand{\td}{{\sf td}}

\newcommand{\tw}{{\sf tw}}
\newcommand{\update}{{\sf update}}
\newcommand{\un}{{\mathds{1}}}
\newcommand{\XP}{{\sf XP}}
\newcommand{\yes}{{\sf yes}}
\newcommand{\Acal}{\mathcal{A}}
\newcommand{\Bcal}{\mathcal{B}}
\newcommand{\Ccal}{\mathcal{C}}
\newcommand{\Dcal}{\mathcal{D}}
\newcommand{\Ecal}{\mathcal{E}}
\newcommand{\Fcal}{\mathcal{F}}
\newcommand{\Gcal}{\mathcal{G}}
\newcommand{\Hcal}{\mathcal{H}}

\newcommand{\Mcal}{\mathcal{M}}

\newcommand{\Ocal}{\mathcal{O}}
\newcommand{\Pcal}{\mathcal{P}}
\newcommand{\Qcal}{\mathcal{Q}}
\newcommand\Rcal{\mathcal{R}}

\newcommand{\Tcal}{\mathcal{T}}

\newcommand{\Wcal}{\mathcal{W}}

\newcommand{\sugar}[1]{{#1}}

\renewcommand{\bigoplus}{\ensuremath{\vcenter{\hbox{\scalebox{1.4}{$\oplus$}}}}}

\title{Faster parameterized algorithms for modification problems to
  minor-closed classes}

\ThCSauthor[montpellier]{Laure Morelle}{laure.morelle@lirmm.fr}
\ThCSauthor[montpellier]{Ignasi Sau}{ignasi.sau@lirmm.fr}
\ThCSauthor[warsaw]{Giannos Stamoulis}{giannos.stamoulis@lirmm.fr}
\ThCSauthor[montpellier]{Dimitrios
  M. Thilikos}{sedthilk@thilikos.info}

\ThCSaffil[montpellier]{LIRMM, Université de Montpellier, CNRS,
  Montpellier, France}
\ThCSaffil[warsaw]{Institute of Informatics, University of Warsaw, Poland}

\ThCSthanks{All authors where supported by  the ANR projects DEMOGRAPH (ANR-16-CE40-0028), ELIT (ANR-20-CE48-0008), ESIGMA (ANR-17-CE23-0010), the French-German Collaboration ANR/DFG Project UTMA (ANR-20-CE92-0027). 
The first and the last author were also supported by the
Franco-Norwegian project PHC AURORA 2024 (Projet no. 51260WL). Most of the research work for this paper was conducted when Giannos Stamoulis was affiliated with LIRMM, Univ Montpellier, CNRS, Montpellier, France. A conference version of this article appeared in the Proceedings of the 50th International Colloquium on Automata,
  Languages, and Programming (ICALP) \cite{MorelleSST23fast}.}

\ThCSshortnames{L. Morelle, I. Sau, G. Stamoulis, D.M. Thilikos}
\ThCSshorttitle{Faster parameterized algorithms for modification problems to
  minor-closed classes}

\ThCSkeywords{Graph minors, Parameterized algorithms, Graph modification problems, Vertex deletion, Elimination distance, Irrelevant vertex technique, Flat Wall Theorem, Obstruction set}

\ThCSyear{2024}
\ThCSarticlenum{19}
\ThCSreceived{Jul 21, 2023} 
\ThCSrevised{May 31, 2024}
\ThCSaccepted{Jul 21, 2024}
\ThCSpublished{Aug 12, 2024}
\ThCSdoicreatedtrue
\begin{document}

\maketitle

\begin{abstract}
Let $\Gcal$ be a minor-closed graph class and let $G$ be an $n$-vertex graph.
 We say that $G$ is a \emph{$k$-apex} of $\Gcal$ if $G$ contains a set $S$ of at most $k$ vertices such that $G\setminus S$ belongs to~$\Gcal$. Our first result is an algorithm that decides whether $G$ is a $k$-apex of $\Gcal$ in time $2^{\poly(k)}\cdot n^2$, improving
a previous algorithm by Sau, Stamoulis, and Thilikos~[ICALP~2020, TALG 2022] whose running time was $2^{\poly(k)}\cdot n^3$.
The \emph{elimination distance} of $G$ to $\Gcal$, denoted by $\ed_\Gcal(G)$, is the minimum number of rounds required to reduce each connected component of $G$ to a graph in $\Gcal$ by removing one vertex from each connected component in each round.
Bulian and Dawar~[Algorithmica~2017] provided an \FPT-algorithm, with parameter $k$, to decide whether $\ed_\Gcal(G)\leq k$.
The class of graphs with $\ed_G\leq k$ is minor-closed, hence characterized by a finite set of excluded minors. The algorithm of Bulian and Dawar is based on the computability of this finite minor-obstruction set and its dependence on $k$ is not explicit.
We extend the techniques used in our first algorithm to decide whether $\ed_\Gcal(G)\leq k$ in time $2^{2^{2^{\poly(k)}}}\cdot n^2$.
This is the first algorithm for this problem with an explicit parametric dependence in $k$.
In the special case where $\Gcal$ excludes some apex-graph as a minor, we give two alternative algorithms, one running in time $2^{{2^{\Ocal(k^2\log k)}}}\cdot n^2$ and one running in time $2^{\poly(k)}\cdot n^3$.
As a stepping stone for these algorithms, we provide an algorithm that decides whether $\ed_\Gcal(G)\leq k$ in time $2^{\Ocal(\tw\cdot k + \tw\log\tw)}\cdot n$, where $\tw$ is the treewidth of $G$.
This algorithm combines the dynamic programming framework of
Reidl,  Rossmanith,  Villaamil, and  Sikdar~[ICALP 2014] for the particular case where $\Gcal$ contains only the empty graph (i.e., for treedepth) with the representative-based techniques introduced by Baste, Sau, and  Thilikos~[SODA 2020].
In the complexities above, {\sf poly} is a polynomial function whose degree depends on ${\cal G}$, and the hidden constants also depend on ${\cal G}$.
Finally, we provide explicit upper bounds on  the size of the graphs in the minor-obstruction set of the class of graphs ${\cal E}_k({\cal G})=\{G \mid \ed_\Gcal(G)\leq k\}$.
\end{abstract}

\section{Introduction}

The \emph{distance from triviality} is a concept formalized by Guo, Hüffner, and Niedermeier~\cite{GuoHN04astr} to express the closeness of a graph to a supposedly ``simple'' target graph class.
One such a measure of closeness is, for instance,  the number of vertices or edges that one must delete/add from/to a graph~$G$ to obtain a graph in the target graph class.
This concept of distance to a graph class has recently gained the interest of the parameterized complexity community.
The motivation is that, if a problem is tractable on a graph class $\Gcal$, it is natural to study other classes of graphs according to their ``distance to $\Gcal$''.
In this paper, we focus on two such measures of distance from triviality: Given a target graph class ${\cal G}$, we consider the {\sl vertex deletion distance} to ${\cal G}$ and the {\sl elimination distance} to ${\cal G}$, which we formalize next.

Given a target graph class ${\cal G}$ and a non-negative integer $k$,
we define ${\cal A}_{k}({\cal G})$ as the set of all graphs containing a set $S$ of at most $k$ vertices whose removal results in a graph in ${\cal G}$.
If $G\in {\cal A}_{k}({\cal G})$, then we say that $G$ is a \emph{$k$-apex} of $\Gcal$. We refer to $S$ as a \emph{$k$-apex set of $G$ for the class $\Gcal$}.
In other words, we consider the following meta-problem for a fixed class ${\cal G}$.
\begin{center}
	\fbox{
		\begin{minipage}{12cm}
			\noindent{\sc Vertex Deletion to $\Gcal$}\\
			\noindent\textbf{Input}:~~A graph $G$ and a non-negative integer $k$.\\
			\textbf{Objective}:~~Find, if it exists, a $k$-apex set of $G$ for the class $\Gcal$.
		\end{minipage}
	}
\end{center}
Throughout the paper, we denote by $n$ the number of vertices of the input graph of the problem under consideration. The importance of {\sc Vertex Deletion to $\Gcal$} can be illustrated by the variety of graph modification problems that it encompasses (hence the term of meta-problem).
For instance, if $\Gcal$ is the class of edgeless (resp. acyclic, planar, bipartite, (proper) interval, chordal) graphs, then we obtain the {\textsc{Vertex Cover}} (resp. {\textsc{Feedback Vertex Set}}, {\textsc{Vertex Planarization}}, {\textsc{Odd Cycle Transversal}}, {\sc (proper) Interval Vertex Deletion}, {\sc Chordal Vertex Deletion}) problem.

The second measure of distance from triviality that we study was recently introduced by Bulian and Dawar~\cite{BulianD16grap,BulianD17fixe}. Given a graph class ${\cal G}$, we define the \emph{elimination distance} of a graph $G$ to  $\Gcal$, denoted by $\ed_\Gcal(G)$, as follows:
\begin{equation*}
   \ed_\Gcal(G)=
   \begin{cases}
     0 & \text{if}\ G\in\cal G, \\
     1+\min\{\ed_\Gcal(G\setminus \{v\})\mid v\in V(G)\} & \text{if}\ G\text{ is connected}, \\
		 \max\{\ed_\Gcal(H)\mid H \text{ is a connected component of}\ G\} & \text{otherwise.}
   \end{cases}
\end{equation*}

Given that $\ed_{\cal G}(G)\leq k$, a set $S\subseteq V(G)$ of  vertices recursively deleted from $G$ to achieve $\ed_\Gcal(G)$ is called a \emph{$k$-elimination set of $G$ for ${\cal G}$}. We define the (parameterized)
class of graphs ${\cal E}_k({\cal G})=\{G\mid \ed_{\cal G}(G)\leq k\}$.
The above notion can be seen as a natural generalization of \emph{treedepth} (denoted by $\td$), which corresponds to the case where $\Gcal$ contains only the empty graph.
Treedepth, along with treewidth, are two of the most studied and widely used parameters to measure the structural complexity of a graph~\cite{CyganFKLMPPS15para,Kloks94,Sparsity12}.
The second meta-problem that we consider is the following, again for a fixed class $\Gcal$.
\begin{center}
	\fbox{
		\begin{minipage}{12cm}
			\noindent{\sc Elimination Distance to $\Gcal$}\\
			\noindent\textbf{Input}:~~A graph $G$ and a non-negative integer $k$.\\
			\textbf{Objective}:~~Find, if it exists, a $k$-elimination set of $G$ for the class $\Gcal$.
		\end{minipage}
	}
\end{center}

Unsurprisingly, {\sc Vertex Deletion to $\Gcal$} is \NP-hard for every non-trivial graph class $\Gcal$ \cite{LewisY80then}, while {\sc Elimination Distance to $\Gcal$} is \NP-hard even when $\Gcal$ contains only the empty graph \cite{Pothen88comp}.
To circumvent this intractability, we study both problems from the parameterized complexity point of view and consider their parameterizations by $k$.
In this setting, the most desirable behavior is the existence of an algorithm running in time $f(k) \cdot n^{\Ocal(1)}$, where $f$ is a computable function depending only on $k$. Such an algorithm is called \emph{fixed-parameter tractable}, or \emph{\FPT-algorithm} for short, and a parameterized problem admitting an \FPT-algorithm is said to belong to the parameterized complexity class \FPT. Also, the function $f$ is called {\em parametric dependence} of the corresponding \FPT-algorithm, and the challenge is to design \FPT-algorithms
with small parametric dependencies and with a polynomial factor of small degree~\cite{CyganFKLMPPS15para,DowneyF13fund,FlumG06para,Niedermeier06invi}.
We may also consider \emph{\XP-algorithms}, i.e., algorithms running in time $\Ocal(f(k)\cdot n^{g(k)})$ for some computable functions $f$ and $g$ depending only on $k$.

In general, for any of the two considered problems, we cannot expect \FPT-algorithms for every graph class $\Gcal$. For instance, the two problems are \NP-hard, even for $k=0$, for every graph class $\Gcal$ whose recognition problem is \NP-hard. This is the case of 3-colorable graphs, which is a class closed under taking (induced) subgraphs. In this paper, we focus on a family of graph classes that exhibits a nice behavior with respect to the considered problems (and many others): we consider $\Gcal$ to be a \emph{minor-closed} graph class, i.e.,
 such that every minor of a graph in ${\cal G}$ (that is, obtained from a subgraph of a graph in $\Gcal$ by contracting edges; see~\autoref{@authoritarianism} for the formal definition) is also in ${\cal G}$. Indeed, it turns out that, for every such a family $\Gcal$, the problems become fixed-parameter tractable, as we proceed to discuss.

 The {\em minor-obstruction set} (in short {\em obstruction set}) of ${\cal G}$ is the set of minor-minimal graphs that do not belong to ${\cal G}$, and is denoted by $\obs({\cal G})$. Notice that $\obs({\cal G})$ gives a complete characterization of ${\cal G}$
as, for every graph $G$, it holds that
$G\in{\cal G}$ if and only if, for every $H\in\obs({\cal G})$,
$H$ is not a minor of $G$.
Because of  Robertson and Seymour's theorem~\cite{RobertsonS04XX}, $\obs({\cal G})$ is {\sl finite} for every minor-closed graph class.
As checking whether an $h$-vertex graph $H$ is a minor of $G$ can be done in time $f(h)\cdot n^{2}$~\cite{RobertsonS95XIII,KawarabayashiKR12thed}, the finiteness of $\obs({\cal G})$ along with the above characterization imply that, for every minor-closed graph class ${\cal G}$, checking whether $G\in {\cal G}$ can be done in time $c \cdot n^{2}$, where $c$ is a constant depending on the graph class $\Gcal$.
This meta-theorem
implies the {\sl existence} of \FPT-algorithms for a wide
family of problems, including {\sc Vertex Deletion to $\Gcal$} and {\sc Elimination Distance to $\Gcal$}.
Indeed, this follows by observing that if ${\cal G}$ is minor-closed, then for every non-negative integer $k$, the classes ${\cal A}_{k}({\cal G})$ and ${\cal E}_{k}({\cal G})$ are also minor-closed.

As Robertson and Seymour's theorem~\cite{RobertsonS04XX} does {\sl not} give any way to construct the corresponding obstruction sets, the aforementioned argument is not constructive, i.e., it is not able to construct the obstruction sets  required for the corresponding
 \FPT-algorithms. Moreover, these algorithms are {\sl non-uniform} in $k$, meaning that we have a {\sl distinct} algorithm for every value of $k$. Important steps towards the constructibility of
such \FPT-algorithms were done by Adler, Grohe, and  Kreutzer~\cite{AdlerGK08comp} and Bulian and Dawar~\cite{BulianD17fixe}, who respectively proved that
$\obs({\cal A}_{k}({\cal G}))$ and $\obs({\cal E}_{k}({\cal G}))$ are effectively computable.
Hence,  for both problems, it is possible to construct uniform (in $k$) algorithms running in time $f(k)\cdot n^2$ for some computable function $f$. However, this does not imply any reasonable, or even explicit, parametric dependence of the obtained algorithms.

The main focus of this paper is on the parametric and polynomial dependence of \FPT-algorithms to solve {\sc Vertex Deletion to $\Gcal$} and {\sc Elimination Distance to $\Gcal$}, i.e., for recognizing the classes $\Acal_{k}(\Gcal)$ and $\Ecal_k(\Gcal)$, when $\Gcal$ is a minor-closed graph class.

Concerning {\sc Vertex Deletion to $\Gcal$}, after a number of articles for particular cases of minor-closed classes $\Gcal$, such as graphs of bounded treewidth~\cite{FominLMS12plan,KimLPRRSS16line}, planar graphs~\cite{MarxS07obta,JansenLS14anea}, or graphs of bounded genus~\cite{KociumakaP19dele}, an explicit \FPT-algorithm for {\sl any} minor-closed graph $\Gcal$ was recently proposed by Sau, Stamoulis, and Thilikos~\cite{SauST21kapiII}, running in time $2^{\Ocal(k^c)}\cdot n^3$, where $c$ is a constant that depends on the maximum size of a graph in the obstruction set of $\Gcal$.
Moreover, in the case where $\obs(\Gcal)$ contains some apex-graph (that is, a 1-apex for the class of planar graphs), Sau, Stamoulis, and Thilikos~\cite{SauST21kapiII} gave an improved running time of $2^{\Ocal(k^c)}\cdot n^2$.
Note also that the more general variant where $\Gcal$ is a topological-minor-closed graph class is in $\FPT$ as well \cite{FominLPSZ20hitt}.

As for {\sc Elimination Distance to $\Gcal$} when $\Gcal$ is minor-closed, no explicit parametric dependence was known, with the notable exception of treedepth,
for which Reidl, Rossmanith, Villaamil, and Sikdar~\cite{ReidlRSS14afas} gave an algorithm deciding whether $\td(G)\leq k$  in time $2^{\Ocal(k\cdot \tw)}\cdot n$, where $\tw:=\tw(G)$ (see also \cite{BodlaenderGKH91appr}).
Using our terminology, and given that $\tw(G)\leq \td(G)$ for every graph $G$,
this yields an \FPT-algorithm for {\sc Elimination Distance to $\Gcal_\emptyset$}, where $\Gcal_\emptyset$ is the class consisting of the empty graph, running in time $2^{\Ocal(k^2)}\cdot n$. Note that this algorithm~\cite{ReidlRSS14afas}, combined with the fact that $\td(G)\leq \log(n)\cdot \tw(G)$ (see \cite{BodlaenderGKH91appr}),
imply an \XP-algorithm
for the problem of computing $\td$ when parameterized by $\tw$, namely an algorithm that computes the value of $\td(G)$ in time $n^{\Ocal(\tw(G)^2)}$. To the best of our knowledge, it is open whether computing $\td$  parameterized by $\tw$ is in \FPT.

Before describing our results, let us mention some recent relevant results dealing with {\sc Elimination Distance to $\Gcal$} for classes $\Gcal$ that are not necessarily minor-closed. Agrawal and Ramanujan~\cite{Agrawal020} (resp. Agrawal, Kanesh,  Panolan,  Ramanujan, and  Saurabh~\cite{AgrawalKP0021}) provided \FPT-algorithms, with parameter $k$, when $\Gcal$ is the class of cliques (resp. graphs of bounded degree). Fomin, Golovach, and Thilikos~\cite{FominGT22} identified sufficient and necessary conditions for the existence of \FPT-algorithms when $\Gcal$ is definable in first-order logic (such as having bounded degree). Jansen, de Kroon, and  Włodarczyk~\cite{JansenKW21vert} proved, among a number of other results, that if $\Gcal$ is a hereditary union-closed graph class and {\sc Vertex Deletion to $\Gcal$} can be solved in time $2^{k^{\Ocal(1)}}\cdot n^{\Ocal(1)}$ (as it is the case for every minor-closed class $\Gcal$ by the results of~\cite{SauST21kapiII}), then there is an algorithm that, given an $n$-vertex graph $G$, computes an $\Ocal(\ed_\Gcal(G)^3)$-elimination set of $G$ for ${\cal G}$ in time $2^{\ed_\Gcal(G)^{\Ocal(1)}}\cdot n^{\Ocal(1)}$. Therefore, for union-closed minor-closed graph classes~$\Gcal$, the result of~\cite{JansenKW21vert}
yields an \FPT-approximation algorithm for {\sc Elimination Distance to $\Gcal$}.

Note that, for every graph class $\Gcal$ and every graph $G$, it holds that $\ed_\Gcal(G)$ is not larger than the smallest $k$ such that $G$ admits a $k$-apex for $\Gcal$. Thus, if {\sc Elimination Distance to $\Gcal$} is in \FPT, so is {\sc Vertex Deletion to $\Gcal$}. Agrawal, Kanesh,  Lokshtanov,  Panolan, Ramanujan,  Saurabh, and  Zehavi~\cite{AgrawalKLPRSZ22} showed, among other results, that in many cases the reverse implication also holds. Namely, they proved that if $\Gcal$  is hereditary, union-closed, and definable in monadic second-order logic, and {\sc Vertex Deletion to $\Gcal$} is in \FPT, then {\sc Elimination Distance to $\Gcal$} is also (non-uniformly) in \FPT. Incidentally, they also showed that if $\Gcal$ is defined by excluding a finite number of connected topological minors, then {\sc Elimination Distance to $\Gcal$} is  (uniformly) in \FPT.
We note that the results of~\cite{AgrawalKLPRSZ22} do not provide explicit parametric dependencies for these \FPT-algorithms. Also, let us mention that it was conjectured in~\cite{AgrawalKLPRSZ22}  that {\sc Elimination Distance to $\Gcal$} is in \FPT\ parameterized by a generalization of treewidth called \emph{$\Gcal$-treewidth} (see~\cite{JansenKW21vert,AgrawalKLPRSZ22,EibenGHK21}). Note that, if true, this conjecture would answer the open problem mentioned above of whether computing $\td$ parameterized by $\tw$ is in \FPT.

\bigskip
\noindent\textbf{Our results.}
In this paper, we provide explicit \FPT-algorithms for {\sc Vertex Deletion to $\Gcal$} and {\sc Elimination Distance to $\Gcal$} for {\sl every} fixed minor-closed graph class $\Gcal$. Our first result is the following.

\begin{theorem}\label{@sensations}
For every minor-closed graph class $\Gcal$, there exists an algorithm that solves  {\sc Vertex Deletion to $\Gcal$} in time $2^{k^{\Ocal(1)}}\cdot n^2$.
\end{theorem}
The degree of $k$ in the running time of \autoref{@sensations}, as well as the constants hidden in the $\cal O$-notation in the running time of the algorithms of the results below, depend on the maximum size of a graph in $\obs({\cal G})$. Thus, the algorithm of \autoref{@sensations}, while being uniformly \FPT in $k$, is {\sl not} uniform in the target class $\Gcal$, as one needs to know an upper bound on the size of the minor-obstructions. This ``meta-non-uniformity'' applies to all the algorithms presented in this paper, and it is also the case, among many others, of the \FPT-algorithms in~\cite{SauST21kapiII}. The algorithm of \autoref{@sensations} improves the algorithm of~\cite{SauST21kapiII} from cubic to quadratic complexity in $n$ while keeping the same parametric dependence on $k$. This answers positively one of the open problems posed in~\cite{SauST21kapiII}.

\smallskip
Our next algorithmic results concern {\sc Elimination Distance to $\Gcal$} and provide, to the authors' knowledge, the first \FPT-algorithms for this problem, when $\Gcal$ is minor-closed, with an {\sl explicit parametric dependence}.
\begin{theorem}\label{@communicated}
For every minor-closed graph class $\Gcal$, there exists an algorithm that solves {\sc Elimination Distance to $\Gcal$}  in time $2^{2^{2^{k^{\Ocal(1)}}}}\cdot n^2$. In the particular case where $\obs({\cal G})$ contains an apex-graph, this algorithm runs in time $2^{2^{\Ocal(k^2\log k)}}\cdot n^2$.
\end{theorem}

As examples of classes $\Gcal$ where $\obs({\cal G})$ contains an apex-graph, we may consider $\Gcal$ whose graphs have bounded Euler genus, such as planar graphs. Our next result improves the parametric dependence of the algorithm of \autoref{@communicated} when $\obs({\cal G})$ contains an apex-graph, but with a worse polynomial factor.

\begin{theorem}\label{@swineherds}
For every minor-closed graph class $\Gcal$ such that $\obs({\cal G})$ contains an apex-graph, there exists an algorithm that solves  {\sc Elimination Distance to $\Gcal$} in time $2^{k^{\Ocal(1)}}\cdot n^3$.
\end{theorem}

As discussed later, a crucial ingredient in the algorithms of \autoref{@communicated} and \autoref{@swineherds} is to solve {\sc Elimination Distance to $\Gcal$} parameterized by the treewidth of the input graph. The following result, which may be of independent interest, deals with this case.

\begin{theorem}\label{@achievements}
For every minor-closed graph class $\Gcal$, there exists an algorithm that solves {\sc Elimination Distance to $\Gcal$} in time $2^{\Ocal(k\cdot\tw+\tw\log\tw)}\cdot n$, where $\tw$ denotes the treewidth of the input graph.
\end{theorem}

The algorithm of \autoref{@achievements} can be seen as a generalization of the algorithm of Reidl,  Rossmanith,  Villaamil, and  Sikdar~\cite{ReidlRSS14afas} deciding whether $\td(G)\leq k$  in time $2^{\Ocal(k\cdot \tw)}\cdot n$.  Since, for any graph $G$ and any graph class $\Gcal$, $\ed_{\Gcal}(G)\leq \td(G)\leq \tw(G)\cdot\log n$, \autoref{@achievements} implies the existence of an \XP-algorithm
for {\sc Elimination Distance to $\Gcal$} parameterized by treewidth, when $\Gcal$ is minor-closed, running in time $n^{\Ocal(\tw^2)}$.
Given that the conjecture of  \cite{AgrawalKLPRSZ22} is still open, this is the best type of algorithm that one can expect for {\sc Elimination Distance to $\Gcal$} parameterized by treewidth.
Furthermore, since $\tw(G)\leq\td(G)$ for any graph $G$, \autoref{@achievements} implies an \FPT-algorithm for {\sc Elimination Distance to $\Gcal$} parameterized by treedepth, running in time $2^{\Ocal(\td^2)}\cdot n$.

\smallskip
Finally, for any minor-closed graph class $\Gcal$, we provide an upper bound on the size of the graphs in the obstruction set of $\Ecal_k(\Gcal)$.

\begin{theorem}\label{@reconstructions}
For every minor-closed graph class $\Gcal$ and for every positive integer $k$, each graph in $\obs(\Ecal_k (\Gcal))$ has at most $2^{2^{2^{2^{k^{\Ocal(1)}}}}}$ vertices.
Moreover, if $\obs(\Gcal)$ contains an apex-graph, this bound drops to $2^{2^{k^{\Ocal(1)}}}$.
\end{theorem}

The only previously known bound for the graphs in $\obs(\Ecal_k(\Gcal))$ is the one for treedepth by Dvořák, Giannopoulou, and Thilikos~\cite{DvorakGT12forb}, who proved that
every graph in $\obs(\Ecal_k(\Gcal_{\emptyset}))$
has size at most $2^{2^{k-1}}$. \autoref{@reconstructions} can be seen as a generalization of the results of Sau, Stamoulis, and Thilikos~\cite{SauST21kapiI}, who provided similar upper bounds for the graphs in $\obs(\Acal_k (\Gcal))$.

These two results are, to the authors' knowledge, the first upper bounds on the size of the graphs in the obstruction set for the treedepth and the elimination distance parameters, and give, as an immediate consequence, the first known upper bound for the size of these obstruction sets.

\medskip
\noindent\textbf{Our techniques.}
We now proceed to provide a high-level overview of the main tools used to prove our results, without getting into technical details.
This paper builds heavily on the techniques recently introduced in~\cite{SauST21kapiII} in order to deal with {\sc Vertex Deletion to $\Gcal$}, which are based on exploiting the Flat Wall Theorem of Robertson and Seymour~\cite{RobertsonS95XIII}, namely the version proved by Kawarabayashi, Thomas, and Wollan~\cite{KawarabayashiTW18anew} and its recent restatement
by Sau, Stamoulis, and Thilikos~\cite{SauST21amor}. In a nutshell, the idea of \autoref{@sensations}, \autoref{@communicated}, and \autoref{@swineherds} is that, as far as the treewidth of the input graph is sufficiently large as an appropriate function of $k$, it is possible to either ``branch'' into a number of subproblems that depends only on $k$ and where the value of the parameter is strictly smaller, or to find an irrelevant vertex (i.e., a vertex that does not change the answer to the considered problem) and remove it from the graph. The irrelevant vertex technique originates from Robertson and Seymour \cite{RobertsonS95XIII} and is further developed in \cite{SauST21amor,SauST21kapiI,SauST21kapiII}. Once the treewidth is bounded, what remains is to apply the most efficient possible algorithm to solve the problem via dynamic programming on tree decompositions.

Let us focus more particularly on the techniques we use to prove \autoref{@sensations}.
Contrary to the algorithm of~\cite{SauST21kapiII} that solves {\sc Vertex Deletion to $\Gcal$} for any minor-closed class $\Gcal$, we {\sl avoid using iterative compression}. This explains the improvement from cubic to quadratic complexity in $n$.
The algorithm of \autoref{@sensations} can be seen as an extension of the algorithm of~\cite{SauST21kapiII} that solves {\sc Vertex Deletion to $\Gcal$} in the particular case where $\obs(\Gcal)$ contains some apex-graph, and uses ideas
 that date back to the work of Marx and Schlotter~\cite{MarxS07obta} for the {\sc Planarization} problem, that is, when ${\cal G}$ is the class of planar graphs. In \autoref{@recipients} we provide a sketch of the algorithms claimed in \autoref{@sensations}, \autoref{@communicated}, and
\autoref{@swineherds}, and in \autoref{@conditioned}, \autoref{@oeconomicus}, and \autoref{@stubbornly}, respectively, we present the algorithms in full detail, along with a proof of their correctness.

The proof of  \autoref{@achievements} consists of a dynamic programming algorithm
that combines the framework of~\cite{ReidlRSS14afas} for the particular case where $\Gcal$ contains only the empty graph (i.e., for treedepth) with the representative-based techniques introduced in~\cite{BasteST20acom}.
A bit more precisely, the idea is to encode the partial solutions (called \emph{characteristic}) via sets of annotated trees with some additional properties.
Here, the trees correspond to partial elimination trees and the annotations indicate the representatives, in the leaves of the elimination trees, with respect to the canonical equivalence relation defined for the target class $\Gcal$.
The size of the characteristic (cf.~\autoref{@belligerent}) dominates the running time of the whole algorithm. As usual when dealing with dynamic programming, the formal description of the algorithm, given in \autoref{@fanaticism}, is quite technical and lengthy.

Finally, to obtain the upper bound on the size of a graph $G\in\obs(\Ecal_k(\Gcal))$ claimed in \autoref{@reconstructions}, we proceed in two steps.
First, we bound the treewidth of $G$ by a function of $k$. To do so, we observe that if the treewidth of $G$ is big enough, then there is a big enough wall in $G$, and we find an irrelevant vertex $v$ for {\sc Elimination Distance to $\Gcal$} in $G$. However, $G\setminus\{v\}\in\Ecal_k(\Gcal)$ and $G\notin\Ecal_k(\Gcal)$, hence we reach a contradiction.
The second step is to bound the size of a minor-minimal obstruction of small treewidth. This uses the classic technique of Lagergren~\cite{Lagergren98uppe} (see also \cite{GiannopoulouPRT19cutw,GiannopoulouKRT19lean,KanteK14anup,KanteK18line,GiannopoulouPRT17line,Lagergren91anup,LagergrenA91mini,SauST21kapiI}) combined with the encoding of the tables of the
dynamic programming algorithm that we use to prove \autoref{@achievements}; see \autoref{@participant}.

\section{Basic definitions and restatement of the problems}\label{@immeasurable}
In \autoref{@authoritarianism} we give some basic definitions on graphs and minors
and in \autoref{@graphically} we redefine the problems in a more convenient way and we establish some conventions that we will use throughout the paper.

\subsection{Basic definitions}\label{@authoritarianism}

\subparagraph{Sets and integers.}
We denote by $\bN$ the set of non-negative integers.
Given two integers $p$ and $q$, the set $[p,q]$ contains every integer $r$ such that $p\leq r\leq q$.
For an integer $p\geq 1$, we set $[p]=[1,p]$ and $\bN_{\geq p}=\bN\setminus [0,p-1]$.
Given a non-negative integer $x$,
we denote by $\odd(x)$ the smallest odd number that is not smaller than $x$.
For a set $S$, we denote by $2^{S}$ the set of all subsets of $S$ and, given an integer $r\in[|S|]$,
we denote by $\binom{S}{r}$ the set of all subsets of $S$ of size $r$ and by $\binom{S}{\leq r}$ (resp. $\binom{S}{< r}$) the set of all subsets of $S$ of size at most $r$ (resp. $r-1$).
If ${\cal S}$ is a collection of objects where the operation $\cup$ is defined,
then we denote $\cupall {\cal S}=\bigcup_{X\in {\cal S}}X$.

\subparagraph{Basic concepts on graphs.}
All graphs considered in this paper are undirected, finite, and without loops or multiple edges.
We use standard graph-theoretic notation and we refer the reader to \cite{Diestel10grap} for any undefined terminology.
For convenience, we use $uv$ instead of $\{u,v\}$ to denote an edge of a graph.
Let $G$ be a graph. In the rest of this paper we always use $n=|G|$ for the size of $G$, i.e., the cardinality of $V(G)$,
and $m$ for the cardinality of $E(G)$, where $G$ is the input graph of the problem under consideration.
We say that a pair $(L,R)\in 2^{V(G)}\times 2^{V(G)}$ is a {\em separation} of $G$
if $L\cup R=V(G)$ and there is no edge in $G$ between $L\setminus R$ and $R\setminus L$.
The \emph{order} of $(L,G)$ is $|L\cap G|$.
Given a vertex $v\in V(G)$, we denote by $N_{G}(v)$ the set of vertices of $G$ that are adjacent to $v$ in $G$.
A vertex $v \in V(G)$ is \emph{isolated} if $N_G(v) = \emptyset$.
For $S \subseteq V(G)$, we set $G[S]=(S,E\cap \binom{S}{2} )$
and use the shortcut $G \setminus S$ to denote $G[V(G) \setminus S]$.
We may also use $G\setminus v$ instead of $G\setminus\{v\}$ for $v\in V(G)$.
For $A,B\subseteq V(G)$, $E(A,B)$ denotes the set of edges of $G$ with one extremity in $A$ and the other in $B$.
We use $\cc(G)$ to denote the set of connected components of $G$.

\subparagraph{Dissolutions and subdivisions.}
Given a vertex $v\in V(G)$ of degree two with neighbors $u$ and $w$, we define the {\em dissolution} of $v$
to be the operation of deleting $v$ and, if $u$ and $w$ are not adjacent, adding the edge $uw$.
Given two graphs $H$ and $G$, we say that $H$ is a {\em dissolution} of $G$
if $H$ can be obtained from $G$ after dissolving vertices of $G$.
Given an edge $e=uv\in E(G)$, we define the {\em subdivision} of $e$
to be the operation of deleting $e$, adding a new vertex $w$ and making it adjacent to $u$ and $v$.
Given two graphs $H$ and $G$, we say that $H$ is a {\em subdivision} of $G$
if $H$ can be obtained from $G$ after subdividing edges of $G$.
Observe that $G$ is a subdivision of $H$ iff $H$ is a dissolution of $G$.

\subparagraph{Contractions and minors.}
The \emph{contraction} of an edge $e = uv$ of a simple graph $G$ results in a simple graph $G'$
obtained from $G \setminus \{u,v\}$ by adding a new vertex $w$ adjacent to all the vertices
in the set $N_G(u) \cup N_G(v)\setminus \{u,v\}$.
A graph $G'$ is a \emph{minor} of a graph $G$, denoted by $G'\preceq G$,
if $G'$ can be obtained from $G$ by a sequence of vertex removals, edge removals, and edge contractions.
If only edge contractions are allowed, we say that $G'$ is a \emph{contraction} of $G$.
Let $H$ be a graph that is a minor of a graph $G$. We call any subgraph $M$ of $G$ that can be contracted to $H$ a {\em model} of $H$ in $G$.
Given two graphs $H$ and $G$, if $H$ is a minor of $G$ then for every vertex $v\in V(H)$ there is
a set of vertices in $G$ that are the endpoints of the edges of $G$ contracted towards creating $v$.
We call this set {\em model} of $v$ in $G$.
Given a finite collection of graphs $\Fcal$ and a graph $G$,
we use notation $\Fcal\preceq G$ to denote that some graph in $\Fcal$ is a minor of $G$.

\subparagraph{Minor obstructions.}
Let $\Gcal$ be a graph class that is closed under taking minors.
Recall that the {\em minor obstruction set} of $\Gcal$ is defined as the set of all minor-minimal graphs that are not in $\Gcal$, and is denoted by $\obs(\Gcal)$.
Given a finite non-empty collection of non-empty graphs $\Fcal$, we denote by $\exc(\Fcal)$ the set containing every graph $G$ that excludes all graphs in $\Fcal$ as minors.
We call each graph in $\exc(\Fcal)$ {\em $\Fcal$-minor-free}.
We use $\Gcal_{\emptyset}$ for the graph class containing only the empty graph $G_{\emptyset}$. Notice that $\obs(\Gcal_{\emptyset})=\{K_{1}\}$.

\subsection{Restating the problems}\label{@graphically}

Let $\Gcal$ be a minor-closed graph class and $\Fcal$ be its obstruction set.
Clearly, {\sc Vertex Deletion to $\Gcal$} is the same problem as
asking, given a graph $G$ and some $k\in \bN$, for a vertex set $S\subseteq V(G)$ of at most $k$ vertices
such that $G\setminus S\in \exc(\Fcal)$.
Following the terminology of~\cite{BasteST20hittI,BasteST20hittII,BasteST20hittIII,BasteST20acom,FominLPSZ20hitt,FominLMS12plan,KimLPRRSS16line,KimST18data,SauST21kapiII},
we call this problem {\sc $\Fcal$-M-Deletion}.
Likewise, {\sc Elimination Distance to $\Gcal$} is the same problem as
asking whether $\ed_{\exc(\Fcal)}(G)\leq k$.
We will thus follow a similar notation and call this problem {\sc $\Fcal$-M-Elimination Distance}. Using the notation, {\sc $\{K_{1}\}$-M-Elimination Distance} is the problem of asking whether $\td(G)\leq k$.
We say that ${\cal F}$ is {\em non-trivial} when all graphs in ${\cal F}$ contain at least two vertices.

\subparagraph{Some conventions.}
In the rest of the paper, we fix $\Gcal$ to be a minor-closed graph class and $\Fcal$ to be the set $\obs(\Gcal)$.
From Robertson and Seymour's theorem \cite{RobertsonS04XX}, we know that $\Fcal$ is a finite collection of graphs.
Given a graph $G$, we define its {\em apex number} to be the smallest integer $a$ for which there is a set $A\subseteq V(G)$ of size at most $a$ such that $G\setminus A$ is planar.
An \emph{apex-graph} is a graph with apex number one.
Also, we define the {\em detail} of $G$, denoted by ${\sf detail}(G)$, to be the maximum among $|E(G)|$ and $|V(G)|$.
We define three constants depending on $\Fcal$ that will be used throughout the paper whenever we consider such a collection $\Fcal$.
We define $a_\Fcal$ as the minimum apex number of a graph in $\Fcal$, we set $s_\Fcal:=\max\{|V(H)|\mid H\in \Fcal\}$, and we set $\ell_\Fcal:=\max\{ {\sf detail}(H)|\mid H\in\Fcal\}$.
Given a tuple $\textbf{t}=(x_{1},\ldots,x_{\ell})\in \bN^{\ell}$ and two functions $\chi,\psi:\bN\to\bN$, we write $\chi(n)=\Ocal_{\textbf{t}}(\psi(n))$ in order to denote that there exists a computable function $\phi:\bN^{\ell} \rightarrow \bN$ such that $\chi(n)=\Ocal(\phi(\textbf{t})\cdot \psi(n))$.
Notice that $s_\Fcal\leq\ell_\Fcal\leq s_\Fcal(s_\Fcal-1)/2$, and thus $\Ocal_{\ell_\Fcal}(\cdot)=\Ocal_{s_\Fcal}(\cdot)$.
Observe also that $\Acal_k(\Gcal)$ and $\Ecal_k(\Gcal)$ are $K_{s_\Fcal+k}$-minor-free graph classes, and thus, due to \cite{Thomason01thee}, we can always assume that $G$ has $\Ocal_{s_\Fcal}(k\sqrt{\log k}\cdot n)$ edges, otherwise we can directly conclude that $(G,k)$ is a \no-instance for both problems.

\section{Sketch of the algorithms}\label{@recipients}

Before going further through the definitions, let us provide a sketch of the algorithms claimed in \autoref{@sensations}, \autoref{@communicated} and \autoref{@swineherds}.
As mentioned in the introduction, \autoref{@sensations} can be seen as a generalization of the algorithm of~\cite{SauST21kapiII} that solves {\sc $\Fcal$-M-Vertex Deletion} in the particular case where $\Fcal$ contains some apex-graph. While many techniques taken from~\cite{SauST21kapiII} remain the same,
some new ingredients are needed so as to deal with the possible existence of many apices in all graphs in $\Fcal$. On the other hand, \autoref{@communicated} and \autoref{@swineherds} can be seen as an adaptation of \autoref{@sensations} to {\sc $\Fcal$-M-Elimination Distance}. Since these three algorithms follow a common streamline, we sketch all of them simultaneously while pointing out the steps where they differ.
Moreover, the full proofs of \autoref{@sensations}, \autoref{@communicated}, and \autoref{@swineherds} are given in \autoref{@conditioned}, \autoref{@oeconomicus}, and \autoref{@stubbornly}, respectively.

\subparagraph{Walls and flat walls.}
In this paper we extensively deal with walls and flat walls, following the framework of \cite{SauST21amor}.
Unfortunately, almost ten pages are required to provide all the technical notions to correctly present all this framework, that is necessary to use the tools developed in \cite{SauST21amor,SauST21kapiI,SauST21kapiII}. Thus, we only give some intuition on those definitions for the sketch of the algorithm, while the formal definitions are deferred to \autoref{@envisioning}.
More precisely, in \autoref{@postulated}, we introduce walls and several notions concerning them (just look at \autoref{@manuscript} to understand what a wall is).
In \autoref{@commiseratio} we provide the definitions of a rendition and a painting, which are not crucial for understanding this section.
Using the above notions, in \autoref{@tugendlehre} we define flat walls and flatness pairs. There are a number of technical terms (such as tilts, influence, regular flatness pairs, ...) that are not the main focus of the sketch.
Let us just mention that the perimeter of a flat wall of a graph $G$ separates $V(G)$ into two sets $X$ and $Y$ with $Y$ containing the wall. The \emph{compass} of a flat wall is $G[Y]$.
See \autoref{figflat}: $X$ is the set of vertices in the light green part, and $Y$ the set of vertices in the pink part.

\begin{figure}[ht]
	\centering
	\includegraphics[scale=2]{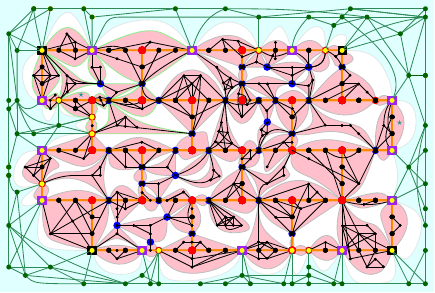} % TS changed
                                % from 1.3
	\caption{\small Illustration of a flat wall $W$ inside a graph $G$. The edges of $W$ are depicted in orange and the compass of $W$ in $G$ is the union of all parts of $G$ that are drawn in pink cells. The graph in each such pink cell corresponds to a flap of the flat wall.}
	\label{figflat}
\end{figure}

In \autoref{@commodation} and~\autoref{@unequivocally} we define canonical partitions and the notion of bidimensionality. Informally speaking, a \emph{canonical partition} of a graph with respect to some wall $W$ refers to a partition of the vertex set of a graph in bags that follow the structure of a wall subgraph of the given graph; see \autoref{@aberration} for an illustration.
The \emph{bidimensionality} of a vertex set $X$ with respect to a wall $W$ of a graph $G$ intuitively expresses the ``spread'' of a set $X$ in a $W$-canonical partition of $G$. The crucial idea is that a set $X$ of small bidimensionality cannot ``destroy'' a large (flat) wall too much.

Finally, in \autoref{@indigenous} we present homogeneous walls.
Intuitively, \emph{homogeneous flat walls} are flat walls that allow the routing of the same set of (topological) minors in the augmented flaps (i.e., the flaps together with the apex set) ``cropped'' by each one of their bricks.
Such a homogeneous wall can be detected in a big enough flat wall (\autoref{@disreputable}) and this ``homogeneity'' property implies
that some central part of a big enough homogeneous wall can be declared irrelevant (\autoref{@civilizing}).
\medskip

The first common step is to run the following algorithm that states that a graph $G$ in $\Acal_k(\exc(\Fcal))$ or  $\Ecal_k(\exc(\Fcal))$ either has bounded treewidth (see \autoref{@complacency} for the formal definition) or contains a large wall.
This result was proved in \cite{SauST21kapiII} in the case of {\sc $\Fcal$-M-Deletion}.
The proof in the case of {\sc $\Fcal$-M-Elimination Distance}, necessary for \autoref{@communicated} and \autoref{@swineherds}, can be found in \autoref{@hierarchical}.

\begin{proposition}[\cite{SauST21kapiII}, \autoref{@hierarchical}]\label{@transforma}
Let $\Fcal$ be a finite collection of graphs.
There exist a function $\newfun{@veneration}: \bN\to\bN$ and an algorithm with the following specifications:\medskip

	\noindent{\tt Find-Wall}$(G,r,k)$\\
	\noindent{\textbf{Input}:} A graph $G$, an odd $r\in\bN_{\geq 3}$, and $k\in\bN$.\\
	\noindent{\textbf{Output}:} One of the following:
	\begin{itemize}
		\item {\bf Case 1:} Either a report that $(G,k)$ is a \no-instance of {\sc $\Fcal$-M-Deletion} (resp. {\sc $\Fcal$-M-Elimination Distance}), or
		\item {\bf Case 2:} a report that $G$ has treewidth at most $\funref{@veneration}(s_\Fcal)\cdot r+k$, or
		\item {\bf Case 3:} an $r$-wall $W$ of $G$.
	\end{itemize}
	Moreover, $\funref{@veneration}(s_\Fcal)=2^{\Ocal(s_\Fcal^{2} \cdot \log s_\Fcal)}$, and the algorithm runs in time $2^{\Ocal_{\ell_\Fcal}(r^2+(k+r) \cdot \log (k+r))}\cdot n$ (resp. $2^{\Ocal_{\ell_\Fcal}(r^2+k^2)}\cdot n$).
\end{proposition}
In  {\bf Case 1}, we can immediately conclude. In  {\bf Case 2}, since the treewidth of $G$ is bounded, we  use a dynamic programming algorithm to solve the corresponding problem.
Namely, we solve {\sc $\Fcal$-M-Deletion} on instances of bounded treewidth using the main result from \cite{BasteST20acom}.

\begin{proposition}[\cite{BasteST20acom}]\label{@calculated}
For every finite collection of graphs $\Fcal$, there exists an algorithm that, given a triple $(G,\tw,k)$ where $G$ is a graph of treewidth at most $\tw$
and $k$ is a non-negative integer, solves {\sc $\Fcal$-M-Deletion} in time $2^{\Ocal_{\ell_\Fcal}(\tw \cdot \log \tw)}\cdot n$.
\end{proposition}

For {\sc $\Fcal$-M-Elimination Distance}, we use \autoref{@achievements} to conclude.
The proof of this (quite technically involved) dynamic programming algorithm is given in  \autoref{@fanaticism}.\smallskip

Therefore, it only remains to deal with {\bf Case 3}. Given an $r$-wall $W$ of $G$, we want to reduce the size of $G$.
To do so, we observe that we can either:
\begin{itemize}
\item {\bf Case 3a:} find a subwall $W_a$ of $W$ and an apex set $A_a$ such that $W_a$ is flat in $G\setminus A_a$ and has a compass of bounded treewidth, or
\item {\bf Case 3b:} find a subwall $W_b$ of $W$ that is very ``well connected'' to an apex set $A_b$ of small size.
\end{itemize}
The above distinction is done using two algorithmic versions of the Flat Wall Theorem consecutively.
The first one comes from \cite[Theorem 7.7]{KawarabayashiTW18anew} and is translated here in the new framework with tilts of \cite{SauST21amor}.
Informally, we say that a graph $H$ is \emph{grasped} by a wall $W$ in a graph $G$ if there is a model of $H$ in $G$ such that the model of every node of $H$ intersects $W$. See \autoref{@envisioning} for the formal definition.

\begin{proposition}[\cite{KawarabayashiTW18anew}]\label{@possession}
There are two functions $\newfun{@classifications}, \newfun{@collaboration}: \bN\to\bN$, such that the images of $\funref{@classifications}$ are odd integers, and an algorithm with the following specifications:\medskip

\noindent{\tt Grasped-or-Flat}$(G,r,t,W)$\\
\noindent{\textbf{Input}:} A graph $G$, an odd $r\in\bN_{\geq 3}$, $t\in\bN_{\geq 1}$, and an $\funref{@classifications}(t)\cdot r$-wall $W$ of $G$.\\
\noindent{\textbf{Output}:} One of the following:
\begin{itemize}
	\item Either a model of a $K_t$-minor in $G$ grasped by $W$, or
	\item a set $A\subseteq V(G)$ of size at most $\funref{@collaboration}(t)$ and a flatness pair $(W',\mathfrak{R}')$ of $G\setminus A$ of heigth $r$ such that $W'$ is a $\tilde{W'}$-tilt of some subwall $\tilde{W'}$ of $W$.
\end{itemize}
	Moreover, $\funref{@classifications}(t)=\Ocal (t^{26})$, $\funref{@collaboration}(t)=\Ocal (t^{24})$, and the algorithm runs in time $\Ocal(t^{24} m+n)$.
\end{proposition}
We would like to mention that the notion of being grasped by a wall is one of the new main arguments yileding the improvement of the complexity for {\sc $\Fcal$-M-Deletion} compared to \cite{SauST21kapiII}.\smallskip

The second one comes from \cite{SauST21kapiII} and adds the condition that $W'$ has a compass of bounded treewidth, at the price of dropping the condition that the model of $K_t$ is grasped by $W$.

\begin{proposition}[\cite{SauST21kapiII}]\label{@unimportant}
There exist a function $\newfun{@corollaries}:\bN\to\bN$ and an algorithm with the following specifications:\medskip

	\noindent{\tt Clique-Or-twFlat}$(G,r,t)$\\
	\noindent{\textbf{Input}:} A graph $G$, an odd $r\in\bN_{\geq 3}$, and $t\in\bN_{\geq 1}$.\\
	\noindent{\textbf{Output}:} One of the following:
	\begin{itemize}
		\item Either a report that $K_t$ is a minor of $G$, or
		\item a tree decomposition of $G$ of width at most $\funref{@corollaries}(t)\cdot r$, or
		\item a set $A\subseteq V(G)$ of size at most $\funref{@collaboration}(t)$ and a regular flatness pair $(W',\mathfrak{R}')$ of $G\setminus A$ of height $r$ whose $\mathfrak{R}'$-compass has treewidth at most $\funref{@corollaries}(t)\cdot r$.
	\end{itemize}
	Moreover, $\funref{@corollaries}(t)=2^{\Ocal(t^2\log t)}$ and this algorithm runs in time $2^{\Ocal_t(r^2)}\cdot n$. The algorithm can be modified to obtain an explicit dependence on $t$ in the running time, namely  $2^{2^{\Ocal(t^2\log t)}\cdot r^3\log r}\cdot n$.
\end{proposition}

{\tt Grasped-or-Flat} is used to find a big enough complete graph ``controlled'' by the input wall, while we need {\tt Clique-or-twFlat} to find a flat wall whose compass has bounded treewidth. Unfortunately, we cannot obtain both conditions simultaneously, and this is why we need  both results. If, after using both algorithms, we obtain a flatness pair $(\tilde{W'},\mathfrak{R}')$ of $G\setminus A_a$ of heigth $r_a$ whose compass has bounded treewidth, then we are in {\bf Case 3a}. In that case, the following result from~\cite{SauST21kapiII} provides an algorithm that, given a flatness pair of big enough height, outputs a homogeneous flatness pair.

\begin{proposition}[\cite{SauST21kapiII}]\label{@disreputable}
There is a function $\newfun{@philistines}:\bN^4\to \bN$, whose images are odd integers, and an algorithm with the following specifications:\medskip

	\noindent{\tt Homogeneous}$(r,\tilde{a},a,\ell,t,G,A,W,\cal R)$\\
	\noindent{\textbf{Input}:} Five integers $r\in\bN_{\geq 3}$, $\tilde{a},a,\ell,t\in \bN$, where $\tilde{a}\leq a$, a graph $G$, a set $A\subseteq V(G)$ of size at most $a$, and a flatness pair $(W,\mathfrak{R})$ of $G\setminus A$ of height $\funref{@philistines}(r,a,\tilde{a},\ell)$ whose $\mathfrak{R}$-compass has treewidth at most $t$.\\
	\noindent{\textbf{Output}:} A flatness pair $(\breve{W},\breve{\mathfrak{R}})$ of $G\setminus A$ of height $r$ that is $\ell$-homogeneous with respect to $\binom{A}{\leq \tilde{a}}$ and is a $W'$-tilt of $(W,\mathfrak{R})$ for some subwall $W'$ of $W$.\\
Moreover, $\funref{@philistines}(r,\tilde{a},a,\ell) = \Ocal(r^{\newfun{@withdrawing}(\tilde{a},a,\ell)})$ where $\funref{@withdrawing}(\tilde{a},a,\ell)=2^{a^{\tilde{a}}\cdot2^{\Ocal((\tilde{a}+\ell)\cdot\log(\tilde{a}+\ell))}}$ and the algorithm runs in time $2^{\Ocal(\funref{@withdrawing}(\tilde{a},a,\ell)\cdot r \log r+t\log t)}\cdot(n+m)$.
\end{proposition}

Then we use the next result, that essentially says that the central vertex $v$ of a big enough homogeneous wall is irrelevant, i.e., $(G,k)$ and $(G\setminus v,k)$ are equivalent instances of the corresponding problem. Here, $\bid_{G\setminus A,W}(X)$ denotes the bidimensionality of a set $X$ in the wall $W$ with apex set $A$.
The combinatorial version of this result is stated in \cite[Lemma 16]{SauST21kapiI} and can be algorithmized using \cite[Theorem 5]{SauST21amor} (\autoref{@expurgated}). Before presenting the next result we give some insight on the Unique Linkage Theorem.

\medskip

A \emph{linkage} $L$ of \emph{order} $k$ in a graph $G$ is the union of a collection of $k$ pairwise-disjoint paths of $G$. The set of pairs
of vertices corresponding to the endpoints of these paths is the \emph{pattern} of $L$.
The Unique Linkage Theorem, proven in \cite{RobertsonSGM22,RobertsonS09XXI} and also \cite{KawarabayashiW10asho}, asserts that there is a function $f_{\sf ul}$
such that if  $L$  is a linkage of pattern $\Pcal$ of order $k$ in a graph $G$ with $V(G) = V(L)$ and $L$ is unique with pattern $\Pcal$, then the treewidth of $G$ is at most $f_{\sf ul}(k)$. The linkage function appears in the general dependency of several results related to the application of the irrelevant vertex technique (see~\cite{SauST21amor,SauST21kapiI,BasteST23hitti,GolovachST23comb,AdlerKKLST17irre,FominGSST23comp,GolovachST23model}).

\begin{proposition}[\cite{SauST21amor,SauST21kapiI}]\label{@civilizing}
Let $\Fcal$ be a finite collection of graphs.
There exist two functions $\newfun{@differences}:\bN^4\to\bN$ and $\newfun{@deliberation}:\bN^2\to\bN$, and an algorithm with the following specifications:\medskip

	\noindent{\tt Find-Irrelevant-Vertex}$(k,a,G,A,W,\cal R)$\\
	\noindent{\textbf{Input}:} Two integers $k, a\in\bN$, a graph $G$, a set $A\subseteq V(G)$,
	and a regular flatness pair $(W,\cal R)$ of $G\setminus A$ of height at least $\funref{@differences}(a,\ell_{\cal F},3,k)$ that is $\funref{@deliberation}(a,\ell_{\cal F})$-homogeneous with respect to $\binom{A}{\leq a}$.
\\
	\noindent{\textbf{Output}:} A vertex $v$ of $G\setminus A$ such that for every set $X\subseteq V(G)$
	with $\bid_{G\setminus A,W}(X)\leq k$ and $|A\setminus X|\leq a$,
	it holds that $G\setminus X\in\exc(\Fcal)$ if and only if $G\setminus (X\setminus v)\in\exc(\Fcal)$.\\
Moreover, $\funref{@differences}(a,\ell_\Fcal,q,k)=\Ocal(k\cdot (f_{\sf ul}(16a+12\ell_{\cal F}))^3 + q )$, where $f_{\sf ul}$ is the function of the Unique Linkage Theorem ({\em \cite[Theorem 1]{KawarabayashiW10asho}}) and $\funref{@deliberation}(a,\ell_{\cal F})= a+\ell_{\cal F} +3$, and this algorithm runs in time ${\cal O}(n+m)$.
\end{proposition}

We can prove that both $k$-apex sets and $k$-elimination sets have small bidimensionality (cf. \autoref{@psychoanalyse} and \autoref{@lucinatory}).
If, for every $k$-apex set $S$, $G\setminus S\in\exc(\Fcal)$ if and only if $G\setminus (S\setminus v)\in\exc(\Fcal)$, then it is straightforward to see that $v$ is irrelevant for {\sc $\Fcal$-M-Deletion}.
It is slightly less trivial to prove that, for each $k$-elimination set~$S$, we can find some superset $X\supseteq S$ of small bidimensionality such that a similar statement holds.
Additional details are available in~\autoref{@unequivocally} of \autoref{@envisioning}.

Therefore we can recursively solve the problems on the instance $(G\setminus v,k)$.\medskip

If no flatness pair whose compass has bounded treewidth was found, then we are in {\bf Case~3b}.
In this case, inspired by \cite{MarxS07obta} and \cite{SauST21kapiII}, we use the following result of~\cite{SauST21kapiI} that basically says that if there is a big enough flat wall $W$ and an apex set $A'$ of $a_\Fcal$ vertices that are all adjacent to many bags of a canonical partition of $W$, then each $k$-apex set or $k$-elimination set intersects~$A'$.

\begin{proposition}[\cite{SauST21kapiI}]\label{@proclamation}
There exist three functions $\newfun{@unaffected}, \newfun{@categories}, \newfun{@provincial}: \bN^{3}\to \bN$,
	such that if
	$G$ is a graph,
	$k\in\bN$,
	$A$ is a subset of $V(G)$,
	$(W,\mathfrak{R})$ is a flatness pair of $G\setminus A$ of height at least $\funref{@unaffected}(a_\Fcal,s_\Fcal,k)$,
	$\tilde\Qcal$ is a $W$-canonical partition of $G\setminus A$,
	$A'$ is a subset of vertices of $A$ that are adjacent, in $G$, to vertices of at least $\funref{@categories}(a_\Fcal,s_\Fcal,k)$ $\funref{@provincial}(a_\Fcal,s_\Fcal,k)$-internal bags of $\tilde\Qcal$, and $|A'|\geq a_\Fcal$,
	then for every set $X\subseteq V(G)$ such that $G\setminus X \in \exc(\Fcal)$ and $\bid_{G\setminus A,W}(X)\leq k$,
	it holds that $X\cap A'\neq\emptyset$.
	Moreover, $\funref{@unaffected}(a,s,k)=\Ocal(2^a \cdot s^{5/2} \cdot k^{5/2})$,
	$\funref{@categories}(a,s,k)=\Ocal(2^a \cdot s^3 \cdot k^3)$, and $\funref{@provincial}(a,s,k)=\Ocal((a^2 +k)\cdot s)$, where $a=a_\Fcal$ and $s= s_\Fcal$.
\end{proposition}

For the {\sc $\Fcal$-M-Deletion} problem, if we find such a set $A'$, then we can branch by guessing which vertex $v\in A'$ belongs to a $k$-apex set and recursively solving $(G\setminus v,k-1)$.
Given that $A'$ has size $a_\Fcal$ and that $k$ decreases after each guess, this step is applied at most $a_\Fcal^k$ times.\smallskip

However, for {\sc $\Fcal$-M-Elimination Distance}, $k$ does not decrease, given that the size of a $k$-elimination set may not depend on $k$. Thus, this step may be done $a_\Fcal^n$ times, which does not give an \FPT-algorithm.
To circumvent this problem, we propose two alternatives:\smallskip

{\bf Option 1:} The first alternative is to only use {\bf Case 3a}. This is possible given that $(K_{s_\Fcal+k},k)$ is a \no-instance of both problems.
Thus, when using the algorithms {\tt Grasped-or-Flat} and {\tt Clique-or-twFlat}, we force the outcome to be an apex set $A$ and a flatness pair of $G\setminus A$. However, the bound on the size of $A$ now depends on $k$, and thus, so does the variable $a$ in the input of the algorithm {\tt Homogeneous}. This explains the triple-exponential parametric dependence on $k$ in \autoref{@communicated}. Interestingly, a precise analysis of the time complexity, which can be found in \autoref{@oeconomicus}, shows that if $a_\Fcal = 1$, i.e.,~when $\Fcal$ contains an apex graph, the parametric dependence is only double-exponential on $k$ (cf.~\autoref{@communicated}).

\smallskip

{\bf Option 2:} The second alternative is to restrict ourselves to the case where $a_\Fcal=1$.
Thus, in {\bf Case 3b}, we find a vertex $v$ that belongs to every $k$-elimination set. There is no need to branch, and this step is done at most $n$ times. However, the fact that the time complexity of this step is quadratic in $n$ explains the cubic complexity of the algorithm in \autoref{@swineherds}.\smallskip

It remains to show that if no flatness pair whose compass has bounded treewidth was found, then we can find a flatness pair and a set $A'$ satisfying the conditions of \autoref{@proclamation}.
To do so, using flow techniques, we find the set $A$ of vertices with sufficiently many internally-disjoint paths to $W$, independently from one another.
If this set is too large, we can safely declare a \no-instance.
Otherwise,  we extend the canonical partition of $W$ and just check whether $a_\Fcal$ vertices of $A$ are adjacent to many vertices of this new canonical partition. If this happens, then we can safely use \autoref{@proclamation}.
The second main improvement with respect to the algorithm in \cite{SauST21kapiII} is the new argument that the extension of the canonical partition of $W$ can be done in a totally arbitrary manner.
The quadratic complexity of this step stems from the search for internally-disjoint paths for every vertex of the input graph.

\section{More definitions}

In this section, we give some more definitions.
Namely, in \autoref{@choanalytic}, we define elimination trees, that is an alternative way to define elimination distance of a graph to a graph class.
Next, in~\autoref{@complacency}, we define (nice) tree decompositions and we present some preliminary results.
Finally, in \autoref{@coordinated} we define boundaried graphs, an equivalence relation between them, and the notion of representatives.

\subsection{\texorpdfstring{$\Fcal$}{F}-elimination trees}
\label{@choanalytic}

We start this subsection by defining some notions on (rooted) trees.
\subparagraph{Trees and rooted trees.}
Let $T$ be a tree and $u,v$ be two nodes of $T$. We denote by $uTv$  the path in $T$ $u$ and $v$. A \emph{rooted tree} is a pair $(T,r)$ where $T$ is a tree and $r$ is a node of $T$ called \emph{root} of $(T,r)$.
Let $(T,r)$ be a rooted tree and let $u$ be a node of $T$.
We define the {\em descendants} of $u$ in $(T,r)$ by
$\desc_{T,r}(u)=\{x\in V(T)\mid u\in V(xTr)\}$
and the {\em ancestors} of $u$ in $(T,r)$ by
$\anc_{T,r}(u)=V(rTu)\setminus\{u\}$.
We define the \emph{leaves} in $(T,r)$ by $\leaf(T,r)=\{u\in V(T)\mid \desc_{T,r}(u)=\{u\}\}$
and the \emph{internal nodes} in $(T,r)$ by $\Int(T,r)=V(T)\setminus\leaf(T,r)$.
If $u\neq r$ then we denote by $\Par_{T,r}(u)$ the unique
node in $\anc_{T,r}(u)\cap N_{T}(u)$.
We also agree that $\Par_{T,r}(r)={\sf void}$.
We denote by $\ch_{T,r}(u)=\desc_{T,r}(u)\cap N_{G}(u)$ the set of the children of $u$ (certainly $\ch_{T,r}(u)=\emptyset$ if $u$ is a leaf of $T$).
Given $K\subseteq V(T)$, the \emph{least common ancestor} of $K$ in $(T,r)$ is the node $u$ such that $K\subseteq\desc_{T,r}(u)$ and there is no child $v$ of $u$ such that $K\subseteq\desc_{T,r}(v)$.

The \emph{height function} $\height_{T,r}:V(T)\to\bN$ maps $v\in \leaf(T,r)$ to 0 and $v\in \Int(T,r)$ to $1+\max\{\height_{T,r}(x)\mid x\in \ch_{T,r}(v)\}$.
The \emph{height} of $(T,r)$ is $\height_{T,r}(r)$.
Note that the height function is decreased by one here compared to the usual definition of the height.

We use $(T_u^r,u)$ to denote the rooted tree where $T_{u}^r=T[\desc_{T,r}(u)]$
and we call $(T_u^r,u)$ \emph{subtree of $(T,r)$ rooted at $u$}.
To simplify notation and when the root $r$ is clear from the context, we use $T_u$ instead of $T_u^r$.

A \emph{rooted forest} is a pair $(F,R)$ where $F$ is a forest and $R$ is a set of roots such that each tree in $F$ has exactly one root in $R$. All notations above naturally extend to forests.

\subparagraph{Elimination trees.}
We now define {\sl elimination trees}, that can be used to define alternatively graphs of bounded elimination distance.
Let $\Fcal$ be a non-empty finite collection of non-empty graphs.
An \emph{$\Fcal$-elimination tree} of a connected graph $G$ is a triple $(T,\chi,r)$ where $(T,r)$ is a rooted tree and $\chi:V(T)\to 2^{V(G)}$ such that:
\begin{itemize}
	\item for each $t\in\Int(T,r)$, $|\chi(t)|=1$,
	\item $(\chi(t))_{t\in V(T)}$ is a partition of $V(G)$,
	\item for each $uv\in E(G)$, if $u\in\chi(t_1)$ and $v\in\chi(t_2)$, then $t_1\in\anc_{T,r}(t_2)\cup\desc_{T,r}(t_2)$,
	\item for each $t\in\leaf(T,r)$, $G[\chi(t)]\in\exc(\Fcal)$, and
	\item for each $t\in V(T)$, $G[\chi(T_t)]$ is connected.
\end{itemize}

The \emph{height} of $(T,\chi,r)$ is the height of $(T,r)$.
It is straightforward to see that the minimum height of an $\Fcal$-elimination tree of a connected graph $G$ is $\ed_{\exc(\Fcal)}(G)$.
Note that $\chi(\Int(T,r))$ is a $k$-elimination set of $G$ for $\exc(\Fcal)$ and that,
if $\Fcal$ is trivial, then, for each $t\in\leaf(T,r)$, $\chi(t)=\emptyset$.
Observe also that for every $u\in\Int(T,r)$ with at least two children $x$ and $y$, any path between $\chi(T_x)$ and $\chi(T_y)$ intersects $\chi(uTr)$.

An \emph{$\Fcal$-elimination forest} of a graph $G$ is a triple $(F,\chi,R)$, such that, if $\cc(G)=\{G_1,...,G_l\}$,
then $F$ is the disjoint union of the trees $T_1,\ldots, T_l$
and $R=\{r_1,...,r_l\}$
where $(T_i,\chi|_{V(T_i)},r_i)$ is an $\Fcal$-elimination tree of $G_i$ for $i\in[l]$.\medskip

\sugar{
The following simple lemma is based on the fact that, given an $\Fcal$-elimination tree $(T,\chi,r)$ of a graph $G$,
for every non-leaf node $u$ of $T$,
$\chi(uTr)$ separates the vertex sets $\chi(T_x)$ and $\chi(T_y)$,
where $x$ and $y$ are distinct children of $v$ in $(T,r)$.
}

\begin{lemma}\label{@unreflectingly}
Let $\Fcal$ be a finite collection of graphs.
Let $G$ be a graph and let $H$ be a connected subgraph of $G$.
Let $(F,\chi,R)$ be an $\Fcal$-elimination forest of $G$.
Then the least common ancestor of $\chi^{-1}(V(H))$ exists and belongs to $\chi^{-1}(V(H))$.
\end{lemma}

\begin{proof}
Let $K:=\chi^{-1}(V(H))$.
Since $H$ is connected, $K$ is a subset of a tree in $F$, and therefore the least common ancestor of $K$ is defined.
Let $u$ be the least common ancestor of $K$. Let $r\in R$ be the root of the tree containing $u$.
Let $x,y\in V(H)$ such that the least common ancestor of $\chi^{-1}(x)$ and $\chi^{-1}(y)$ is $u$.
Since $H$ is connected, there is a path $P$ in $H$ between $x$ and $y$.
By the third property of elimination trees, $\{u\}\cup\anc_{F,R}(u)$ intersects $\chi^{-1}(V(P))$, and so $\{u\}\cup\anc_{F,R}(u)$ intersects $K$.
Since $u$ is the least common ancestor of $K$, $u\in K$.
\end{proof}

\sugar{We now present a lemma to justify that the graphs with bounded elimination distance to $\exc(\Fcal)$ are minor-free.
Intuitively, the proof of this lemma is based on the fact that, due to~\autoref{@unreflectingly},
the size of the largest clique minor that can ``fit'' inside an elimination tree
is equal to the height of the elimination tree.
}

\begin{lemma}\label{@inoperative}
Let $\Fcal$ be a finite collection of graphs.
Let $G$ be a graph and $k\in\bN$ such that $\ed_{\exc{\Fcal}}(G)\leq k$.
Then $K_{s_\Fcal+k}$ is not a minor of $G$.
\end{lemma}

\begin{proof}
Let $(F,\chi,R)$ be an $\Fcal$-elimination forest of $G$ of height at most $k$.
Suppose towards a contradiction that there is a model of $K_{s_\Fcal+k}$ in $G$.
Let $x_1,...,x_{s_\Fcal+k}$ be the vertices of $K_{s_\Fcal+k}$ and for every $i\in[s_\Fcal+k]$, let $V_i$ be the model of $x_i$ in $G$.
Let $G'$ be the graph obtained by contracting, for each $i\in[s_\Fcal+k]$, the edges in each $V_i$.
Let $v_i, i\in[s_\Fcal+k]$ be the resulting vertices after the contraction of each $V_i$.
Thus, the graph $G'[\{v_1,...,v_{s_\Fcal+k}\}]$ is isomorphic to $K_{s_\Fcal+k}$.

Let $(F',\chi',R)$ be obtained from $(F,\chi,R)$ as follows.
For every $i\in[s_\Fcal+k]$, let $u_i$ be the least common ancestor of $K_i:=\chi^{-1}(V_i)$ in $(F,R)$.
Due to \autoref{@unreflectingly}, $u_i\in K_i$.
The forest $F'$ is obtained after removing each node $v\in (K_i\setminus u_i) \cap (\Int(F,R)\setminus R)$ from $V(F)$ and adding an edge between $\Par_{F,R}(v)$ and each node in $\ch_{F,R}(v)$.
The function $\chi'$ is defined as $\chi'(v) := \chi(v)$ if $v\in \Int(F',R)$ and
$\chi'(v) := \chi(v)\setminus (V_i\setminus u_i)$ if $v\in \leaf(F',R)$.
In the latter case, if $G'[\chi'(v)]$ is not connected, then we update $F'$ by replacing $v$ by $|{\sf cc}(G'[\chi'(v)])|$ nodes, each one associated with a connected component of $G'[\chi'(v)]$.
Observe that $(F',\chi',R)$ is an $\Fcal$-elimination forest of $G'$ of height at most $k$ and that we can assume that
for every $i\in[s_\Fcal+k]$, $v_i\in \chi'(u_i)$.

Since the vertices in $\{v_1,...,v_{s_\Fcal+k}\}$ are pairwise connected by an edge in $G'$, the third property of elimination trees implies that there is $u\in\leaf(F',R)$ and $r\in R$ such that $\{v_1,...,v_{s_\Fcal+k}\}\subseteq\chi'(uF'r)$.
Let $u':=\Par_{F',R}(u)$.
For every $t\in u'F'r$, $t\in\Int(F',R)$, so $|\chi'(t)|=1$, and therefore, $|\chi'(u'Fr)|\leq k$.
Thus, $|\chi'(u)\cap V(H)|\geq s_\Fcal$.
Therefore, $K_{s_\Fcal}$ is a minor of $G'[\chi'(u)]$.
This contradicts the fact that $G'[\chi'(u)]\in\exc(\Fcal)$, so $K_{s_\Fcal+k}$ is not a minor of $G$.
\end{proof}

\subsection{Tree decompositions}
\label{@complacency}

The notions of tree decompositions and treewidth are used throughout the paper.
\subparagraph{Treewidth.}
A \emph{tree decomposition} of a graph~$G$
is a pair~$({\sf T},\beta)$ where ${\sf T}$ is a tree and $\beta: V({\sf T})\to 2^{V(G)}$
such that
\begin{itemize}
	\item $\bigcup_{t \in V({\sf T})} \beta(t) = V(G)$,
	\item for every $e\in E(G)$, there is a $t\in V({\sf T})$ such that $\beta(t)$ contains both endpoints of~$e$, and
	\item for every~$v \in V(G)$, the subgraph of~${\sf T}$ induced by $\{t \in V({\sf T})\mid {v \in \beta(t)}\}$ is connected.
\end{itemize}

The {\em width} of $({\sf T},\beta)$ is equal to $\max\big\{\left|\beta(t)\right|-1 \bigmid t\in V({\sf T})\big\}$
and the {\em treewidth} of $G$, denoted by $\tw(G)$, is the minimum width over all tree decompositions of $G$.
A \emph{rooted tree decomposition} is a triple $({\sf T},\beta,{\sf r})$ where $({\sf T},\beta)$ is a tree decomposition and $({\sf T},{\sf r})$ is a rooted tree.
For every $q\in V({\sf T})$, we set $G_q^{({\sf T},\beta,{\sf r})} = G[\beta({\sf T}_q)]$.
We may write $G_q$ instead of $G_q^{({\sf T},\beta,{\sf r})}$ when there is no ambiguity about ${({\sf T},\beta,{\sf r})}$.\medskip

To compute a tree decomposition of a graph of bounded treewidth, we use the recent single-exponential $2$-approximation algorithm for treewidth of Korhonen \cite{Korhonen21asin}.

\begin{proposition}[\cite{Korhonen21asin}]\label{@naturalism}
There is an algorithm that, given an graph $G$ and an integer $k$,
outputs either a report that $\tw(G)>k$, or a tree decomposition of $G$ of width at most $2k+1$ with $\Ocal(n)$ nodes.
Moreover, this algorithm runs in time $2^{\Ocal(k)} \cdot n$.
\end{proposition}

\sugar{The relation between the treedepth and the treewidth of a graph proved by Bodlaender,  Gilbert,  Kloks, and  Hafsteinsson in \cite{BodlaenderGKH91appr} will be used in~\autoref{@fanaticism} in order to obtain an \XP-algorithm for {\sc $\Fcal$-M-Elimination Distance} parameterized by treewidth.}

\begin{proposition}[\cite{BodlaenderGKH91appr}]\label{@theosophical}
Let $G$ be a graph with $n$ vertices. Then $\tw(G)\leq\td(G)\leq\tw(G)\cdot \log n$.
\end{proposition}\medskip

\sugar{The following result has been proved by Adler, Dorn, Fomin, Sau, and Thilikos in~\cite{AdlerDFST11fast}.
We use it in \autoref{@oeconomicus} and, in particular, in the algorithm \noindent{\tt Find-Wall-Ed} of~\autoref{@transforma}, to detect a wall in a graph of bounded treewidth.}
\begin{proposition}[\cite{AdlerDFST11fast}]\label{@camouflage}
	There is an algorithm that, given a graph $G$ on $m$ edges,
	a graph $H$ on $h$ edges without isolated vertices,
	and a tree decomposition of $G$ of width at most $k$, outputs,
	if it exists, a minor of $G$ isomorphic to $H$.
	Moreover, this algorithm runs in time $2^{\Ocal(k\log k)}\cdot h^{\Ocal(k)}\cdot 2^{\Ocal(h)}\cdot m$.
\end{proposition}

We also show that given a graph $G$ and an integer $k$,
the removal of a $k$-elimination set from $G$ does not decrease the treewidth of $G$ more than $k$.
\begin{lemma}\label{@barbarians}
Let $\Fcal$ be a finite collection of graphs.
Let $c,k$ be two integers and let $G$ be a graph such that $\tw(G)\geq c$.
Let $S$ be a $k$-elimination set of $G$ for $\exc({\Fcal})$.
Then $\tw(G\setminus S)\geq c-k$.
\end{lemma}

\begin{proof}
Suppose first that $G$ is connected.
Let $(T,\chi,r)$ be an $\Fcal$-elimination tree of $G$ with $\chi(\Int(T,r))=S$.
Let $v_1,...,v_l$ be the leaves of $(T,r)$, whose label is given by a depth-first search order starting from $r$.
Let $C_i:=G[\chi(v_i)]$ for $i\in[l]$, and note that  $\tw(G\setminus S)=\max_{i\in[l]}\tw(C_i)$. Suppose for contradiction that  $\tw(G\setminus S) < c-k$, and we proceed to show that $\tw(G) < c$, contradicting our hypothesis.  Let $({\sf T}_i,\beta_i)$ be an optimal tree decomposition of $C_i$ of width $w_i$ and let $P_i$ be the path from the parent of $v_i$ to $r$ in $T$, for $i\in[l]$.
Let $w:=\max_{i\in[l]}w_i$, so we have that $w \leq c-k-1$.
We construct a tree decomposition $(\Tcal,\beta)$ of $G$, starting from the tree decompositions $({\sf T}_i,\beta_i)$, as follows.
Create a path with nodes $x_1,...,x_l$ such that for $i\in[l]$, $\beta(x_i)=V(P_i)$.
Then for $i\in[l]$, add an edge between $x_i$ and a node of ${\sf T}_i$. For each $x\in V({\sf T}_i)$, we set $\beta(x):=\beta_i(x)\cup V(P_i)$.
Since the height of $(T,r)$ is at most $k$, $P_i$ has size at most $k$ for $i\in[l]$, so $(\Tcal,\beta)$ has width at most $w+k \leq c-1$, a contradiction.

If $G$ is not connected, we can apply the above proof to each of its connected components.
\end{proof}

To describe our dynamic programming algorithm presented in \autoref{@fanaticism}, we need a particular type of tree decompositions, namely nice tree decompositions.

\subparagraph{Nice tree decompositions.} A \emph{nice tree decomposition} of a graph $G$ is a rooted tree decomposition $({\sf T},\beta,{\sf r})$ such that:
\begin{itemize}
	\item every node has either zero, one or two children,
	\item if $x$ is a leaf of ${\sf T}$, then $\beta(x)$ is a singleton ($x$ is a \emph{leaf node}),
	\item if $x$ is a node of ${\sf T}$ with a single child $y$, then $|\beta(x)\setminus\beta(y)|=1$ ($x$ is an \emph{introduce node}) or $|\beta(y)\setminus\beta(x)|=1$ ($x$ is a \emph{forget node}), and
	\item if $x$ is a node with two children $x_1$ and $x_2$, then $\beta(x)=\beta(x_1)=\beta(x_2)$ ($x$ is a \emph{join node}).
\end{itemize}

To find a nice tree decomposition from a given a tree decomposition, we use the following well-known result proved, for instance, in~\cite{AlthausZ21opti}.

\begin{proposition}[\cite{AlthausZ21opti}]\label{@estclusire}
Given a graph $G$ with $n$ vertices and a tree decomposition $({\sf T},\beta)$ of $G$ of width $w$, there is an algorithm that computes a nice tree decomposition of $G$ of width $w$ with at most $\Ocal(w\cdot n)$ nodes in time $\Ocal(w^2\cdot (n+|V({\sf T})|))$.
\end{proposition}

\subsection{Boundaried graphs and representatives}\label{@coordinated}

Boundaried graphs are extensively used in \autoref{@fanaticism} and \autoref{@participant}.
We present here some useful definitions and results.

\subparagraph{Boundaried graphs.} Let $t\in \bN$. A \emph{$t$-boundaried graph} is a triple ${\bf G}=(G,B,\rho)$ where $G$ is a graph, $B\subseteq V(G)$, $|B|=t$, and $\rho:B\to[t]$ is a bijection.
We say that two $t$-boundaried graphs ${\bf G}_1=(G_1,B_1,\rho_1)$ and ${\bf G}_2=(G_2,B_2,\rho_2)$ are \emph{isomorphic} if there is an isomorphism from $G_1$ to $G_2$ that extends the bijection $\rho_2^{-1}\circ\rho_1$.
The triple $(G,B,\rho)$ is a \emph{boundaried graph} if it is a $t$-boundaried graph for some $t\in\bN$.
We denote by $\Bcal^t$ the set of all (pairwise non-isomorphic) $t$-boundaried graphs.
We also set $\Bcal=\bigcup_{t\in\bN}\Bcal^t$.

\subparagraph{Minors of boundaried graphs.} We say that a $t$-boundaried graph $G_1=(G_1,B_1,\rho_1)$ is a minor
of a $t$-boundaried graph $G_2=(G_2,B_2,\rho_2)$, denoted by $G_1\preceq G_2$, if there is a sequence of removals
of non-boundary vertices, edge removals, and edge contractions in $G_2$, not allowing contractions
of edges with both endpoints in $B_2$, that transforms $G_2$ to a boundaried graph that is isomorphic
to $G_1$ (during edge contractions, boundary vertices prevail). Note that this extends the usual
definition of minors in graphs without boundary.

\subparagraph{Equivalent boundaried graphs and representatives.} We say that two boundaried graphs ${\bf G}_1=(G_1,B_1,\rho_1)$ and ${\bf G}_2=(G_2,B_2,\rho_2)$ are \emph{compatible} if $\rho_2^{-1}\circ\rho_1$ is an isomorphism from $G_1[B_1]$ to $G_2[B_2]$.
Given two compatible boundaried graphs ${\bf G}_1=(G_1,B_1,\rho_1)$ and ${\bf G}_2=(G_2,B_2,\rho_2)$, we define ${\bf G}_1\oplus {\bf G}_2$ as the graph obtained if we take the disjoint union of $G_1$ and $G_2$ and, for every $i\in[|B_1|]$, we identify vertices $\rho_1^{-1}(i)$ and $\rho_2^{-1}(i)$.
We also define ${\bf G}_1\bigoplus {\bf G}_2$
as the {\sl boundaried} graph $({\bf G}_1\oplus {\bf G}_2,B_1,\rho_1)$.
Given $h\in\bN$, we say that two boundaried graphs ${\bf G}_1$ and ${\bf G}_2$ are $h$-\emph{equivalent}, denoted by ${\bf G}_1\equiv_h {\bf G}_2$, if they are compatible and, for every graph $H$ with detail at most $h$ and every boundaried graph ${\bf F}$
compatible with ${\bf G}_1$ (hence, with ${\bf G}_2$ as well), it holds that
\[ H\preceq {\bf F}\oplus {\bf G}_1 \Longleftrightarrow H\preceq {\bf F}\oplus {\bf G}_2.\]
Note that $\equiv_h$ is an equivalence relation on $\cal B$. A minimum-sized (in terms of number of vertices) element of an equivalent class of $\equiv_h$ is called \emph{representative} of $\equiv_h$. For $t\in\bN$, a \emph{set of $t$-representatives} for $\equiv_h$, denoted by $\Rcal_h^t$, is a collection containing a minimum-sized representative for each equivalence class of $\equiv_h$ restricted to ${\cal B}^t$.\medskip

The following results were proved by Baste, Sau,  and Thilikos~\cite{BasteST20acom} and give a bound on the size of a representative and on  the
number of representatives for this equivalence relation, respectively.

\begin{proposition}[\cite{BasteST20acom}]\label{@iinelstaai}
For every $t\in\bN$, $q,h\in\bN_{\geq 1}$, and $\textbf{G}=(G,B,\rho)\in\Rcal_h^t$, if $G$ is does not contain $K_q$ as a minor, then $|V(G)|=\Ocal_{q,h}(t)$.
\end{proposition}

\begin{proposition}[\cite{BasteST20acom}]\label{@encounters}
For every $t\in\bN_{\geq 1}$, $|\Rcal_h^t|=2^{\Ocal_h(t\log t)}$.
\end{proposition}

Moreover, given a boundaried graph of bounded size, the following lemma gives an algorithm to find its representative.
\sugar{While this is might be considered folklore, we include here its proof for the sake of completeness.}

\begin{lemma}\label{@objectives}
Given a finite collection of graphs $\Fcal$, $h,t,k\in\bN$, the set $\Rcal$ of representatives in $\Rcal_h^t$ whose underlying graphs are $\Fcal$-minor-free, and a $t$-boundaried graph ${\bf G}$ with $k$ vertices whose underlying graph is $\Fcal$-minor-free, there is an algorithm that outputs the representative of ${\bf G}$ in $\Rcal$ in time $2^{\Ocal_{{\ell_\Fcal},h}(t\log t+\log(k+t))}$.
\end{lemma}

\begin{proof}
Let $\Hcal$ be the set of graphs with detail at most $h$.
For a $t$-boundaried graph ${\bf G}$ of size $k$ whose underlying graph is $\Fcal$-minor-free, we define the matrix $M_{\bf G}$, whose rows are the representatives in $\Rcal$ and whose columns are the graphs of $\Hcal$, such that for ${\bf R}\in \Rcal$ and $H\in\Hcal$, we have $M_{\bf G}({\bf R},H)=1$ if ${\bf G}$ and ${\bf R}$ are compatible and $H\preceq{\bf G}\oplus {\bf R}$, and $M_{\bf G}({\bf R},H)=0$ otherwise.
Observe that ${\bf R}\in\Rcal$ is the representative of ${\bf G}$ if and only if $M_{\bf R}=M_{\bf G}$.
According to \autoref{@encounters}, $M_{\bf G}$ has size $2^{\Ocal_h(t\log t)}\cdot \Ocal_h(1)=2^{\Ocal_h(t\log t)}$.

For all ${\bf R}\in\Rcal$, we compute $M_{\bf R}$.
Every representative in $\Rcal$ has size at most $\Ocal_{{s_\Fcal},h}(t)$ by \autoref{@iinelstaai}, so when two representatives ${\bf R}$ and ${\bf R}'$ are compatible, ${\bf R}\oplus {\bf R}'$ has size $\Ocal_{{\ell_\Fcal},h}(t)$ as well.
From \cite{KawarabayashiKR12thed}, we know that checking if a graph $H\in\Hcal$ is a minor of ${\bf R}\oplus {\bf R}'$ can be done in time $\Ocal_{{\ell_\Fcal},h}(t^2)$.
Therefore, we can compute $M_{\bf R}$ in time $2^{\Ocal_{h,{s_\Fcal}}(t\log t)}$.

Let ${\bf G}$ be a $t$-boundaried graph of size $k$ whose underlying graph is $\Fcal$-minor-free.
For ${\bf R}\in\Rcal$ compatible with ${\bf G}$, ${\bf G}\oplus {\bf R}$ has size $\Ocal_{{\ell_\Fcal},h}(k+t)$, so checking if $H\in\Hcal$ is a minor of ${\bf G}\oplus {\bf R}$ can be done in time $\Ocal_{{\ell_\Fcal},h}((k+t)^2)$.
Thus we can compute $M_{\bf G}$ in time $2^{\Ocal_h(t\log t)}\cdot\Ocal_{{\ell_\Fcal},h}((k+t)^2)=2^{\Ocal_{{\ell_\Fcal},h}(t\log t+\log(k+t))}$.

Finally, we just need to find ${\bf R}\in\Rcal$ such that $M_{\bf R}=M_{\bf G}$, which can be done in time $2^{\Ocal_h(t\log t)}$.
Thus, we can find the representative of ${\bf G}$ in time $2^{\Ocal_{{\ell_\Fcal},h}(t\log t+\log(k+t))}$.
\end{proof}

\section{Even more definitions: Flat walls}\label{@envisioning}

In this section we deal with flat walls using the framework of \cite{SauST21amor}.
More precisely, in \autoref{@postulated}, we introduce walls and several notions concerning them.
In \autoref{@commiseratio} we provide the definitions of a rendition and a painting.
Using the above notions, in \autoref{@tugendlehre} we define flat walls and provide some results about them, including the Flat Wall Theorem (namely, the version proved in \cite{KawarabayashiTW18anew}) and its algorithmic version restated in the ``more accurate'' framework of \cite{SauST21amor}.
In \autoref{@commodation} and~\autoref{@unequivocally}, we define canonical partitions and the notion of bidimensionality and give some combinatorial results that allow us to use branching in the algorithm in \autoref{@conditioned}.
Finally in \autoref{@indigenous} we present homogeneous walls and give some results to find an irrelevant vertex.
We note that most of the definitions of this section can also be found in \cite{SauST21amor,SauST21kapiI,SauST21kapiII} with more details and illustrations.

\subsection{Walls and subwalls}\label{@postulated}
We start with some basic definitions about walls.

\subparagraph{Walls.}
Let $k,r\in\bN$. The
\emph{$(k\times r)$-grid} is the
graph whose vertex set is $[k]\times[r]$ and two vertices $(i,j)$ and $(i',j')$ are adjacent if and only if $|i-i'|+|j-j'|=1$.
An \emph{elementary $r$-wall}, for some odd integer $r\geq 3$, is the graph obtained from a
$(2 r\times r)$-grid
with vertices $(x,y)
	\in[2r]\times[r]$,
after the removal of the
``vertical'' edges $\{(x,y),(x,y+1)\}$ for odd $x+y$, and then the removal of
all vertices of degree one.
Notice that, as $r\geq 3$, an elementary $r$-wall is a planar graph
that has a unique (up to topological isomorphism) embedding in the plane $\bR^{2}$
such that all its finite faces are incident to exactly six
edges.
The {\em perimeter} of an elementary $r$-wall is the cycle bounding its infinite face,
while the cycles bounding its finite faces are called {\em bricks}.
Also, the vertices
in the perimeter of an elementary $r$-wall that have degree two are called {\em pegs},
while the vertices $(1,1), (2,r), (2r-1,1), (2r,r)$ are called {\em corners} (notice that the corners are also pegs).

\begin{figure}[ht]
	\begin{center}
		\includegraphics[width=14cm]{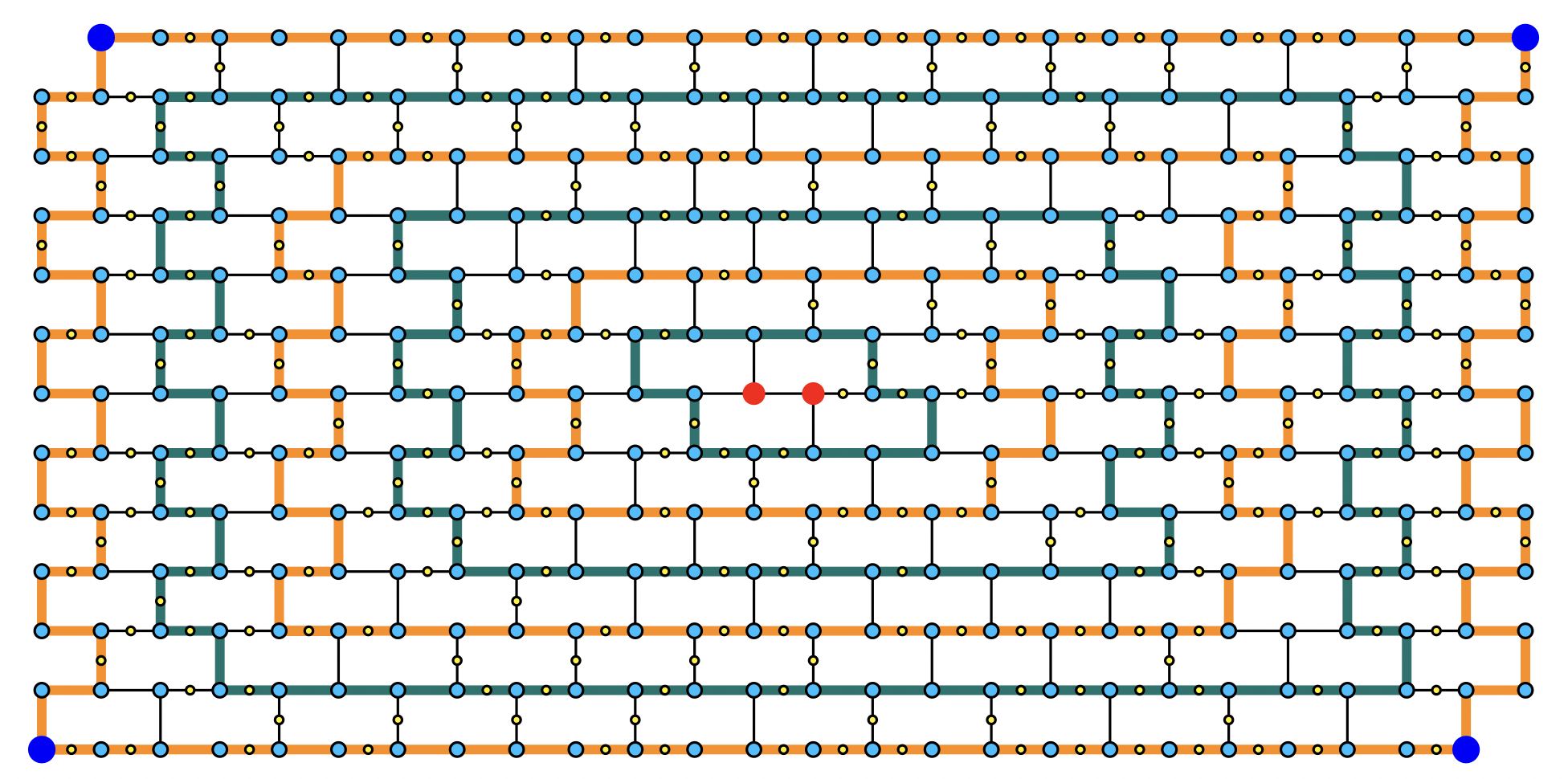} % TS
                                % changed from 11cm
	\end{center}
	\caption{A $13$-wall and its six layers, depicted in alternating orange and green. The central vertices of the wall are depicted in red and the corners are depicted in blue.}
	\label{@manuscript}
\end{figure}

An {\em $r$-wall} is any graph $W$ obtained from an elementary $r$-wall $\bar{W}$
after subdividing edges (see \autoref{@manuscript}). A graph $W$ is a {\em wall} if it is an $r$-wall for some odd $r\geq 3$
and we refer to $r$ as the {\em height} of $W$. Given a graph $G$,
a {\em wall of} $G$ is a subgraph of $G$ that is a wall.
We insist that, for every $r$-wall, the number $r$ is always odd.

We call the vertices of degree three of a wall $W$ {\em 3-branch vertices}.
A cycle of $W$ is a {\em brick} (resp. the {\em perimeter}) of $W$
if its 3-branch vertices are the vertices of a brick (resp. the perimeter) of $\bar{W}$.
We denote by ${\cal C}(W)$ the set of all cycles of $W$.
We use $D(W)$ in order to denote the perimeter of the wall $W$.
A brick of $W$ is {\em internal} if it is disjoint from $D(W)$.

\subparagraph{Subwalls.}
Given an elementary $r$-wall $\bar{W}$, some odd $i\in \{1,3,\ldots,2r-1\}$, and $i'=(i+1)/2$,
the {\em $i'$-th vertical path} of $\bar{W}$ is the one whose
vertices, in order of appearance, are $(i,1),(i,2),(i+1,2),(i+1,3),
	(i,3),(i,4),(i+1,4),(i+1,5),
	(i,5),\ldots,(i,r-2),(i,r-1),(i+1,r-1),(i+1,r)$.
Also, given some $j\in[2,r-1]$ the {\em $j$-th horizontal path} of $\bar{W}$
is the one whose
vertices, in order of appearance, are $(1,j),(2,j),\ldots,(2r,j)$.

A \emph{vertical} (resp. \emph{horizontal}) path of an $r$-wall $W$ is one
that is a subdivision of a vertical (resp. horizontal) path of $\bar{W}$.
Notice that the perimeter of an $r$-wall $W$
is uniquely defined regardless of the choice of the elementary $r$-wall $\bar{W}$.
A {\em subwall} of $W$ is any subgraph $W'$ of $W$
that is an $r'$-wall, with $r' \leq r$, and such the vertical (resp. horizontal) paths of $W'$ are subpaths of the
	{vertical} (resp. {horizontal}) paths of $W$.
	
\subparagraph{Layers.}
The {\em layers} of an $r$-wall $W$ are recursively defined as follows.
The first layer of $W$ is its perimeter.
For $i=2,\ldots,(r-1)/2$, the $i$-th layer of $W$ is the $(i-1)$-th layer of the subwall $W'$
obtained from $W$ after removing from $W$ its perimeter and
removing recursively all occurring vertices of degree one.
We refer to the $(r-1)/2$-th layer as the {\em inner layer} of $W$.
The {\em central vertices} of an $r$-wall are its two 3-branch vertices that do not belong to any of its layers and that are connected by a path of $W$ that does not intersect any layer. See \autoref{@manuscript} for an illustration of the notions defined above.

\subparagraph{Central walls.}
Given an $r$-wall $W$ and an odd $q\in\bN_{\geq 3}$ where $q\leq r$,
we define the {\em central $q$-subwall} of $W$, denoted by $W^{(q)}$,
to be the $q$-wall obtained from $W$ after removing
its first $(r-q)/2$ layers and all occurring vertices of degree one.

\subparagraph{Tilts.}
The {\em interior} of a wall $W$ is the graph obtained
from $W$ if we remove from it all edges of $D(W)$ and all vertices
of $D(W)$ that have degree two in $W$.
Given two walls $W$ and $\tilde{W}$ of a graph $G$,
we say that $\tilde{W}$ is a {\em tilt} of $W$ if $\tilde{W}$ and $W$ have identical interiors.

\subparagraph{Minor models grasped by walls.}
Let $G$ be a graph and $W$ be an $r$-wall in $G$. Let $P_1,...,P_r$ be the horizontal paths and $Q_1,...,Q_r$ be the vertical paths of $W$. Let $t\geq 1$ be an integer and $v_1,...,v_t$ be the vertices in $K_t$. A model of a $K_t$-minor in $G$ is \emph{grasped} by $W$ if, for every model $X_i$ of $v_i$, there exist distinct indices $i_1,...,i_t\in [r]$ and distinct indices $j_1,...,j_t\in [r]$ such that $V(P_{i_l})\cap V(Q_{j_l})\subseteq X_l$ for all $l\in [t]$.\medskip

We present the following result
of Kawarabayashi and Kobayashi \cite{KawarabayashiK20line},
which provides a linear relation between the treewidth and the height of a largest wall in a minor-free graph.

\begin{proposition}[\cite{KawarabayashiK20line}]\label{@overwhelmingly}
	There is a function $\newfun{@relentlessly}:\bN\to \bN$ such that, for every $t,r\in \bN$
	and every graph $G$ that does not contain $K_{t}$ as a minor, if $\tw(G)\geq \funref{@relentlessly}(t)\cdot r$, then $G$ contains an $r$-wall as a subgraph.
	In particular, one may choose $\funref{@relentlessly}(t)=2^{\Ocal(t^{2} \cdot \log t)}$.
\end{proposition}

\subsection{Paintings and renditions}\label{@commiseratio}

In this subsection we present the notions of renditions and paintings, originating in the work of Robertson and Seymour \cite{RobertsonS95XIII}.
The definitions presented here were introduced in \cite{KawarabayashiTW18anew} (see also \cite{BasteST20acom,SauST21amor}).
\subparagraph{Paintings.}
A {\em closed} (resp. {\em open}) {\em disk} is a set homeomorphic to the set
$\{(x,y)\in \bR^{2}\mid x^{2}+y^{2}\leq 1\}$ (resp. $\{(x,y)\in \bR^{2}\mid x^{2}+y^{2}< 1\}$).
Let $\Delta$ be a closed disk.
Given a subset $X$ of $\Delta$, we
denote its closure by $\bar{X}$ and its boundary by $\bd(X)$.
A {\em {$\Delta$}-painting} is a pair $\Gamma=(U,N)$
where
\begin{itemize}
	\item $N$ is a finite set of points of $\Delta$,
	\item $N \subseteq U \subseteq \Delta$, and
	\item $U \setminus N$ has finitely many arcwise-connected components, called {\em cells}, where for every cell~$c$,
	      \begin{itemize}
		      \item the closure $\bar{c}$ of $c$
		            is a closed disk
		            and
		      \item  $|\tilde{c}|\leq 3$, where $\tilde{c}:=\bd(c)\cap N$.
	      \end{itemize}
\end{itemize}

We use the notation $U(\Gamma) := U$,
$N(\Gamma) := N$ and denote the set of cells of $\Gamma$
by $C(\Gamma)$.
For convenience, we may assume that each cell of $\Gamma$ is an open disk of $\Delta$.

Notice that, given a $\Delta$-painting $\Gamma$,
the pair $(N(\Gamma),\{\tilde{c}\mid c\in C(\Gamma)\})$ is a hypergraph whose hyperedges have cardinality at most three and $\Gamma$ can be seen as a plane embedding of this hypergraph in $\Delta$.

\subparagraph{Renditions.}
Let $G$ be a graph and let $\Omega$ be a cyclic permutation of a subset of $V(G)$ that we denote by $V(\Omega)$. By an {\em $\Omega$-rendition} of $G$ we mean a triple $(\Gamma, \sigma, \pi)$, where
\begin{itemize}
	\item[(a)] $\Gamma$ is a $\Delta$-painting for some closed disk $\Delta$,
	\item[(b)] $\pi: N(\Gamma)\to V(G)$ is an injection, and
	\item[(c)] $\sigma$ assigns to each cell $c \in C(\Gamma)$ a subgraph $\sigma(c)$ of $G$, such that
	      \begin{enumerate}
		      \item[(1)] $G=\bigcup_{c\in C(\Gamma)}\sigma(c)$,
		      \item[(2)]  for distinct $c, c' \in  C(\Gamma)$,  $\sigma(c)$ and $\sigma(c')$  are edge-disjoint,
		      \item[(3)] for every cell $c \in  C(\Gamma)$, $\pi(\tilde{c}) \subseteq V (\sigma(c))$,
		      \item[(4)]  for every cell $c \in  C(\Gamma)$,
		            $V(\sigma(c)) \cap \bigcup_{c' \in  C(\Gamma) \setminus  \{c\}}V(\sigma(c')) \subseteq \pi(\tilde{c})$, and
		      \item[(5)]  $\pi(N(\Gamma)\cap \bd(\Delta))=V(\Omega)$, such that the points
		            in $N(\Gamma)\cap \bd(\Delta)$ appear in $\bd(\Delta)$ in the same ordering
		            as their images, via $\pi$, in $\Omega$.
	      \end{enumerate}
\end{itemize}

\subsection{Flatness pairs}\label{@tugendlehre}

In this subsection we define the notion of a flat wall.
The definitions given here are originating from \cite{SauST21amor}.
We refer the reader to that paper for a more detailed exposition of these definitions and the reasons for which they were introduced.
We use the more accurate framework of \cite{SauST21amor} concerning flat walls, instead of that of \cite{KawarabayashiTW18anew},
in order to be able to use tools that are developed in \cite{SauST21amor,SauST21kapiI,SauST21kapiII}.

\subparagraph{Flat walls.}
Let $G$ be a graph and let $W$ be an $r$-wall of $G$, for some odd integer $r\geq 3$.
We say that a pair $(P,C)\subseteq V(D(W))\times V(D(W))$ is a {\em choice of pegs and corners for $W$} if $W$ is a subdivision of an elementary $r$-wall $\bar{W}$ where $P$ and $C$ are the pegs and the corners of $\bar{W}$, respectively (clearly, $C\subseteq P$).
To get more intuition, notice that a wall $W$ can occur in several ways from the elementary wall $\bar{W}$,
depending on the way the vertices in the perimeter of $\bar{W}$ are subdivided.
Each of them gives a different selection $(P,C)$ of pegs and corners of $W$.

We say that $W$ is a {\em flat $r$-wall}
of $G$ if there is a separation $(X,Y)$ of $G$ and a choice $(P,C)$
of pegs and corners for $W$ such that:
\begin{itemize}
	\item $V(W)\subseteq Y$,
	\item $P\subseteq X\cap Y\subseteq V(D(W))$, and
	\item if $\Omega$ is the cyclic ordering of the vertices $X\cap Y$ as they appear in $D(W)$,
	      then there exists an $\Omega$-rendition $(\Gamma,\sigma,\pi)$ of  $G[Y]$.
\end{itemize}

We say that $W$ is a {\em flat wall}
of $G$ if it is a flat $r$-wall for some odd integer $r \geq 3$.

\subparagraph{Flatness pairs.}
Given the above, we say that the choice of the 7-tuple $\mathfrak{R}=(X,Y,P,C,\Gamma,\sigma,\pi)$
{\em certifies that $W$ is a flat wall of $G$}.
We call the pair $(W,\mathfrak{R})$ a {\em flatness pair} of $G$ and define
the {\em height} of the pair $(W,\mathfrak{R})$ to be the height of $W$.
We use the term {\em cell of} $\mathfrak{R}$ in order to refer to the cells of $\Gamma$.

We call the graph $G[Y]$ the {\em $\mathfrak{R}$-compass} of $W$ in $G$,
denoted by $\compass_\mathfrak{R}(W)$.
We can assume that $\compass_\mathfrak{R} (W)$ is connected, updating $\mathfrak{R}$ by removing from $Y$ the vertices of all the connected components of $\compass_\mathfrak{R} (W)$
except for the one that contains $W$ and including them in $X$ ($\Gamma, \sigma, \pi$ can also be easily modified according to the removal of the aforementioned vertices from $Y$).
We define the {\em flaps} of the wall $W$ in $\mathfrak{R}$ as
$\flaps_\mathfrak{R}(W):=\{\sigma(c)\mid c\in C(\Gamma)\}$.
Given a flap $F\in \flaps_\mathfrak{R}(W)$, we define its {\em base}
as $\partial F:=V(F)\cap \pi(N(\Gamma))$.
A cell $c$ of $\mathfrak{R}$ is {\em untidy} if $\pi(\tilde{c})$ contains a vertex
$x$ of ${W}$ such that two of the edges of ${W}$ that are incident to $x$ are edges of $\sigma(c)$. Notice that if $c$ is untidy then $|\tilde{c}|=3$.
A cell $c$ of $\mathfrak{R}$ is {\em tidy} if it is not untidy.

\subparagraph{Cell classification.}
Given a cycle $C$ of $\compass_\mathfrak{R}(W)$, we say that
$C$ is {\em $\mathfrak{R}$-normal} if it is {\sl not} a subgraph of a flap $F\in \flaps_\mathfrak{R}(W)$.
Given an $\mathfrak{R}$-normal cycle $C$ of $\compass_\mathfrak{R}(W)$,
we call a cell $c$ of $\mathfrak{R}$ {\em $C$-perimetric} if
$\sigma(c)$ contains some edge of $C$.
Notice that if $c$ is $C$-perimetric, then $\pi(\tilde{c})$ contains two points $p,q\in N(\Gamma)$
such that $\pi(p)$ and $\pi(q)$ are vertices of $C$ where one,
say $P_{c}^{\rm in}$, of the two $(\pi(p),\pi(q))$-subpaths of $C$ is a subgraph of $\sigma(c)$ and the other,
denoted by $P_{c}^{\rm out}$, $(\pi(p),\pi(q))$-subpath contains at most one internal vertex of $\sigma(c)$,
which should be the (unique) vertex $z$ in $\partial\sigma(c)\setminus\{\pi(p),\pi(q)\}$.
We pick a $(p,q)$-arc $A_{c}$ in $\hat{c}:={c}\cup\tilde{c}$ such that $\pi^{-1}(z)\in A_{c}$ if and only if $P_{c}^{\rm in}$ contains
the vertex $z$ as an internal vertex.

We consider the circle $K_{C}=\cupall\{A_{c}\mid \mbox{$c$ is a $C$-perimetric cell of $\mathfrak{R}$}\}$
and we denote by $\Delta_{C}$ the closed disk bounded by $K_{C}$ that is contained in $\Delta$.
A cell $c$ of $\mathfrak{R}$ is called {\em $C$-internal} if $c\subseteq \Delta_{C}$
and is called {\em $C$-external} if $\Delta_{C}\cap c=\emptyset$.
Notice that the cells of $\mathfrak{R}$ are partitioned into $C$-internal, $C$-perimetric, and $C$-external cells.

Let $c$ be a tidy $C$-perimetric cell of $\mathfrak{R}$ where $|\tilde{c}|=3$. Notice that $c\setminus A_{c}$ has two arcwise-connected components and one of them is an open disk $D_{c}$ that is a subset of $\Delta_{C}$.
If the closure $\overline{D}_{c}$ of $D_{c}$ contains only two points of $\tilde{c}$ then we call the cell $c$ {\em $C$-marginal}.

\subparagraph{Influence.}
For every $\mathfrak{R}$-normal cycle $C$ of $\compass_\mathfrak{R}(W)$ we define the set
$\influence_\mathfrak{R}(C)=\{\sigma(c)\mid \mbox{$c$ is a cell of $\mathfrak{R}$ that is not $C$-external}\}$.

A wall $W'$ of $\compass_\mathfrak{R}(W)$ is \emph{$\mathfrak{R}$-normal} if $D(W')$ is $\mathfrak{R}$-normal.
Notice that every wall of $W$ (and hence every subwall of $W$) is an $\mathfrak{R}$-normal wall of $\compass_\mathfrak{R}(W)$. We denote by ${\cal S}_\mathfrak{R}(W)$ the set of all $\mathfrak{R}$-normal walls of $\compass_\mathfrak{R}(W)$. Given a wall $W'\in {\cal S}_\mathfrak{R}(W)$ and a cell $c$ of $\mathfrak{R}$,
we say that $c$ is {\em $W'$-perimetric/internal/external/marginal} if $c$ is $D(W')$-perimetric/internal/external/mar\-ginal, respectively.
We also use $K_{W'}$, $\Delta_{W'}$, $\influence_\mathfrak{R}(W')$ as shortcuts
for $K_{D(W')}$, $\Delta_{D(W')}$, and $\influence_\mathfrak{R}(D(W'))$, respectively.

\subparagraph{Regular flatness pairs.}
We call a flatness pair $(W,\mathfrak{R})$ of a graph $G$ {\em regular}
if none of its cells is $W$-external, $W$-marginal, or untidy.

\subparagraph{Tilts of flatness pairs.}
Let $(W,\mathfrak{R})$ and $(\tilde{W}',\tilde{\mathfrak{R}}')$ be two flatness pairs of a graph $G$ and let $W'\in {\cal S}_\mathfrak{R}(W)$.
We assume that $\mathfrak{R}=(X,Y,P,C,\Gamma,\sigma,\pi)$ and $\tilde{\mathfrak{R}}'=(X',Y',P',C',\Gamma',\sigma',\pi')$.
We say that $(\tilde{W}',\tilde{\mathfrak{R}}')$ is a {\em $W'$-tilt} of $(W,\mathfrak{R})$ if
\begin{itemize}
	\item $\tilde{\mathfrak{R}}'$ does not have $\tilde{W}'$-external cells,
	\item $\tilde{W}'$ is a tilt of $W'$,
	\item the set of $\tilde{W}'$-internal cells of $\tilde{\mathfrak{R}}'$ is the same as the set of $W'$-internal
	      cells of $\mathfrak{R}$ and their images via $\sigma'$ and ${\sigma}$ are also the same,
	\item $\compass_{\tilde{\mathfrak{R}}'}(\tilde{W}')$ is a subgraph of $\cupall\influence_\mathfrak{R}(W')$, and
	\item if $c$ is a cell in $C(\Gamma') \setminus C(\Gamma)$, then $|\tilde{c}| \leq 2$.
\end{itemize}

The next observation follows from the third item above and the fact that the cells corresponding to flaps containing a central vertex of $W'$ are all internal (recall that the height of a wall is always at least three).

\begin{observation}\label{@unimagined}
	Let $(W,\mathfrak{R})$ be a flatness pair of a graph $G$ and $W'\in{\cal S}_\mathfrak{R}(W)$.
	For every $W'$-tilt $(\tilde{W}',\tilde{\mathfrak{R}}')$ of $(W,\mathfrak{R})$, the central vertices of $W'$ belong to the vertex set of $\compass_{\tilde{\mathfrak{R}}'}(\tilde{W}')$.
\end{observation}

Also, given a regular flatness pair $(W,\mathfrak{R})$ of a graph $G$ and a $W'\in {\cal S}_\mathfrak{R}(W)$,
for every $W'$-tilt $(\tilde{W}', \tilde{\mathfrak{R}}')$ of $(W,\mathfrak{R})$, by definition none of its cells is $\tilde{W}'$-external, $\tilde{W}'$-marginal, or untidy -- thus, $(\tilde{W}', \tilde{\mathfrak{R}}')$ is regular.
Therefore, regularity of a flatness pair is a property that its tilts ``inherit''.

\begin{observation}\label{@successively}
	If $(W,\mathfrak{R})$ is a regular flatness pair, then for every $W'\in {\cal S}_\mathfrak{R}(W)$, every $W'$-tilt of $(W,\mathfrak{R})$ is also regular.
\end{observation}
\medskip

Furthermore, we need the following propositions, that are the main results of~\cite{SauST21amor}.

\begin{proposition}[\cite{SauST21amor}]\label{@expurgated}
There exists an algorithm that, given a graph $G$, a flatness pair $({W},\mathfrak{R})$ of $G$, and a wall $W'\in {\cal S}_\mathfrak{R}(W)$, outputs a $W'$-tilt of $({W},\mathfrak{R})$ in time $\Ocal(n+m)$.
\end{proposition}

\begin{proposition}[\cite{SauST21amor}]\label{@philosophic}
Let $G$ be a graph and $(W,\mathfrak{R})$ be a flatness pair of $G$.
There is a regular flatness pair $(W^*,\mathfrak{R}^*)$ of $G$, with the same height as $(W,\mathfrak{R})$, such that $\compass_{\mathfrak{R}^*}(W^*)\subseteq \compass_\mathfrak{R}(W)$.
\end{proposition}

\subsection{Canonical partitions}\label{@commodation}

\sugar{In this subsection, we define the notion of canonical partition of
a graph with respect to some wall. This refers to a partition of the vertex set of a graph in bags that follow the structure of a wall subgraph of the given graph.
For this reason, we start by defining the canonical partition of a wall, as a ``canonical'' way to partition the vertices of the wall in connected subsets that preserve the grid-like structure of the wall.}

\subparagraph{Canonical partition of a wall.}
Let $r\geq 3$ be an odd integer.
Let $W$ be an $r$-wall and let $P_{1}, \ldots, P_{r}$ (resp. $L_{1},\ldots, L_{r}$) be its vertical (resp. horizontal) paths.
For every even (resp. odd) $i\in[2,r-1]$ and every $j\in[2,r-1]$, we define ${A}^{(i,j)}$ to be the subpath of $P_{i}$ that starts from a vertex of $P_{i}\cap L_{j}$ and finishes at a neighbor of a vertex in $L_{j+1}$ (resp. $L_{j-1}$), such that $P_{i}\cap L_{j}\subseteq A^{(i,j)}$ and $A^{(i,j)}$ does not intersect $L_{j+1}$ (resp. $L_{j-1}$).
Similarly, for every $i,j\in[2,r-1]$, we define $B^{(i,j)}$ to be the subpath of $L_{j}$ that starts from a vertex of $P_{i}\cap L_{j}$ and finishes at a neighbor of a vertex in $P_{i-1}$, such that $P_{i}\cap L_{j}\subseteq B^{(i,j)}$ and $B^{(i,j)}$ does not intersect $P_{i-1}$.

\begin{figure}[ht]
	\centering
	\includegraphics[width=8cm]{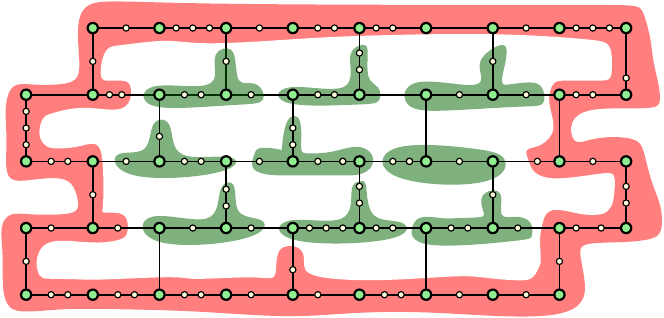} % TS changed
                                % from 6cm
	\caption{\small A $5$-wall and its canonical partition ${\cal Q}$. The red bag is the external bag $Q_{\rm ext}$.}
	\label{@aberration}
\end{figure}

For every $i,j\in[2,r-1]$, we denote by $Q^{(i,j)}$ the graph $A^{(i,j)}\cup B^{(i,j)}$ and by ${Q_{\rm ext}}$ the graph $W\setminus \bigcup_{i,j\in[2,r-1]} V(Q_{i,j})$.
Now consider the collection ${\cal Q}=\{Q_{\rm ext}\}\cup\{Q_{i,j}\mid i,j\in[2,r-1]\}$
and observe that the graphs in ${\cal Q}$ are connected subgraphs of $W$ and their vertex sets form a partition of $V(W)$.
We call ${\cal Q}$ the {\em canonical partition} of $W$. Also, we call every $Q_{i,j}, i,j\in[2,r-1]$ an {\em internal bag} of ${\cal Q}$, while we refer to $Q_{\rm ext}$ as the {\em external bag} of ${\cal Q}$. See \autoref{@aberration} for an illustration of the notions defined above.
For every $i\in[(r-1)/2]$, we say that a set $Q\in {\cal Q}$ is an {\em $i$-internal bag of ${\cal Q}$} if $V(Q)$ does not contain any vertex of the first $i$ layers of $W$.
Notice that the $1$-internal bags of ${\cal Q}$ are the internal bags of ${\cal Q}$.

\subparagraph{Canonical partition of a graph with respect to a wall.}
Let $W$ be a wall of a graph $G$.
Consider the canonical partition ${\cal Q}$ of $W$.
The \emph{enhancement} of the canonical partition $\Qcal$ on $G$ is the following operation.
We set $\tilde{\cal Q}:={\cal Q}$
and, as long as there is a vertex $x\in G\setminus V(\cupall \tilde{\cal Q})$ that is adjacent to a vertex of a graph $Q\in \tilde{\cal Q}$, we update $\tilde{\cal Q}:=\tilde{\cal Q}\setminus \{Q\}\cup \{\tilde{Q}\}$, where $\tilde{Q}=G[\{x\}\cup V(Q)]$.
We call the $\tilde{Q}\in\tilde{\cal Q}$ that contains $Q_{\rm ext}$ as a subgraph the {\em external bag} of $\tilde{\cal Q}$, and we denote it by $\tilde{Q}_{\rm ext}$, while we call {\em internal bags} of $\tilde{\cal Q}$ all graphs in $\tilde{\cal Q}\setminus \{\tilde{Q}_{\rm ext}\}$.
Moreover, we enhance $\tilde{\cal Q}$ by adding all vertices of $G\setminus\bigcup_{\tilde{Q}\in\tilde{\Qcal}}V(\tilde{Q})$ in its external bag, i.e., by updating $\tilde{Q}_{\rm ext}: = G[V(\tilde{Q}_{\rm ext})\cup V(G)\setminus\bigcup_{\tilde{Q}\in\tilde{\Qcal}}V(\tilde{Q})]$.

We call such a partition $\tilde{\cal Q}$ a {\em $W$-canonical partition of $G$.}
Notice that a $W$-canonical partition of $G$ is not unique, since the sets in ${\cal Q}$ can be ``expanded'' arbitrarily when introducing vertex $x$.

Let $W$ be an $r$-wall of a graph $G$, for some odd integer $r\geq 3$ and let
$\tilde{\cal Q}$ be a $W$-canonical partition of $G$.
For every $i\in[(r-1)/2]$, we say that a set
$Q\in \tilde{\cal Q}$ is an {\em $i$-internal bag of $\tilde{\cal Q}$}
if it contains an $i$-internal bag of ${\cal Q}$ as a subgraph.
\medskip

\sugar{We stress that given a graph $G$ and an $r$-wall $W$ of $G$ for some odd integer $r\geq 3$,
a $W$-canonical partition $\tilde\Qcal$ of $G$ can have internal bags
that are adjacent to vertices of arbitrarily many other bags.
However, if $W$ is a flat wall of $G$ certified by some $7$-tuple $\mathfrak{R}$, the ``flat structure'' of the $\mathfrak{R}$-compass of $W$ implies that every bag can be adjacent to only bags that contain vertices of the same brick of $W$.
\medskip

The next result is also proved in \cite{SauST21kapiI} and
intuitively states that, given a flatness pair $(W,\mathfrak{R})$ of ``big enough'' height and a $W$-canonical partition $\tilde\Qcal$ of $G$, we can find a ``packing'' of subwalls of $W$ that are inside some central part of $W$ and such that the vertex set of every internal bag of $\tilde{\cal Q}$ intersects the vertices of the flaps in the influence of at most one of these walls.
We will use this result in the case where the set $A'$ of \autoref{@proclamation} is ``small'', i.e., there are only ``few'' vertices in $A$ that have ``big enough'' degree with respect to the central part of the canonical partition, and therefore \autoref{@proclamation} cannot justify branching.
Following the latter condition and \autoref{@prohibitions}, we will be able to find a flatness pair with ``few'' apices so as to build irrelevant vertex arguments inside its compass.}
{In~\cite{SauST21kapiI}, {\em $(W,\mathfrak{R})$-canonical partition} is used to denote a $W$-canonical partition of $G$, where $(W,\mathfrak{R})$ is a flatness pair of $G$.
However, in this subsection we provide a more general definition that does not take into account the flatness of $W$.}

\begin{proposition}[\cite{SauST21kapiI}]\label{@prohibitions}
There exists a function $\newfun{@idealistic}: \bN^3 \to \bN$ such that if $p,l\in\bN_{\geq 1}$, $r\in\bN_{\geq 3}$ is an odd integer, $G$ is a graph, $(W,\mathfrak{R})$ is a flatness pair of $G$ of height at least $\funref{@idealistic}(l,r,p)$, and $\tilde\Qcal$ is a $W$-canonical partition of $G$, then
	there is a collection $\Wcal=\{W^1, \ldots, W^l\}$ of $r$-subwalls of $W$ such that
	\begin{itemize}
		\item for every $i \in [l]$, $\bigcup \influence_\mathfrak{R}(W^i)$ is a subgraph of $\bigcup \{Q\mid Q \text{ is a $p$-internal bag of }\tilde\Qcal\}$ and
		\item for every $i,j\in[l]$, with $i\neq j$, there is no internal bag of $\tilde\Qcal$ that contains vertices of both $V(\bigcup \influence_\mathfrak{R}(W^i))$ and $V(\bigcup \influence_\mathfrak{R}(W^j))$.
	\end{itemize}
	Moreover, $\funref{@idealistic}(l,r,p)= \Ocal(\sqrt{l}\cdot r + p)$ and $\Wcal$ can be constructed in time $\Ocal(n+m)$.
\end{proposition}

\subsection{Bidimensionality of sets}
\label{@unequivocally}

\sugar{
In this subsection, we present the notion of bidimensionality of a set with respect to a wall of a graph.
This notion intuitively expresses the ``spread'' of a set $X$ in a $W$-canonical partition of $G$. The crucial idea is that a set $X$ of small bidimensionality cannot ``destroy'' a (flat) wall too much.}

\subparagraph{Bidimensionality.}
Let $W$ be a wall of a graph $G$, $\tilde{\Qcal}$ be a $W$-canonical partition of $G$, and $X\subseteq V(G)$.
The \emph{bidimensionality} of $X$ in $G$ with respect to $\tilde{\Qcal}$, denoted by $\bid_{\tilde{\Qcal}}(X)$, is the number of internal bags of $\tilde{\Qcal}$ intersected by $X$.
The \emph{bidimensionality} of $X$ in $G$ with respect to $W$, denoted by $\bid_{G,W}(X)$, is the maximum bidimensionality of $X$ with respect to a $W$-canonical partition of $G$.
\medskip

\autoref{@psychoanalyse} and \autoref{@lucinatory} provide sets $X\subseteq V(G)$ that can be used to apply \autoref{@proclamation} to {\sc $\Fcal$-M-Deletion} and {\sc $\Fcal$-M-Elimination Distance}, respectively.\smallskip

Every set $S\subseteq V(G)$ of size at most $k$ clearly
has bidimensionality at most $k$.

\begin{observation}\label{@psychoanalyse}
Let $\Fcal$ be a finite collection of graphs, let $G$ be a graph, let $k\in\bN$, and let $W$ be a wall of $G$.
Then for every $k$-apex set $S$ of $G$ for $\exc(\Fcal)$ it holds that $\bid_{G,W}(S)\leq k$.
\end{observation}

\sugar{Moreover, given a $k$-elimination set $S$, we can find a set $X\supseteq S$ of bidimensionality at most $k(k+1)/2$ such that $G\setminus X\in\exc(\Fcal)$. To prove this, we first prove the following result, which
 intuitively states that a $k$-elimination set can intersect at most $k$ horizontal and vertical paths of a wall.}

\begin{lemma}\label{@unchangeable}
Let $\Fcal$ be a finite collection of graphs.
Let $G$ be a graph, let $k\in\bN$, let $r,h$ be odd integers with $r\geq h+k$, let $W$ be an $r$-wall of $G$, and let $S$ be a $k$-elimination set of $G$ for $\exc(\Fcal)$.
Then there is an $h$-subwall $W'$ of $W$ with $V(W') \cap S = \emptyset$.
\end{lemma}

\begin{proof}
Since $S$ is a $k$-elimination set of $G$ for $\exc(\Fcal)$, there is an
 $\Fcal$-elimination forest $(F,\chi,R)$ of $G$ of height $k$ such that
$\chi(\Int(F,R))= S$.

We set $W_0:=W$.
For $i\in[k]$, we proceed to construct an $r_i$-subwall $W_i$ of the $r_{i-1}$-wall $W_{i-1}$ with $r_i\geq r-i$
such that
for every $i\in[k]$ there is a node $z_i$ of $F$ such that $(F_{z_i},z_i)$ has height at most $k-i$
and $V(W_i)\subseteq \chi(V(F_{z_i}))$.
This will imply the existence of a wall of size at least $r-k$ whose vertex set will be a subset of $\chi(V(F_{z})$, where $z\in \leaf(F,R)$.

Let $S_i:=V(W_{i-1})\cap S$.
If $S_i=\emptyset$, we set $W_j:=W_{i-1}$ for $j\in[i,k]$.
Otherwise, let $u_i$ be the least common ancestor of $\chi^{-1}(S_i)$ in $(F,R)$.
According to \autoref{@unreflectingly}, since $W_{i-1}$ is connected, $u_i$ exists and belongs to $\chi^{-1}(S_i)$.

We obtain an $(r_{i-1} - 1)$-subwall $W_i$ of $W_{i-1}$ that does not contain $\chi(u_i)$ by taking the wall containing all horizontal and vertical paths of $W_{i-1}$ aside from the ones intersecting $u_i$ and we set $r_i := r_{i-1} - 1$ (to simplify the argument, here we call the resulting graph a wall even when the height is even).
Note that since $W_i$ is a subgraph of $G[\chi(V(F_{u_i}))]$ that is connected, there is a $z_i\in \ch_{F_{u_i},u_i}(u_i)$ such that
$V(W_i)\subseteq V(F_{z_i})$.
Notice that $(F_{u_i},\chi|_{V(F_{u_i}},u_i)$ is an $\Fcal$-elimination forest of $G[\chi(F_{u_i})]$ of height at most $k-i$.

Observe that $z_k$ should be a leaf of $(F,R)$ and therefore, since we have that
$V(W_k)\subseteq V(F_{z_k})\subseteq V(G)\setminus S$,
$W_k$ is a wall of $G\setminus S$ of height at least $r-k\geq h$.
\end{proof}

Now we prove our result regarding the bidimensionality of $k$-elimination sets.

\begin{lemma}\label{@lucinatory}
Let $\Fcal$ be a finite collection of graphs.
Let $G$ be a graph, let $A\subseteq V(G)$, let $k\in\bN$, let $r\geq 2k+3$ be an odd integer, let $(W,\mathfrak{R})$ be a flatness pair of $G\setminus A$, and let $S$ be a $k$-elimination set of $G$ for $\exc(\Fcal)$.
There is a set $X\supseteq S$ such that $G\setminus X\in\exc(\Fcal)$ and $\bid_{G\setminus A, W}(X)\leq k(k+1)/2$.
\end{lemma}

\begin{proof}
Let $p={\sf odd}(r-k)$.
Let $W'$ be a $p$-subwall of $W$ that is a wall of $G\setminus S$, which exists due to \autoref{@unchangeable}.
Let $C$ be the connected component of $G\setminus S$ that contains $W'$.
Since $C\in\exc(\Fcal)$, $K_{s_\Fcal}$ is not a minor of $C$.
Moreover, since $S$ is a $k$-elimination set of $G$ for $\exc(\Fcal)$, there is a set $P\subseteq S$ of size at most $k$ such that $(L,R):=(V(G)\setminus V(C),V(C)\cup P)$ is a separation of $G$ with $L\cap R=P$.

Let us show that $\bid_{G\setminus A, W}(V(G)\setminus V(C))\leq k(k+1)/2$.
Let $\tilde{\Qcal}$ be a $W$-canonical partition of $G\setminus A$.
Let $l$ be the number of internal bags of $\tilde{\Qcal}$ intersected by $P$ and note that $l\leq k$.

Let $G'$ be the graph obtained from $G$ after contracting each bag of $\tilde{\Qcal}$ to a vertex. It is easy to observe that $G'$ is isomorphic to a planar supergraph of an $h$-grid $H$, where $h=r-2$, together with an additional vertex that is adjacent to every vertex of the perimeter of $H$.

We let $[h]^2$ be the vertex set of $H$, where $(i,j)$ and $(i',j')$ are adjacent if and only if $|i-i'|+|j-j'|=1$.
We will show that there is a separation $(L',R')$ of $H$ of order at most $l$ that maximizes $\min\{|L'|,|R'|\}$.
Let $A=L'\cap R'$.
We suppose without loss of generality that $|L'|\leq|R'|$.
Notice that $l<h$. We take $A:=\{(i,j)\in[h]^2\mid i+j=l\}$, i.e., $L'$ is the set of pairs of indices in the triangle bounded by $(0,0)$, $(0,l)$, and $(l,0)$. Thus, $|L'|=l(l+1)/2$. It is easy to verify that this maximizes $|L'|$.

Therefore, since the vertices of $H$ are the internal bags of $\tilde{\Qcal}$ and $P$ intersects $l$ internal bags, it implies that one of $L$ and $R$ intersects at most $l(l+1)/2 \leq k(k+1)/2$ internal bags of $\tilde{\Qcal}$.
Recall that $W'$ is a wall of $C$ of height $p$.
It is easy to verify that an elementary $x$-wall $W^*$ has $2x^2-2$ vertices with $8x-10$ vertices in the perimeter.
Hence, it has $2(x-2)^2$ vertices not in the perimeter, and therefore the canonical partition of $W^*$ has $(x-2)^2$ internal bags.
Thus, the canonical partition of $W'$ has $(p-2)^2$ internal bags.
Observe that each such a bag is contained in an internal bag of $\tilde{\Qcal}$ and therefore $V(C)$ intersects at least $(p-2)^2$ internal bags of $\tilde{\Qcal}$.
Since $(p-2)^2\geq (r-k-2)^2\geq(k+1)^2>k(k+1)/2$, it holds that $\bid_{G\setminus A, W}(V(C)\cup P)>k(k+1)/2$.
Therefore, $\bid_{G\setminus A, W}(V(G)\setminus V(C))\leq k(k+1)/2$.
\end{proof}

\subsection{Homogeneous walls}\label{@indigenous}

\sugar{
In this subsection, we define homogeneous flat walls.
 Intuitively, homogeneous flat walls are flat walls that allow the routing of the same set of (topological) minors in the augmented flaps (i.e., the flaps together with the apex set) ``cropped'' by each one of their bricks.
Such a flat wall can be detected in a big enough flat wall (\autoref{@disreputable}) and this ``homogeneity'' property implies
that some central part of a big enough homogeneous wall can be declared irrelevant (\autoref{@civilizing}).
The results presented in this subsection are from~\cite{SauST21kapiI,SauST21kapiII}.
}

\subparagraph{Folios.}
We say that $(M,T)$ is a {\em {\sf tm}-pair} if $M$ is a graph, $T\subseteq V(M)$, and all vertices in
$V(M)\setminus T$ have degree two. We denote by ${\sf diss}(M,T)$ the graph obtained
from $M$ by {dissolving} all vertices in $V(M)\setminus T$.
A {\em {\sf tm}-pair} of a graph $G$ is a {\sf tm}-pair $(M,T)$ where $M$ is a subgraph of $G$.
We call the vertices in $T$ {\em branch} vertices of $(M,T)$.
We need to deal with topological minors for the notion of homogeneity defined below, on which the statement of~\cite[Theorem 5.2]{BasteST20acom} relies.
If $\textbf{M}=(M,B,\rho)\in{\cal B}$ and $T\subseteq V(M)$ with $B\subseteq T$, we call $(\textbf{M},T)$ a {\em {\sf btm}-pair}
and we define ${\sf diss}(\textbf{M},T)=({\sf diss}(M, T),B,\rho)$. Note that we do not permit dissolution of boundary vertices, as we consider all of them to be branch vertices. If $\textbf{G}=(G,B,\rho)$ is a boundaried graph and $(M,T)$ is a {\sf tm}-pair of $G$
where $B\subseteq T$, then we say that
$(\textbf{M},T)$, where $\textbf{M}=(M,B,\rho)$, is a {\em {\sf btm}-pair} of $\textbf{G}=(G,B,\rho)$.
Let $\textbf{G}_{1},{\bf G}_{2}$ be two boundaried graphs.
We say that $\textbf{G}_{1}$ is a {\em topological minor}
of $\textbf{G}_{2}$, denoted by $\textbf{G}_{1}\pretp\textbf{G}_{2}$, if
$\textbf{G}_{2}$ has a {\sf btm}-pair $(\textbf{M},T)$
such that ${\sf diss}(\textbf{M},T)$ is isomorphic to $\textbf{G}_{1}$.
Given a $\textbf{G}\in {\cal B} $ and a positive integer $\ell$, we define the {\em $\ell$-folio} of ${\bf G}$
as
$${\ell}\mbox{\sf-folio}(\textbf{G})=\{\textbf{G}'\in {\cal B} \mid \textbf{G}'\pretp \textbf{G} \mbox{~and $\textbf{G}'$ has detail at most $\ell$}\}.$$

\subparagraph{Augmented flaps.}
Let $G$ be a graph, $A$ be a subset of $V(G)$ of size $a$, and $(W,\mathfrak{R})$ be a flatness pair of $G\setminus A$.
For each flap $F\in \flaps_\mathfrak{R}(W)$ we consider a labeling $\ell_{F}: \partial F\rightarrow\{1,2,3\}$ such that
the set of labels assigned by $\ell_{F}$ to $\partial F$ is one of $\{1\}$, $\{1,2\}$, $\{1,2,3\}$.
Also, let $\tilde{a}\in[a]$.
For every set $\tilde{A}\in\binom{A}{\tilde{a}}$, we consider a bijection $\rho_{\tilde{A}}: \tilde{A}\to [\tilde{a}]$.
The labelings in ${\cal L}=\{\ell_{F} \mid F\in \flaps_\mathfrak{R}(W)\}$ and the labelings in $\{\rho_{\tilde{A}} \mid \tilde{A}\in\binom{A}{\tilde{a}}\}$ will be useful for defining a set of boundaried graphs that we will call augmented flaps.
We first need some more definitions.

Given a flap $F\in\flaps_\mathfrak{R}(W)$, we define an ordering
$\Omega(F)=(x_{1},\ldots,x_{q})$, with $q\leq 3$, of the vertices of $\partial{F}$
so that
\begin{itemize}
	\item $(x_{1},\ldots,x_{q})$ is a counter-clockwise cyclic ordering of the vertices of $\partial F$ as they appear in the corresponding cell of $C(\Gamma)$. Notice that this cyclic ordering is significant only when $|\partial F|=3$,
in the sense that $(x_{1},x_{2},x_{3})$ remains invariant under shifting, i.e., $(x_{1},x_{2},x_{3})$ is the same as $ (x_{2},x_{3},x_{1})$ but not under inversion, i.e., $(x_{1},x_{2},x_{3})$ is not the same as $(x_{3},x_{2},x_{1})$, and
	\item for $i\in[q]$, $\ell_{F}(x_{i})=i$.
\end{itemize}
Notice that the second condition is necessary for completing the definition of the ordering $\Omega(F)$,
and this is the reason why we set up the labelings in ${\cal L}$.\medskip

For each set $\tilde{A}\in\binom{A}{\tilde{a}}$ and each $F\in \flaps_\mathfrak{R}(W)$ with $t_{F}:=|\partial F|$,
we fix $\rho_{F}: \partial F\to [\tilde{a}+1,\tilde{a}+t_F]$ such that
$(\rho^{-1}_{F}(\tilde{a}+1),\ldots,\rho^{-1}_{F}(\tilde{a}+t_F))= \Omega(F)$.
Also, we define the boundaried graph $$\textbf{F}^{\tilde{A}}:=(G[\tilde{A}\cup F],\tilde{A}\cup \partial F,\rho_{\tilde{A}}\cup \rho_F)$$
and we denote by $F^{\tilde{A}}$ the underlying graph of $\textbf{F}^{\tilde{A}}$. We call $\textbf{F}^{\tilde{A}}$ an {\em $\tilde{A}$-augmented flap} of the flatness pair $(W,\mathfrak{R})$ of $G\setminus A$
in $G$.

\subparagraph{Palettes and homogeneity.}
For each $\mathfrak{R}$-normal cycle $C$ of $\compass_\mathfrak{R} (W)$ and each set $\tilde{A}\subseteq A$, we define $(\tilde{A},\ell)\mbox{\sf -palette}(C)=\{\ell\mbox{\sf -folio}({\bf F}^{\tilde{A}})\mid F\in {\sf influence}_\mathfrak{R}(C)\}$.
Given a set $\tilde{A}\subseteq A$, we say that the flatness pair $(W,\mathfrak{R})$ of $G\setminus A$ is {\em $\ell$-homogeneous with respect to $\tilde{A}$} if every {\sl internal} brick of ${W}$ has the {\sl same} $(\tilde{A},\ell)$\mbox{\sf -palette} (seen as a cycle of $\compass_\mathfrak{R} (W)$).
Also, given a collection ${\cal S}\subseteq 2^A$, we say that the flatness pair $(W,\mathfrak{R})$ of $G\setminus A$ is {\em $\ell$-homogeneous
		with respect to ${\cal S}$}
if it is $\ell$-homogeneous with respect to every $\tilde{A}\in {\cal S}$.

\section{Vertex deletion to a minor-closed graph class}\label{@conditioned}

In this section, we prove our main result for the {\sc $\Fcal$-M-Deletion} problem. The following theorem is a restatement of \autoref{@communicated} using the reformulation introduced in \autoref{@graphically}.

\begin{theorem}\label{@uncritical}
For every finite collection of graphs $\Fcal$, there exists an algorithm that, given a graph $G$ and a non-negative integer $k$, runs in time $2^{k^{\Ocal_{\ell_\Fcal}(1)}}\cdot n^2$ and either outputs
a $k$-apex set of $G$ for $\exc{(\Fcal)}$
or reports that such a set does not exist.
\end{theorem}

This algorithm is a generalization of the algorithm in the apex-minor free case of \cite{SauST21kapiII}.
In order to give an algorithm without the apex-minor restriction,
we enhance the techniques of \cite{SauST21kapiII} with some new tricks.
We present this algorithm in this paper as a stepping stone to present the algorithms for elimination distance in \autoref{@oeconomicus} and \autoref{@stubbornly}, since the techniques used are similar (however not the same).

\subsection{Description of the algorithm for {\sc \texorpdfstring{$\Fcal$}{F}-M-Deletion}}

\sugar{
Our algorithm for {\sc $\Fcal$-M-Deletion} has three steps.
In Step~1, either we can easily conclude with a positive or a negative answer or we find a big wall.
If we can find a large flat wall of bounded treewidth inside this wall, then we go to Step~2 and find an irrelevant vertex.
Otherwise, we proceed to Step~3 where, by using flow techniques, we find a set of vertices that intersects every solution, and we branch on this set or we report a negative answer.
The correctness of the algorithm is not trivial and will be justified in \autoref{@neoliberal}.
While the general scheme of the algorithm is similar to the algorithm in~\cite{SauST21kapiII} for {\sc $\Fcal$-M-Deletion} in the case where $\Fcal$ contains an apex-graph, here, in order to obtain a quadratic algorithm for general $\Fcal$, we employ additional novel tricks so as to deal with the possible existence of many apices in all graphs in $\Fcal$.}

\smallskip

We define the following constants.
\begin{align*}
	a =				&	\ \funref{@collaboration}(s_\Fcal+a_\Fcal-1), &
	b =				& \ \funref{@collaboration}(s_\Fcal), \\
	q = 			& \ \funref{@categories}(a_\Fcal,s_\Fcal,k), &
	p =       & \ \funref{@provincial}(a_\Fcal,s_\Fcal,k), \\
	l =       & \ (q-1)\cdot (k+b), &
	r_6 =			& \ \funref{@differences}(a+b,\ell_\Fcal,3,k) \\
	d = 			& \ \funref{@deliberation}(a+b,\ell_\Fcal) &
	r_5 =     & \ \funref{@philistines}(r_6,a+b,a+b,d), \\
	t =				& \ \funref{@corollaries}(s_\Fcal)\cdot r_5, &
	r_4 =     & \ \odd(t+3), \\
	r_3 = 		& \ \funref{@idealistic}(a_\Fcal,r_4,1), &
	r_2 = 		& \ \odd(2+\funref{@classifications}(s_\Fcal+a_\Fcal-1) \cdot r_3), \\
	r_2' =    & \ \odd(\max\{\funref{@unaffected}(a_\Fcal,s_\Fcal,k), \funref{@idealistic}(l+1,r_2,p)\}), &
	r_1 =     & \ \odd( \funref{@classifications}(s_\Fcal)\cdot r_2'+k).
\end{align*}

Note that $r_6=\Ocal_{\ell_\Fcal}(k)$, $r_5,r_4,r_3,r_2,t=\Ocal_{\ell_\Fcal}(k^{c})$ and $r_2',r_1=\Ocal_{\ell_\Fcal}(k^{c+2})$ where $c=\funref{@withdrawing}(a+b,a+b,d)$.
Recall from \autoref{@graphically} that we assume that $G$ has $\Ocal_{s_\Fcal}(k\sqrt{\log k}\cdot n)$ edges.

\subparagraph{Step 1.} Run the algorithm {\tt Find-Wall} from \autoref{@transforma} with input $(G,r_1,k)$
and, in time $2^{\Ocal_{\ell_\Fcal}(r_1^2+(k+r_1)\log(k+r_1))}\cdot n=2^{\Ocal_{\ell_\Fcal}(k^{2(c+2)})}\cdot n$,
\begin{itemize}
	\item either report a \no-instance, or
	\item conclude that $\tw(G)\leq \funref{@veneration}(s_\Fcal)\cdot r_1+k$ and solve {\sc $\Fcal$-M-Deletion} in time $2^{\Ocal_{\ell_\Fcal}((r_1+k)\log(r_1+k))}\cdot n=2^{\Ocal_{\ell_\Fcal}(k^{c+2} \cdot \log k)} \cdot n$ using the algorithm of \autoref{@calculated}, or
	\item obtain an $r_1$-wall $W_1$ of $G$.
\end{itemize}

If the output of \autoref{@transforma} is an $r_1$-wall $W_1$, consider all the $\binom{r_1}{r_2}^2=2^{\Ocal_{\ell_\Fcal}(k^c\log k)}$ $r_2$-subwalls of $W_1$.
For each one of them, say $W_2$, let $W_2^*$ be the central $(r_2-2)$-subwall of $W_2$ and let $D_{W_2}$ be the graph obtained from $G$ after removing the perimeter of $W_2$ and taking the connected component containing $W_2^*$.
Run the algorithm {\tt Grasped-or-Flat} of \autoref{@possession} with input $(D_{W_2},r_3,s_\Fcal+a_\Fcal-1,W_2^*)$.
This can be done in time $\Ocal_{s_\Fcal}(k\sqrt{\log k}\cdot n)$.

If for some of these subwalls the result is
a set $A\subseteq V(D_{W_2})$ with $|A|\leq a$ and a flatness pair $(W_3,\mathfrak{R}_3)$ of $D_{W_2}\setminus A$ of height $r_3$ then,
as in \autoref{@prohibitions}, compute a $W_3$-canonical partition $\tilde{\Qcal}$ of $D_{W_2}\setminus A$ and a collection $\Wcal=\{W^1,...,W^{a_\Fcal}\}$ of $r_4$-subwalls of $W_3$ such that
for every $i\in[a_\Fcal]$,
$\bigcup \influence_{\mathfrak{R}_3}(W^i)$ is a subgraph of $\bigcup \{Q\mid Q \text{ is a $p$-internal bag of }\tilde\Qcal\}$ and
 for every $i,j\in[a_\Fcal]$, with $i\neq j$, there is no internal bag of $\tilde\Qcal$ that contains vertices of both $V(\bigcup \influence_{\mathfrak{R}_3}(W^i))$ and $V(\bigcup \influence_{\mathfrak{R}_3}(W^j))$.
This can be done in time $\Ocal_{s_\Fcal}(k\sqrt{\log k}\cdot n)$.

For $i\in[a_\Fcal]$,
let $W^{i*}$ be the central $(r_4-2)$-subwall of $W^i$ and let $D_{W^i}$ be the graph obtained from $D_{W_2}$ after removing $A$ and the perimeter of $W^i$ and taking the connected component containing $W^{i*}$.
Run the algorithm {\tt Clique-or-twFlat} of \autoref{@unimportant} with input $(D_{W^i},r_5,s_\Fcal)$.
This takes time $2^{\Ocal_{\ell_\Fcal}(r_5^2)}\cdot n=2^{\Ocal_{\ell_\Fcal}(k^{2c})}\cdot n$.
If for one of these subwalls the result is a set $A'$ of size at most $b$ and a regular flatness pair $(W_5,\mathfrak{R}_5)$ of $D_{W^i}\setminus A'$ of height $r_5$ whose $\mathfrak{R}_5$-compass has treewidth at most $t$, then we proceed to Step~2.

If, for every flatness pair $(W_3,\mathfrak{R}_3)$ and for every $i\in[a_\Fcal]$, the result is a report that $K_{s_\Fcal}$ is a minor of $D_{W^i}$, then we proceed to Step~3.

\subparagraph{Step 2 (irrelevant vertex case).}
We now obtain a 7-tuple $\mathfrak{R}_5'$ by adding all vertices of $G\setminus V({\sf Compass}_{\mathfrak{R}_5}(W_5))$ to the set in the first coordinate of $\mathfrak{R}_5$, such that $(W_5,\mathfrak{R}_5')$ is a regular flatness pair of $G\setminus (A\cup A')$ whose $\mathfrak{R}_5'$-compass has treewidth at most $t$.
We apply the algorithm {\tt Homogeneous} of \autoref{@disreputable} with input $(r_6,a+b,a+b,d,t,G,A\cup A',W_5,\mathfrak{R}_5')$, which outputs, in time $2^{\Ocal_{\ell_\Fcal}(t\log t + k\log k)}\cdot n=2^{\Ocal_{\ell_\Fcal}(k^c\log k)}\cdot n$, a flatness pair $(W_6,\mathfrak{R}_6)$ of $G\setminus (A\cup A')$ of height $r_6$ that is $d$-homogeneous with respect to $2^{A\cup A'}$ and is a $W^*$-tilt of $(W_5,\mathfrak{R}_5')$ for some subwall $W^*$ of $W_5$.
We apply the algorithm {\tt Find-Irrelevant-Vertex} of \autoref{@civilizing} with input $(k,a+b,G,A\cup A',W_6,\mathfrak{R}_6)$, which outputs, in time $\Ocal(n+m)=\Ocal_{\ell_\Fcal}(k\sqrt{\log k}\cdot n)$, a vertex $v$ such that $(G,k)$ and $(G\setminus v,k)$
are equivalent instances of {\sc $\Fcal$-M-Deletion}.
Then the algorithm runs recursively on the equivalent instance $(G\setminus v,k)$.

\subparagraph{Step 3 (branching case).}
Consider all the $r_2'$-subwalls of $W_1$, which are at most $\binom{r_1}{r_2'}^2=2^{\Ocal_{\ell_\Fcal}(k^{c+2}\log k)}$ many,
and for each of them, say $W_2'$, compute its canonical partition $\Qcal$.
Then, contract each bag $Q$ of $\Qcal$ to a single vertex $v_Q$, and add a new vertex $v_{\rm all}$ and make it adjacent to all $v_{Q}$'s. In the resulting graph $G'$, for every vertex $y$ of $G\setminus V(W_2')$, check, using a flow augmentation algorithm~\cite{Diestel10grap}, whether there are $q$ internally vertex-disjoint paths from $v_{\rm all}$ to $y$ in time $\Ocal(q\cdot m)=\Ocal_{\ell_\Fcal}(k^4\sqrt{\log k}\cdot n)$.
Let $\tilde{A}$ be the set of all such $y$'s.\smallskip

<{If $|\tilde{A}|<a_\Fcal$, then report a \no-instance.}

If $a_\Fcal\leq|\tilde{A}|\leq k+b$, then consider all the $\binom{|\tilde{A}|}{a_\Fcal}=2^{\Ocal_{\ell_\Fcal}(\log k)}$ subsets of $\tilde{A}$ of size $a_\Fcal$.
For each one of them, say $A^*$, construct $\tilde\Qcal$ by
enhancing $\Qcal$ on $G\setminus A^*$.
Then, we distinguish two cases depending on whether for every $A^*$ all its vertices are adjacent to vertices of $q$ $p$-internal bags of $\tilde\Qcal$.

If each vertex of $A^*$ is adjacent to vertices of $q$ $p$-internal bags of
$\tilde\Qcal$, then (due to \autoref{@proclamation}) $A^*$ should intersect every solution of {\sc $\Fcal$-M-Deletion} for the instance $(G,k)$.
Therefore, the algorithm runs recursively on each instance $(G\setminus y,k-1)$ for $y\in A^*$. If one of them is a \yes-instance with $(k-1)$-apex set $S$ of $G\setminus y$, then $(G,k)$ is a \yes-instance with $k$-apex set $S\cup \{y\}$ of $G$.
If all of them are \no-instances, then report a \no-instance.
This concludes the case where each vertex of $A^*$ is adjacent to vertices of $q$ $p$-internal bags of $\tilde\Qcal$.

If for every subset $A^*$ of $\tilde{A}$ of size $a_\Fcal$, there is a vertex of $A^*$ that is not adjacent to vertices of $q$ $p$-internal bags of the given $\tilde\Qcal$, then report a \no-instance.
This concludes the case that $a_\Fcal\leq|\tilde{A}|\leq k+b$.

\smallskip

If for every wall, $|\tilde{A}|> k+b$, then report that $(G,k)$ is a \no-instance of {\sc $\Fcal$-M-Deletion}.
\bigskip

Notice that Step~3, when applied, takes time $2^{\Ocal_{\ell_\Fcal}(k^{c+2}\log k)} \cdot n^2$, because we apply the flow algorithms to each of the $2^{\Ocal_{\ell_\Fcal}(k^{c+2}\log k)}$ $r_2'$-subwalls and for each vertex of $G$.
However, the search tree created by the branching technique has at most $a_\Fcal$ branches and depth at most $k$. So Step~3 cannot be applied more than ${a_\Fcal}^k$ times during the course of the algorithm.
Since Step~1 runs in time $2^{\Ocal_{\ell_\Fcal}(k^{2(c+2)})} \cdot n$, Step~2 runs in time $2^{\Ocal_{\ell_\Fcal}(k^{2c})}\cdot n$, and both may be applied at most $n$ times, the claimed time complexity follows: the algorithm runs in time $2^{\Ocal_{\ell_\Fcal}(k^{2(c+2)})} \cdot n^2$.

\subsection{Correctness of the algorithm}\label{@neoliberal}

Suppose first that $(G,k)$ is a \yes-instance and let $S$ be a $k$-apex set of $G$.
The application of the algorithm  {\tt Find-Wall} of \autoref{@transforma} with input $(G,r_1,k)$ either returns a report that $\tw(G)\leq \funref{@veneration}(s_\Fcal)\cdot r_1+k$ or returns an $r_1$-wall. In the first case, i.e., if $\tw(G)\leq \funref{@veneration}(s_\Fcal)\cdot r_1+k$, the application of the algorithm of \autoref{@calculated} correctly outputs a $k$-apex set of $G$.
We will focus on the latter case, i.e., where the algorithm  {\tt Find-Wall} returns an $r_1$-wall of $G$, say $W_1$.
Since $r_1\geq \funref{@classifications}(s_\Fcal)\cdot r_2'+k$, there is an $( \funref{@classifications}(s_\Fcal)\cdot r_2')$-subwall
of $W_1$, say $W_1^*$, that does not contain vertices of $S$.
Since $G\setminus S$ does not contain $K_{s_\Fcal}$ as a minor, there is no
model of $K_{s_\Fcal}$ grasped by $W_1^*$ and therefore,
due to \autoref{@possession} with input $(G\setminus S, r_2',s_\Fcal, W_1^*)$,
we know that there is a set $B\subseteq V(G\setminus S)$, with $|B|\leq b$,
and a flatness pair $(W_2',\mathfrak{R}_2')$ of $G\setminus (S\cup B)$ of height $r_2'$ such that $W_2'$ is a $W''$-tilt of some subwall $W''$ of $W_1^*$.\smallskip

Let $\Qcal$ be the canonical partition of $W_2'$.
Let $G'$ be the graph obtained by contracting each bag $Q$ of $\Qcal$
to a single vertex $v_Q$, and adding a new vertex $v_{\rm all}$ and making it adjacent to all $v_{Q}$'s.
Let $\tilde{A}$ be the set of vertices $y$ of $G\setminus V(W_2')$ such that there are $q$ internally vertex-disjoint paths from $v_{\rm all}$ to $y$ in $G'$.
We claim that $\tilde{A}\subseteq S\cup B$.
To show this, we first prove that, for every $y\notin S\cup B$, the maximum number of internally vertex-disjoint paths from $v_{\sf all}$ to $y$ in $G'$ is
$k+b+4$.
Indeed,
if $y$ is a vertex in the $\mathfrak{R}_2'$-compass of $W_2'$,
there are at most $k+b$ such paths that intersect the set $S\cup B$ and
at most four paths that do not intersect $S\cup B$ (in the graph $G'\setminus (S\cup B)$)
due to the fact that $(W_2',\mathfrak{R}_2')$ is a flatness pair of $G\setminus (S\cup B)$.
If $y$ is not a vertex in the $\mathfrak{R}_2'$-compass of $W_2'$, then, since by the definition of flatness pairs the perimeter of $W_2'$ together with the set $S\cup B$ separate $y$ from the $\mathfrak{R}_2'$-compass of $W_2'$,
every collection of internally vertex-disjoint paths from $v_{\rm all}$ to $y$ in $G'$ should intersect the set $\{v_{Q_{\rm ext}}\}\cup S\cup B$, where $Q_{\rm ext}$ is the external bag of $\Qcal$. Therefore, in both cases, if $y\notin S\cup B$, the maximum number of internally vertex-disjoint paths from $v_{\sf all}$ to $y$ in $G'$ is
$k+b+4$.
Since $k+b+4<q$, we have that $y\notin\tilde{A}$.
Hence, $\tilde{A}\subseteq S\cup B$ and therefore $|\tilde{A}|\leq k+b$.
Hence, if $(G,k)$ is a \yes-instance we cannot have that $|\tilde{A}|>k+b$, so the algorithm correctly reports a {\sf no}-instance at the end of Step 3.\smallskip

Let $\tilde\Qcal$ be a $W_2'$-canonical partition of $G\setminus (S\cup B)$ obtained by enhancing $\Qcal$
on $G\setminus (S\cup B)$.
Let $\tilde{A}'$ be the set of vertices in $S\cup B$ that are adjacent to vertices of at least $q$ $p$-internal bags of $\tilde\Qcal$ (recall that $\tilde{A}$ is the set of vertices in $S\cup B$ that are adjacent to vertices of at least $q$ internal bags of $\tilde\Qcal$).
Note that $\tilde{A}'\subseteq\tilde{A}$ and therefore $|\tilde{A}'|\leq|\tilde{A}|$.
\medskip

If $|\tilde{A}'|<a_\Fcal$,
then at most $a_\Fcal-1$ vertices of $S\cup B$ are adjacent to vertices of at least $q$ $p$-internal bags of $\tilde\Qcal$.
This means that the $p$-internal bags
of $\tilde\Qcal$ that contain vertices adjacent to some vertex of $(S\cup B)\setminus \tilde{A}'$ are at most $(q-1)\cdot (k+b)=l$.\smallskip

Consider a family $\Wcal=\{W^{1}, \ldots, W^{l+1}\}$ of $l+1$ $r_2$-subwalls of $W_2'$ such that for every $i \in [l+1]$, $\bigcup \influence_{\mathfrak{R}_2'}(W^i)$ is a subgraph of $\bigcup \{Q\mid Q \text{ is a $p$-internal bag of }\tilde\Qcal\}$ and for every $i,j\in[l+1]$, with $i\neq j$, there is no internal bag of $\tilde\Qcal$ that contains vertices of both $V(\bigcup \influence_{\mathfrak{R}_2'} (W^i))$ and $V(\bigcup \influence_{\mathfrak{R}_2'} (W^j))$. The existence of $\Wcal$ follows from \autoref{@prohibitions} and the fact that $r_2'\geq \funref{@idealistic}(l+1,r_2,p)$.\smallskip

The fact that the $p$-internal bags
of $\tilde\Qcal$ that contain vertices adjacent to some vertex of $(S\cup B)\setminus \tilde{A}'$ are at most $l$ implies that
there exists an $i\in[l+1]$ such that
no vertex of $V(\bigcup \influence_{\mathfrak{R}_2'}({W^i}))$ is adjacent, in $G$, to a vertex in $(S\cup B)\setminus \tilde{A}'$.
Let $W_2:=W^i$, let $W_2^*$ be the central $(r_2-2)$-subwall of $W_2$, and let $D_{W_2}$ be the graph obtained from $G$ by removing the perimeter of $W_2$ and taking the connected component that contains $W_2^*$.
Since no vertex of $V(\bigcup \influence_{\mathfrak{R}_2'}({W^i}))$ is adjacent, in $G$, to a vertex in $(S\cup B)\setminus \tilde{A}'$,
any path in $D_{W_2}$ going from a vertex of $W_2^*$ to a vertex in $S$ must intersect a vertex of $\tilde{A}'$.
Thus, there is no model of $K_{s_\Fcal+a_\Fcal-1}$ grasped by $W_2^*$ in $D_{W_2}$, because otherwise, $K_{s_\Fcal}$ would be a minor of $G\setminus S$.
So, by applying the algorithm {\tt Grasped-or-Flat} of \autoref{@possession} with input $(D_{W_2},r_3,s_\Fcal+a_\Fcal-1,W_2^*)$, since $r_2-2\geq\funref{@classifications}(s_\Fcal+a_\Fcal-1) \cdot r_3$, we should find a set $A\subseteq V(D_{W_2})$ with $|A|\leq a$ and a flatness pair $(W_3,\mathfrak{R}_3)$ of $D_{W_2}\setminus A$ of height $r_3$, such that $W_3$ is a tilt of some subwall $\tilde{W}_3$ of $W_2$.\smallskip

Let $\tilde{\Qcal}'$ be a $W_3$-canonical partition of $D_{W_2}\setminus A$.
Let $\Wcal'=\{W^1,...,W^{a_\Fcal}\}$ be a collection of $r_4$-subwalls of $W_3$ such that
for every $i\in[a_\Fcal]$,
$\bigcup \influence_{\mathfrak{R}_3}(W^i)$ is a subgraph of $\bigcup \{Q\mid Q \text{ is an internal bag of }\tilde\Qcal'\}$ and
 for every $i,j\in[a_\Fcal]$, with $i\neq j$, there is no internal bag of $\tilde\Qcal'$ that contains vertices of both $V(\bigcup \influence_{\mathfrak{R}_3}(W^i))$ and $V(\bigcup \influence_{\mathfrak{R}_3}(W^j))$.
Since $|\tilde{A}'|<a_\Fcal$, there is an $i\in[a_\Fcal]$ such that $V(\bigcup \influence_{\mathfrak{R}_3}(W^i))$ does not intersect $\tilde{A}'$.
The existence of $\Wcal'$ follows from \autoref{@prohibitions} and the fact that $r_3\geq \funref{@idealistic}(a_\Fcal,r_4,1)$.\smallskip

Let $W_4:=W^i$.
Let $W_4^*$ be the central $(r_4-2)$-subwall of $W_4$ and let $D_{W_4}$ be the graph obtained from $D_{W_2}$ after removing $A$ and the perimeter of $W_4$ and taking the connected component containing $W_4^*$.
Observe that any path between a vertex of $S$ and a vertex of $V(\bigcup \influence_{\mathfrak{R}_3}(W_4))$ in $D_{W_2}$ intersects $\tilde{A}'$.
Since $\tilde{A}'$ does not intersect $V(\bigcup \influence_{\mathfrak{R}_3}(W_4))$, it implies that $\tilde{A}'$ does not intersect $D_{W_4}$, and thus $S\cap D_{W_4}=\emptyset$.
Therefore, $D_{W_4}$ is a minor of $G\setminus S$ and $K_{s_\Fcal}$ is not a minor of $D_{W_4}$.
Moreover, $W_4^*$ is a wall of $D_{W_4}$ of height $r_4-2\geq t+1$, so $\tw(D_{W_4})>t=\funref{@corollaries}(s_\Fcal)\cdot r_5$.
Therefore, by applying the algorithm {\tt Clique-or-twFlat} of \autoref{@unimportant} with input $(D_{W_4},r_5,s_\Fcal)$, we should obtain a set $A'$ of size at most $b$ and a regular flatness pair $(W_5,\mathfrak{R}_5)$ of $D_{W_4}\setminus A'$ of height $r_5$ whose $\mathfrak{R}_5$-compass has treewidth at most $t$.
All this is checked in Step~1, and thus, the algorithm should run Step~2.\medskip

If $|\tilde{A}'|\geq a_\Fcal$, then, due to \autoref{@proclamation} and the fact that $r_2'\geq \funref{@unaffected}(a_\Fcal,s_\Fcal,k)$, for any set $X\subseteq V(G)$ such that $\bid_{G\setminus (S\cup B),W_2'}(X)\leq k$ and such that $G\setminus X\in\exc(\Fcal)$, it holds that $X\cap\tilde{A}'\neq\emptyset$. In particular, for any $k$-apex set $S'$, $S'\cap\tilde{A}'\neq\emptyset$ due to \autoref{@psychoanalyse}.
Thus, there is a vertex $y\in\tilde{A}'$ such that $(G\setminus y,k-1)$ is a \yes-instance.
Hence, if the algorithm runs Step~3, it finds a vertex $y\in\tilde{A}'$ such that $(G\setminus y,k-1)$ is a \yes-instance.\medskip

Note that the enhancement $\tilde{\Qcal}$ of the canonical partition $\Qcal$ is not unique. In particular, $\tilde{A}'$ depends on $\tilde{\Qcal}$.
However, as long as there is such a $\tilde{\Qcal}$ such that $|\tilde{A}'|<a_\Fcal$, the algorithm finds the wanted flatness pair $(W_4,\mathfrak{R}_4)$ in Step~1 and then runs Step~2.
Hence, if $(G,k)$ is a \yes-instance, the algorithm runs Step~3 only if for all such $\tilde{A}'$, $|\tilde{A}'|\geq a_\Fcal$. {Note that, since $|\tilde{A}|\geq |\tilde{A}'|$, in this case we have that, for all such $\tilde{A}'$, $|\tilde{A}|\geq a_\Fcal$.}
This justifies the arbitrary canonical partition enhancement in Step~3 {and the fact
that, if $|\tilde{A}|<a_\Fcal$ in Step~3, then the algorithm reports a \no-instance.}
\medskip

{Let us now show the correctness of Step~2, and for this we do not suppose anymore that $(G,k)$ is a \yes-instance since the argument is the same for both types of instances.}
Suppose that the algorithm finds in Step~1 a set $A'$ of size at most $b$ and a regular flatness pair $(W_5,\mathfrak{R}_5)$ of $D_{W_4}\setminus A'$ of height $r_5$ whose $\mathfrak{R}_5$-compass has treewidth at most $t$.
We obtain a 7-tuple $\mathfrak{R}_5'$ by adding all vertices of $G\setminus V({\sf Compass}_{\mathfrak{R}_5}(W_5))$ to the set in the first coordinate of $\mathfrak{R}_5$. Since $(W_5,\mathfrak{R}_5)$ is a regular flatness pair of $D_{W_4}\setminus A'$ and since the vertices added in $\mathfrak{R}_5'$ are either in $A$, or adjacent at most to the perimeter of $W_4$, then $(W_5,\mathfrak{R}_5')$ is a regular flatness pair of $G\setminus (A\cup A')$. Since ${\sf Compass}_{\mathfrak{R}_5}(W_5) = {\sf Compass}_{\mathfrak{R}_5'}(W_5)$, ${\sf Compass}_{\mathfrak{R}_5'}(W_5)$ has treewidth at most $t$.
Thus, if we apply the algorithm {\tt Homogeneous} of \autoref{@disreputable} with input $(r_6,a+b,a+b,d,t,G,A\cup A',W_5,\mathfrak{R}_5')$ we obtain a flatness pair $(W_6,\mathfrak{R}_6)$ of $G\setminus (A\cup A')$ of height $r_6$ that is $d$-homogeneous with respect to $2^{A\cup A'}$ and is a $W^*$-tilt of $(W_5,\mathfrak{R}_5')$ for some subwall $W^*$ of $W_5$.
Due to \autoref{@successively}, we know that $(W_6,\mathfrak{R}_6)$ is regular.
Since $|A\cup A'|\leq a+b$, for any set $X\subseteq V(G)$, $|A\setminus X|\leq a+b$.
Since $G\setminus S\in\exc(\Fcal)$ and \autoref{@psychoanalyse} implies that $\bid_{G\setminus (A\cup A'),W_6}(S)\leq k$,
by applying the algorithm {\tt Find-Irrelevant Vertex} of \autoref{@civilizing} with input $(k,a+b,G,A\cup A',W_6,\mathfrak{R}_6)$, we obtain a vertex $v$ such that $G\setminus S\in\exc(\Fcal)$ if and only if $G\setminus (S\setminus v)\in\exc(\Fcal)$.
It follows that $(G,k)$ and $(G\setminus v,k)$ are indeed equivalent instances of {\sc $\Fcal$-M-Deletion}.
\medskip

We now suppose that $(G,k)$ is a \no-instance.
In the beginning of Step~1, the algorithm either reports a \no-instance or finds a wall.
In the latter case, the algorithm either goes to Step~2 or Step~3.
If it runs Step~2, the previous paragraph justifies that the algorithm finds a vertex $v$ such that $(G\setminus v,k)$ is a \no-instance.
If the algorithm runs Step~3, then it either reports a \no-instance or recursively runs on instances $(G\setminus y,k-1)$.
If $(G\setminus y,k-1)$ is \yes-instance, then so is $(G,k)$. Thus, $(G\setminus y,k-1)$ is a \no-instance for every considered vertex $y$ and the algorithm always reports a \no-instance.
Hence, \autoref{@uncritical} follows.

\section{Solving {\sc \texorpdfstring{$\Fcal$}{F}-M-Elimination Distance} on tree decompositions}\label{@fanaticism}

In the rest of the paper (except for \autoref{@entwickltmg}) we focus on {\sc $\Fcal$-M-Elimination Distance}.
In order to design an algorithm for this problem and prove \autoref{@communicated} and \autoref{@swineherds}
, we follow the same scheme as for  {\sc $\Fcal$-M-Deletion}.
The first step of this strategy is to present a dynamic programming algorithm that will allow us to solve {\sc $\Fcal$-M-Elimination Distance} for instances of bounded treewidth in \FPT-time, namely \autoref{@achievements}.
This is the analogue of for {\sc $\Fcal$-M-Elimination Distance} of the corresponding result for  {\sc $\Fcal$-M-Deletion}, namely \autoref{@calculated}. The following theorem is a reformulation of \autoref{@achievements}.

\begin{theorem}\label{@unquestioned}
For every finite collection of graphs $\Fcal$, there exists an algorithm that, given a graph $G$ of treewidth at most $\tw$ and a non-negative integer $k$, decides whether $\ed_{\exc(\Fcal)}(G)\leq k$ in time $2^{\Ocal_{\ell_\Fcal}(\tw\cdot k+\tw\log \tw)}\cdot n$.
\end{theorem}

According to \autoref{@theosophical}, $\td(G) \leq \tw(G)\cdot\log n$ for any graph $G$.
Since $\ed_{\exc(\Fcal)}(G)\leq \td(G)\leq \tw(G)\cdot\log n$, \autoref{@unquestioned} implies the existence of an \XP-algorithm
for {\sc $\Fcal$-M-Elimination Distance} parameterized by treewidth.

\begin{corollary}\label{@unwittingly}
For every finite collection of graphs $\Fcal$, there exists an algorithm
that,
given a graph $G$ of treewidth at most $\tw$, computes $\ed_{\exc(\Fcal)}(G)$ in time $n^{\Ocal_{\ell_\Fcal}(\tw^2)}$.
\end{corollary}

According to \autoref{@theosophical} again, $\tw(G) \leq \td(G)$ for any graph $G$.
Since we moreover have $\ed_{\exc(\Fcal)}(G)\leq \td(G)$, \autoref{@unquestioned} implies the existence of an \FPT-algorithm
for {\sc $\Fcal$-M-Elimination Distance} parameterized by treedepth.

\begin{corollary}\label{@mysterious}
For every finite collection of graphs $\Fcal$, there exists an algorithm
that,
given a graph $G$ of treedepth at most $\td$, computes $\ed_{\exc(\Fcal)}(G)$ in time $2^{\Ocal_{\ell_\Fcal}(\td^2)}\cdot n$.
\end{corollary}

Our algorithm takes inspiration from the dynamic programming algorithm of Reidl,  Rossmanith,  Villaamil, and  Sikdar~\cite{ReidlRSS14afas} for treedepth.

\begin{proposition}[\cite{ReidlRSS14afas}]\label{@polytheism}
Given a graph $G$, a tree decomposition of $G$ of width $w$, and an integer $k$, there is an algorithm that decides whether $\td(G)\leq k$ in time $2^{\Ocal(k\cdot w)}\cdot n$.
\end{proposition}

If $\Fcal=\{K_1\}$, elimination distance reduces to treedepth. Recall that {\sc $\{K_{1}\}$-M-Elimination Distance} is the problem asking whether $\td(G)\leq k$, which admits an algorithm in time $2^{\Ocal(k\cdot w)}\cdot n$ because of \autoref{@polytheism}.
Therefore, we may assume throughout this section that $\Fcal$ is non-trivial in order to have the useful property that a graph with a single vertex belongs to $\Gcal=\exc(\Fcal)$. This will simplify the algorithm.
Moreover, the elimination distance to $\exc(\Fcal)$ of a disconnected graph is the maximum of the elimination distance of its connected components, and therefore, we may assume that the considered graphs and boundaried graphs are connected.\smallskip

\sugar{Compared to the approach of~\cite{ReidlRSS14afas}, in order to deal with a general minor-closed graph class $\Gcal$, we use the framework introduced in~\cite{BasteST20acom} based on the notion of representatives of an appropriately defined equivalence relation on boundaried graphs. Intuitively, since the ``leaves'' of the desired elimination tree are graphs in the minor-closed family $\Gcal$, it will be possible to encode those graphs via their corresponding representatives. The fact that the boundary size that we need to consider is bounded follows from the description of the algorithm and its analysis.}

\sugar{In order to describe our dynamic programming algorithm, we have to describe its corresponding tables,
encode ``partial elimination sets'', and show how to calculate this information using a nice tree decomposition of the input graph.
For this reason, in  \autoref{@significant} we start by giving some additional notations on functions on sets and in
\autoref{@transposition} we define annotated trees. Annotated trees are labeled rooted trees that come together
with a boundaried graph such that the annotated nodes of the tree are mapped to the
vertices of the boundaried graph with the same label.
This notion is used in \autoref{@philosophers} in order to define the {\sl characteristic}
of a boundaried graph, which intuitively encodes how partial elimination trees can be present inside the boundaried graph.
Forget, introduce, and join procedures that shall be used in the dynamic program on nice tree decompositions
are presented in \autoref{@psychologi}.
In \autoref{@philanthropy}, we present the dynamic program and prove its correctness.
We conclude this section with \autoref{@metropoltheater}, where we show that boundaried graphs
with the same characteristic can be exchanged, i.e., give graphs of the same elimination distance to $\Fcal$ when ``glued'' to the same boundaried graph. This latter result will also be used in \autoref{@participant}.
}

\subsection{Some additional notation}\label{@significant}

We denote by $\im(f)$ the image of a function $f$ and by $\Ker(f)$ its kernel, i.e., the elements whose image by $f$ is 0.
Given two sets $A$ and $B$, two subsets $A'\subseteq A$ and $B'\subseteq B$, a function $f:A'\to B'$, $a\in A$, and $b\in B$, $f\oplus[a\mapsto b]$ is the function that maps $a$ to $b$ and every $a'\in A'\setminus\{a\}$ to $f(a')$.
If $f:A\to B$, we denote by $f|_{A'}$ the restriction of $f$ to $A'$.
When $f$ is a bijection, $f|_{A'}:A'\to\im(A')$ is seen as a bijection.
$(i\leftrightarrow j)$ denotes the transposition of $i$ and $j$, for $i,j$ in some set $I\subseteq\bN$.

\subsection{Annotated trees}\label{@transposition}

We proceed to define annotated trees, which we will use to codify the tables of our dynamic program.

\subparagraph{Annotated trees.} An \emph{annotated tree} is a tuple $\hat{T}=(T,r,h,{\bf R},f)$, where $(T,r)$ is a rooted tree, $h:V(T)\to\bN$, ${\bf R}=(R,B,\phi)$ is a boundaried graph, and $f:[|B|]\to V(T)$.
See \autoref{@comprehends} for an illustration of an annotated tree.
We stress that different integers in $[|B|]$ can be mapped, via $f$, to the same node of $T$.
The \emph{trivial annotated tree}, denoted by $\hat{\un}$, is $(T,r,h,\un,f)$ where $T$ is the rooted tree with a single node $r$,
$h$ is the constant function 0, $\un$ is the boundaried graph with one single vertex that is also part of the boundary, and $f$ maps 1 to $r$.
The \emph{height} of an annotated tree is $h(r)$.
Given an annotated tree $\hat{T}=(T,r,h,(R,B,\phi),f)$, we refer to $(T,r)$ as its {\em rooted tree}.
Given an annotated tree $\hat{T}=(T,r,h,(R,B,\phi),f)$ and a permutation $\sigma$ of $[|B|]$, we use $\sigma(\hat{T})$ to denote $(T,r,h,(R,B,\sigma\circ\phi),f\circ\sigma^{-1})$.

\begin{figure}[ht]
	\centering
\includegraphics[width=0.7\textwidth]{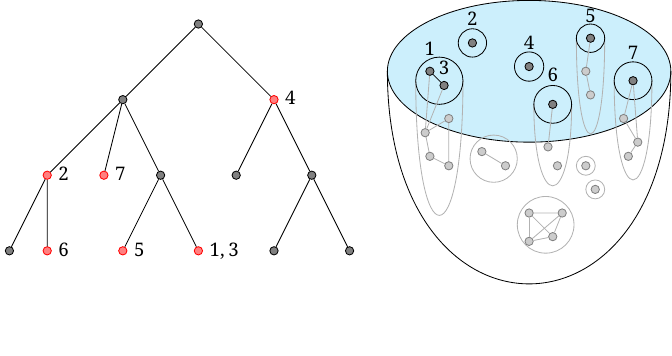}
\vspace{-1cm}

\caption{An annotated tree made of a rooted tree (left) and a boundaried graph (right).
The numbers in the left figure correspond to the pre-images of $f$ for the nodes of $V(T)$ and the numbers on the right figure correspond to the images of $\phi$.
The function $h$ that gives a value to each node of the tree is not represented.
}
\label{@comprehends}
\end{figure}

We also define the following operations on annotated trees, \sugar{which will be used to combine the tables of the dynamic programming algorithm. The first one is inspired by a similar operation introduced in~\cite{ReidlRSS14afas}.}

\subparagraph{Crop operation.}
Given an annotated tree $\hat{T}=(T,r,h,(R,B,\phi),f)$, the \emph{crop operation}, denoted by $\crop(\hat{T})$, outputs the annotated tree obtained from $\hat{T}$ by iteratively removing the leaves of $T$ that are not in $\im(f)$.
Given a set $\Acal$ of annotated trees, $\crop(\Acal):=\bigcup_{\hat{T}\in \Acal}\crop(\hat{T})$.

\subparagraph{Representation operation.}
Given an annotated tree $\hat{T}=(T,r,h,(R,B,\phi),f)$, the \emph{representation operation}, denoted by $\rep(\hat{T})$, outputs the annotated tree $\hat{T}'=(T,r,h,(R',B,\phi),f)$ constructed as follows.
For each $v\in \im(f)$, let $B_v:=\phi^{-1}\circ f^{-1}(v)$, let $\sigma_v:f^{-1}(v)\to[|B_v|]$ be a bijective function, and let $R_v$ be the union of the connected components of $R$ containing $B_v$.
If there is a node $v\in \im(f)$ such that $R_v\notin\exc(\Fcal)$, then $R':=K_{s_\Fcal}$.
Otherwise, for each $v\in \im(f)$, let $(R_v',B_v,\sigma_v\circ\phi|_{B_v})\in\Rcal_{\ell_\Fcal}^{|B_v|}$ be the representative of $(R_v,B_v,\sigma_v\circ\phi|_{B_v})$ for the equivalence relation $\equiv_{\ell_\Fcal}$.
Then $R'=\bigcup_{v\in \im(f)}R_v'$.
Intuitively, if there is a node $v\in \im(f)$ such that $R_v\notin\exc(\Fcal)$, to store this information it suffices to set $R':=K_{s_\Fcal}$, while otherwise, we keep for each $R_v$ (in fact, for the boundaried version of $R_v$) its representative.

An example of the crop and representation operation is give in \autoref{@unconscious}.
Observe that $\rep(\hat{\un})=\hat{\un}$ since this is a minimum-sized representative and since we make the assumption that $\Fcal$ is non-trivial.
Given a set $\Acal$ of annotated trees, $\rep(\Acal):=\bigcup_{\hat{T}\in \Acal}\rep(\hat{T})$.

\begin{figure}[ht]
	\centering
\includegraphics[width=0.6\textwidth]{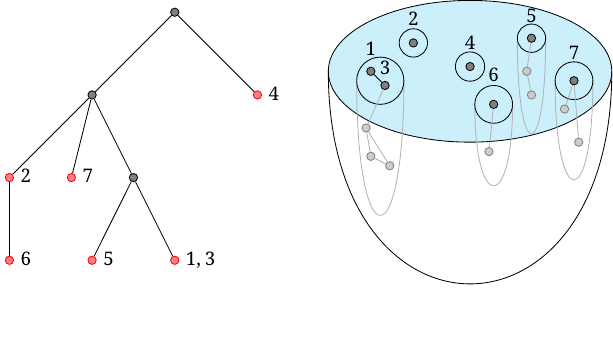}
\vspace{-1cm}

	\caption{The crop and representation operation applied to the annotated tree of \autoref{@comprehends}. The unlabeled leaves of the tree are iteratively removed. A representative of each component attached to the boundary of the boundaried graph is kept.
	}
	\label{@unconscious}
\end{figure}

\subparagraph{Filter operation.}
Given a set $\Acal$ of annotated trees and a positive integer $k$, the \emph{filter operation}, denoted by ${\sf filter}_k$,
outputs the set of annotated trees in $\Acal$ with height at most $k$.\smallskip

Note that the crop, representation, and filter operations are commutative since they do not modify the same objects.
For more simplicity, we define $\Mcal_k={\sf filter}_k\circ\rep\circ\crop$.
We stress that $\Mcal_k$ is an operation acting on {\sl sets} of annotated trees.

\subsection{Characteristic of a boundaried graph}\label{@philosophers}

In this subsection we define the {\sl characteristic} of a boundaried graph that shall be computed by the dynamic program in \autoref{@philanthropy}. This characteristic will consist of a set of annotated trees with some additional properties.
In order to present this definition, we first define the complete characteristic of a boundaried graph, that is a slightly more complicated way to see $\Fcal$-elimination trees with some distinguished nodes.

\subparagraph{Complete characteristic of a boundaried graph.}
Given a connected boundaried graph ${\bf G}=(G,X,\rho)$, the \emph{complete characteristic} of ${\bf G}$, denoted by $\fullchar({\bf G})$, is the set of annotated trees $\hat{T}=(T,r,h,{\bf R},f)$ such that
\begin{itemize}
	\item $|\im(f)|=|X|$,
	\item there exists a function $\chi:V(T)\to 2^{V(G)}$ such that $(T,\chi,r)$ is an $\Fcal$-elimination tree of $G$ and for $x\in X$, $x\in\chi\circ f\circ\rho(x)$,
	\item there exists an isomorphism $\sigma$ between ${\bf R}$ and $(\bigcup_{v\in\im(f)} G[\chi(v)],X,\rho)$, and
	\item $h$ is the height function $\height_{T,r}$.
\end{itemize}

$(\chi,\sigma)$ is called the \emph{witness pair} of $\hat{T}$ with respect to ${\bf G}$.
Since $(T,\chi,r)$ is an $\Fcal$-elimination tree, it is straightforward to see that for any boundaried graph ${\bf G}$ with underlying graph $G$, the minimum height of an annotated tree in $\fullchar({\bf G})$ is $\ed_{\exc(\Fcal)}(G)$.

\subparagraph{Characteristic of a boundaried graph.}
Let ${\bf G}=(G,X,\rho)$ be a boundaried graph and $k$ be an integer. The \emph{characteristic} of ${\bf G}$, denoted by $\Char_k({\bf G})$, is the set $\Mcal_k(\fullchar({\bf G}))$.

\begin{lemma}\label{@ineluctable}
Given a boundaried graph ${\bf G}=(G,X,\rho)$ with $X\neq\emptyset$ and an integer $k$, the elimination distance of $G$ to $\exc(\Fcal)$ is the minimum height of an annotated tree in $\Char_k({\bf G})$ if $\ed_{\exc(\Fcal)}(G)\leq k$, and $\Char_k({\bf G})=\emptyset$ otherwise.
\end{lemma}

\begin{proof}
Let ${\bf G}=(G,X,\rho)$ be a boundaried graph with $X\neq\emptyset$ and let $\hat{T}\in \fullchar({\bf G})$.
Since $X\neq\emptyset$, the crop operation will not remove the root of the underlying tree of $\hat{T}$.
Moreover, neither the crop operation nor the representation operation change the height of the nodes that stay in the tree.
So the height of $\rep\circ \crop(\hat{T})$ is equal to the height of $\hat{T}$.
\end{proof}

We now prove that the size of the characteristic of a boundaried graph is upper-bounded by a function of its boundary size and $k$.

\begin{lemma}\label{@belligerent}
There exists a function $\newfun{@caricatures}:\bN^2\to\bN$ such that,
given two integers $k$ and $w$, if ${\bf G}=(G,X,\rho)$ is a boundaried graph with $|X|\leq w$, then $|\Char_k({\bf G})|\leq\funref{@caricatures}(w,k)$.
Moreover, $\funref{@caricatures}(w,k)=2^{\Ocal_{\ell_\Fcal}(w\cdot k+w\log w)}$.
\end{lemma}

\begin{proof}
Let $\hat{T}=(T,r,h,{\bf R},f)$ be an annotated tree in $\Char_k({\bf G})$. Let $l:=|X|$.
Let $x_1,...,x_l$ be an ordering of the nodes in $\im(f)$ such that if $x_i\in \anc_{T,r}(x_j)$, then $i<j$.
Without loss of generality, we suppose that $f(i)=x_i$ for $i\in[l]$ (this is true up to a permutation of $[l]$).
Let $f_i$ be the restriction of $f$ to $[i]$, let $T_i$ be the tree obtained from $T$ by iteratively removing the leaves not in $\im(f_i)$, and let $h_i$ be the restriction of $h$ to $V(T_i)$.
Note that $\im(h_i)\subseteq[0,k]$ and $T_i$ is a tree with at most $i$ leaves because the leaves of $T$ are in $\im(f_i)$ and $|\im(f_i)|=i$.
So $T_i$ has at most $i\cdot (k+1)$ nodes.

We set $(T_0,h_0,f_0)$ to be the empty triple.
Let us bound the number of triples $(T_i,h_i,f_i)$ that can be constructed from $(T_{i-1},h_{i-1},f_{i-1})$ for $i\in[l]$.
For $i\in [l]$, we can construct $(T_i,h_i,f_i)$ from $(T_{i-1},h_{i-1},f_{i-1})$ by choosing a node of $T_{i-1}$ (if it exists) and adding a path of length at most $k$ with leaf $x_{i}$ (all nodes in this path are new, except from $x_i$).
We consider the function $h_i$ that has the same values as $h_{i-1}$ on $V(T_{i-1})$ and values in $[0,k]$ on the new path such that the value of $h_{i}$ strictly increases from a leaf to the root $r$.
Observe that the value of $h_{i}$ on the new path is a subset of $2^{[k+1]}$.
Therefore,
the number of different triples $(T,h,f)$ is at most
\[ \prod_{i=1}^w i(k+1)2^{k+1}=w!(k+1)^w2^{w(k+1)}\leq 2^{w\log w+w\log (k+1)+w(k+1)}.\]
Since  $\hat{T}=(T,r,h,{\bf R},f)$ is an annotated tree in $\Char_k({\bf G})$, it holds that ${\bf R}$ is the union of $l$ representatives ${\bf R}_i\in\Rcal_{\ell_\Fcal}^{w_i}$ for $i\in[l]$ where $l=|\im(f)|$, such that $\sum_{i=1}^l w_i=l\leq w$.
By \autoref{@encounters}, $|\Rcal_{\ell_\Fcal}^{w_i}|=2^{\Ocal_{\ell_\Fcal}(w_i \log w_i)}$.
So the number of ways to construct ${\bf R}$ is bounded by
\[\prod_{i=1}^l 2^{\Ocal_{\ell_\Fcal}(w_i \log w_i)}=2^{\Ocal_{\ell_\Fcal}(\sum_{i=1}^l w_i \log w)}=2^{\Ocal_{\ell_\Fcal}(w \log w)}.\]
Hence, we obtain the desired result.
\end{proof}

\subsection{The procedures}\label{@psychologi}

We define here the procedures that will be used in the dynamic programming algorithm.
Given a nice tree decomposition $\Tcal$ of a graph $G$, we want to define forget, introduce, and join procedures to obtain the characteristic of $G_v$ for each internal node $v$ in $\Tcal$, given the characteristics of $G_{v'}$ for each child $v'$ of $v$.
Before defining the procedures for the characteristics, we define the procedures for the complete characteristics.

\subsubsection{Forget procedure}

With the forget procedure, given the characteristic of a boundaried graph, we want to compute the characteristic of the boundaried graph obtained by removing a vertex from the boundary.

\subparagraph{Complete forget procedure.}
The complete forget procedure applied on the annotated tree $(T,r,h,{\bf R},f)$ corresponds to removing the vertex with the largest label from the boundary of ${\bf R}$.
More formally, given an annotated tree $\hat{T}=(T,r,h,(R,B,\phi),f)$, the \emph{complete forget procedure}, denoted by $\forget^*(\hat{T})$, outputs the annotated tree $\hat{T}'=(T,r,h,(R,B',\phi|_{B'}),f|_{[|B'|]}))$, where $B'=B\setminus\phi^{-1}(|B|)$.
Given a set $\Acal$ of annotated trees, $\forget^*(\Acal):=\bigcup_{\hat{T}\in \Acal}\forget^*(\hat{T})$.

\begin{lemma}\label{@nationalistic}
Let $G$ be a graph, $({\sf T},\beta,{\sf r})$
be a nice tree decomposition of $G$, $v$ be a forget node of ${\sf T}$ with child $v'$ and forgotten vertex $x$, $\rho:\beta(v)\to[|\beta(v)|]$ be a bijection, and $\rho':=\rho\oplus[x\mapsto|\beta(v')|]$.
Let $\Acal:=\fullchar(G_v,\beta(v),\rho)$ and $\Acal':=\fullchar(G_{v'},\beta(v'),\rho')$.
Then $\Acal=\forget^*(\Acal')$.
\end{lemma}

\begin{proof}
Let $\hat{T}\in \forget^*(\Acal')$. There exists $\hat{T}'\in \Acal'$ with witness pair $(\chi',\sigma')$
such that $\forget^*(\hat{T}')=\hat{T}$.
Let $(T,r)$ (resp. $(T',r)$) be the rooted tree of $\hat{T}$ (resp. $\hat{T}'$).
Note that $\height_{T,r}=\height_{T',r}$.
Thus, it is easy to see that $(\chi',\sigma')$ also witnesses that $\hat{T}\in \Acal$.
Conversely, let $\hat{T}=(T,r,h,(R,X,\phi),f)\in \Acal$ with witness pair $(\chi,\sigma)$. Let $u:=\chi^{-1}(x)$, $w:=\sigma^{-1}(x)$, and $t=|X|+1$.
Let $\hat{T}':=(T,r,h,(R,X\cup\{w\},\phi\oplus[w\mapsto t]),f\oplus[t\mapsto u])$.
Then $\forget^*(\hat{T}')=\hat{T}$ and the pair $(\chi,\sigma)$ also witnesses that $\hat{T}'\in \Acal'$, so $\hat{T}\in \forget^*(\Acal')$.
\end{proof}

\subparagraph{Forget procedure.}
Given an annotated tree $\hat{T}$, the \emph{forget procedure}, denoted by $\forget(\hat{T})$, outputs $\rep\circ\crop\circ\forget^*(\hat{T})$.
See \autoref{@systematically} for an illustration.
Note that we do not apply the filter operation since the height does not change under the forget procedure.
Given a set $\Acal$ of annotated trees, $\forget(\Acal)$ outputs $\bigcup_{\hat{T}\in \Acal}\forget(\hat{T})$.

\begin{figure}[ht]
	\centering
\includegraphics[width=0.6\textwidth]{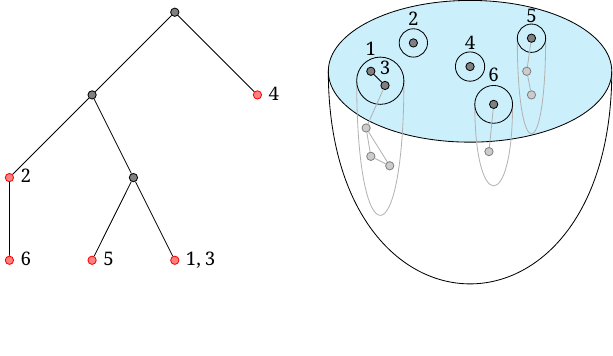}
\vspace{-1cm}

	\caption{The forget procedure applied to the annotated tree of \autoref{@comprehends}.
}
	\label{@systematically}
\end{figure}

\begin{lemma}\label{@socialization}
Let $G$ be a graph, $k$ be an integer, $({\sf T},\beta, {\sf r})$ be a nice tree decomposition of $G$, $v$ be a forget node of ${\sf T}$ with child $v'$ and forgotten vertex $x$, $\rho:\beta(v)\to[|\beta(v)|]$ be a bijection, and $\rho':=\rho\oplus[x\mapsto|\beta(v')|]$.
Let $\Acal:=\Char_k(G_v,\beta(v),\rho)$ and $\Acal':=\Char_k(G_{v'},\beta(v'),\rho')$.
Then $\Acal=\forget(\Acal')$.
\end{lemma}

\begin{proof}
Let $\Dcal:=\fullchar(G_v,\beta(v),\rho)$ and $\Dcal':=\fullchar(G_{v'},\beta(v'),\rho')$.
According to \autoref{@nationalistic}, $\Dcal=\forget^*(\Dcal')$.
The representation and the crop operation do not change the labels of an annotated tree, while the complete forget procedure only changes the labels of the input annotated tree and it does not change its height.
Thus, $\rep\circ\forget^*=\rep\circ\forget^*\circ\rep$, $\crop\circ\forget^*=\crop\circ\forget^*\circ\crop$, and ${\sf filter}_k\circ\forget^*=\forget^*\circ {\sf filter}_k$.
Since the three operations are commutative, it follows that $\Mcal_k\circ\forget^*=\rep\circ\crop\circ\forget^*\circ \Mcal_k=\forget\circ \Mcal_k$.
Hence, $\Acal=\Mcal_k(\Dcal)=\Mcal_k\circ\forget^*(\Dcal')=\forget(\Mcal_k(\Dcal'))=\forget(\Acal')$.
\end{proof}

\subsubsection{Introduce procedure}

With the introduce procedure, given the characteristic of a boundaried graph and a set $I$ of labels from the boundary, we want to compute the characteristic of the boundaried graph obtained by adding a new vertex to the boundary, which is adjacent to the nodes with a label in $I$.

\subparagraph{Diamond-introduce operation.}
Let $(T,r)$ be a rooted tree, $w$ be an integer, $f:[w]\to V(T)$ be a function, and $I$ be a subset of $[w]$.
$(T,r,f)\diamondsuit_{\sf intr} I$ is defined as the set of all pairs $(T',r',f')$ such that:
\begin{enumerate}
	\item $(T',r')$ is a rooted tree,
	\item $V(T')=V(T)\cup\{u\}$ for some new node $u$,
	\item $f'=f\oplus[w+1\mapsto u]$,
	\item if $v_1\in V(T)$ and $v_2\in \anc_{T,r}(v_1)$, then $v_2\in \anc_{T',r'}(v_1)$,
	\item if $v\in f(I)$, then $v\in \anc_{T,r}(u)\cup\desc_{T,r}(u)$, and
	\item $T_u'\cap f(I)\neq\emptyset$, or $u\in\leaf(T',r')$ and $\Par_{T',r'}(u)\in f(I)$.
\end{enumerate}

This operation corresponds to introducing a new node $u$ in $T$ so that $u$ has ancestor-descendant relations with the nodes labeled by a label in $I$.
The last item, which states that $u$ either has a descendant in $f(I)$ or is a leaf and its parent belongs to $f(I)$, is a property needed to ensure connectivity and allows the application of the crop operation in \autoref{@neopositivist} and \autoref{@philosopher}.

Let $(T,r)$ be a rooted tree, $K$ be a subset of $V(T)$, and $h:K\to\bN$.
We define the function $\update_{T,r,K}(h):V(T)\to\bN$,
that maps every $v\in V(T)$ to the integer
\[\update_{T,r,K}(h)(v) = \max\{h(v),1+\max_{c\in \ch(v)}\{\update_{T,r,K}(h)(c)\}\},\]
where we suppose that $h(v)=0$ if $v\notin K$ and $\update_{T,r,K}(h)(c)=-1$ if $v$ is childless.
Let $(T',r')$ be a rooted tree with $V(T')=K$ and such that the ancestor-descendant relationship between the nodes of $K$ is the same in $(T,r)$ and $(T',r')$.
Then we can observe that $\update_{T,r,K}(\height_{T',r'})=\height_{T,r}$.

\subparagraph{Complete introduce procedure.}
Let $\hat{T}=(T,r,h,(R,X,\phi),f)$ be an annotated tree and let $I$ be a set of labels in $[|X|]$.
The complete introduce procedure corresponds to adding a new vertex $v$ to the boundary $X$ of the boundaried graph $(R,X,\phi)$ of an annotated tree $(T,r,h,(R,X,\phi),f)$, such that $v$ is adjacent to the nodes with a label in $I$ and it is either mapped, via $f\circ\phi$, to an already existing leaf of $T$ (item (a) below) or to a new node of $T$ (item (b) below).

More formally, given an annotated tree $\hat{T}=(T,r,h,(R,X,\phi),f)$, a set $I\subseteq [|X|]$ of labels, the \emph{complete introduce procedure}, denoted by $\intr^*(\hat{T},I)$, outputs a set $\Acal$ of annotated trees constructed as follows.
For each $(T',r',f')\in (T,r,f)\diamondsuit_{\sf intr} I$, let $w=|\im(f')|$, $u:=f'(w)$, and $u':=\Par_{T',r'}(u)$ (or $u':=u$ if $u$ is the root).
We add a new vertex $v$ to $R$ and we set $X':=X\cup\{v\}$ and $\phi':=\phi\oplus[v\mapsto w]$.
Let $(R_{u'},X_{u'},\phi_{u'})$ be the part of the representative $R$ corresponding to $u'$, as defined for the representation operation (i.e., $R_{u'}$ is the union of the connected components of $R$ containing $\phi^{-1}\circ f^{-1}(u')$, $X_{u'}=\phi^{-1}\circ f^{-1}(u')$, $\phi_{u'}=\sigma_{u'}\circ\phi|_{X_{u'}}$ and $\sigma_{u'}:f^{-1}(u')\to[|X_{u'}|]$ is a bijective function).

\begin{itemize}
\item[(a)] \emph{If $h(u')=0$:}
We set $f'':=f\oplus[w\mapsto u']$
and $R':=(V(R)\cup\{v\},E(R)\cup E(\{v\},\phi^{-1}(I)\cap X_{u'}))$
and we add $(T,r,h,(R',X',\phi'),f'')$ to $\Acal$
if the connected component of $R'$ containing $v$ belongs to $\exc(\Fcal)$.

\item[(b)] \emph{If $u'\notin \im(f)$, or if $u'\in \im(f)$ and $|V(R_{u'})|=1$:}
We set $R':=(V(R)\cup\{v\},E(R))$
and $h':=\update_{T',r',V(T)}(h)$ and
we add $(T',r',h',(R',X,'\phi'),f')$ to $\Acal$.
\end{itemize}
Note that this is not a dichotomy: if both criteria are fulfilled, then we apply both cases.
Given a set $\Acal$ of annotated trees, we define $\intr^*(\Acal,I):=\bigcup_{\hat{T}\in \Acal}\intr^*(\hat{T},I)$.

\begin{lemma}\label{@neopositivist}
Let $G$ be a graph, $({\sf T},\beta,{\sf r})$ be a nice tree decomposition of $G$, $v$ be an introduce node of ${\sf T}$ with child $v'$ and introduced vertex $x$, $\rho':\beta(v)\to[|\beta(v)|]$ be a bijection, $\rho:=\rho'\oplus[x\mapsto|\beta(v)|]$, and $I:=\rho(N_{G_v}(x))$.
Let $\Acal:=\fullchar(G_v,\beta(v),\rho)$ and $\Acal':=\fullchar(G_{v'},\beta(v'),\rho')$.
Then $\Acal=\intr^*(\Acal',I)$.
\end{lemma}

\begin{proof}
Let $\hat{T}=(T,r,h,(R,X,\phi),f)\in \intr^*(\Acal',I)$.
There is $\hat{T}'=(T',r',h',{\bf R}',f')\in \Acal'$ with witness pair $(\chi',\sigma')$ such that $\hat{T}\in\intr^*(\hat{T'},I)$.
Let $\chi:=\chi'\oplus[u\mapsto x]$ and $\sigma=\sigma'\oplus[w\to x]$ where $u:=f(|X|)$ and $w:=\phi^{-1}(|X|)$.

Suppose that we are in case (a).
Thus, $h=h'=\height_{T',r'}=\height_{T,r}$.
By the construction of ${\bf R}$ from ${\bf R}'$, since $x$ is only adjacent in $G_v$ to the nodes with a label in $I$, $\sigma$ is an isomorphism between ${\bf R}$ and $(\bigcup_{w\in\im(f)} G_v[\chi(w)],\beta(v),\rho)$.
Moreover, $T=T'$ so $(T,\chi,r)$ is an $\Fcal$-elimination tree of $G_v$.
Hence, $(\chi,\sigma)$ witnesses that $\hat{T}\in \Acal$.

Suppose now that we are in case (b).
Notice that in this case $$h=\update_{T,r,V(T')}(h')=\update_{T,r,V(T')}(\height_{T',r'})=\height_{T,r}.$$
It is easy to see that $\sigma$ is an isomorphism between ${\bf R}$ and $(\bigcup_{w\in\im(f)} G_v[\chi(w)],\beta(v),\rho)$.
Moreover, as  $\hat{T}\in\intr^*(\hat{T'},I)$, it holds that $(T,r,f)\in (T',r',f')\diamondsuit_{\sf intr} I$ and therefore $(T,r)$ keeps the ancestor-descendant relations from $(T',r')$ (item 4 of the $\diamondsuit_{\sf intr}$ operation), while adding ancestor-descendant relations between the new node of $T$ and the nodes labeled by $I$ (item 5),
and guaranteeing the connectivity of $G_v[\chi(T_w)]$ for each node $w\in V(T)$ (item 6).
Hence, $(\chi,\sigma)$ witnesses that $\hat{T}\in \Acal$.

Conversely, let $\hat{T}=(T,r,h,{\bf R},f)\in \Acal$ with witness pair $(\chi,\sigma)$.
If $\chi(\chi^{-1}(x))=\{x\}$, then let $T'$ be the tree obtained from $T$ by removing $u:=\chi^{-1}(x)$ and adding edges between the parent of $u$ and the children of $u$, if any. Otherwise, set $T':=T$.
Let $h'$ be the height function of $T'$ and $r'$ be the root of $T'$.
Let $(R',X',\phi')$ be the boundaried graph obtained from ${\bf R}$ by removing $\sigma^{-1}(x)$.
Let $f':=f|_{X'}$.
Then the functions obtained from $\chi$ and $\sigma$ after restricting their image to $V(G)\setminus x$ witness that $\hat{T}'=(T',r',h',(R',X',\phi'),f')\in \Acal'$.
To show that $\hat{T}\in\intr^*(\hat{T'},I)$, the only non-trivial part is to prove that either $T_u\cap f(I)\neq\emptyset$, or $u\in\leaf(T,r)$ and $\Par_{T,r}(u)\in f(I)$ (item 6 of the $\diamondsuit_{\sf intr}$ operation).
Suppose that $T_u\cap f(I)=\emptyset$.
We know that $x$ is exactly adjacent to $\rho^{-1}(I)$ in $G_v$ since it is an introduce vertex.
Moreover, $(T,\chi,r)$ is an $\Fcal$-elimination tree of $G_v$, so $G_v[\chi(T_w)]$ is connected for all $w\in V(T)$.
In particular, $G_v[\chi(T_u)]$ is connected, but $x\in G_v[\chi(T_u)]$ and $\rho^{-1}(I)\notin G_v[\chi(T_u)]$ because $T_u\cap f(I)=\emptyset$. Thus, we must have $G_v[\chi(T_u)]=\{x\}$, and so $u\in\leaf(T,r)$.
Since $G_v[\chi(T_{\Par_{T,r}(u)})]$ is also connected, it implies that $\chi(\Par_{T,r}(u))$ and $x$ are connected, so $\Par_{T,r}(u)\in f(I)$.
Hence, $\hat{T}\in \intr^*(\Acal',I)$.
\end{proof}

\subparagraph{Introduce procedure.}
Given an annotated tree $\hat{T}$, a set $I$ of labels, and a positive integer $k$, the \emph{introduce procedure}, denoted by $\intr_k(\hat{T},I)$, outputs ${\sf filter}_k\circ\rep\circ\intr^*(\hat{T},I)$.
Examples of the introduce procedure can be found in \autoref{@paramilitary}.
Note that we do not apply the crop operation, since the complete introduction procedure applied to a cropped annotated tree outputs a cropped annotated tree.
Given a set $\Acal$ of annotated trees, $\intr_k(\Acal,I)$ outputs $\bigcup_{\hat{T}\in \Acal}\intr_k(\hat{T},I)$.

\begin{figure}[ht]
	\centering
\includegraphics[width=0.9\textwidth]{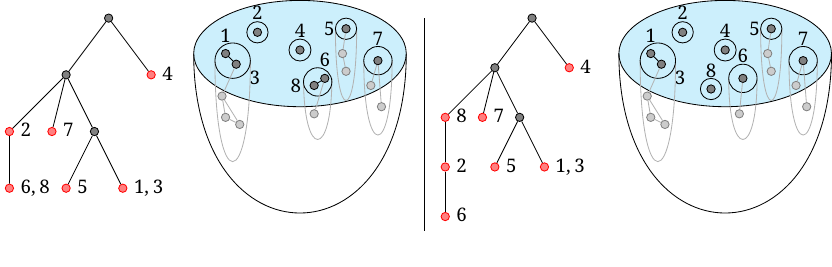}

	\caption{Two annotated trees obtained from the introduce procedure applied to the annotated tree of \autoref{@comprehends} with $I=\{2,6\}$.}
	\label{@paramilitary}
\end{figure}

\begin{lemma}\label{@philosopher}
Let $G$ be a graph, $({\sf T},\beta,{\sf r})$ be a nice tree decomposition of $G$, $v$ be an introduce node of ${\sf T}$ with child $v'$ and introduced vertex $x$, $\rho':\beta(v)\to[|\beta(v)|]$ be a bijection, $\rho:=\rho'\oplus[x\mapsto|\beta(v)|]$, and $I:=\rho(N_{G_v}(x))$.
Let $k$ be an integer, $\Acal:=\Char_k(G_v,\beta(v),\rho)$, and $\Acal':=\Char_k(G_{v'},\beta(v'),\rho')$.
Then $\Acal=\intr_k(\Acal',I)$.
\end{lemma}

\begin{proof}
Let $\Dcal:=\fullchar(G_v,\beta(v),\rho)$ and $\Dcal':=\fullchar(G_{v'},\beta(v'),\rho')$.
According to \autoref{@neopositivist}, $\Dcal=\intr^*(\Dcal',I)$.

The crop operation acts on the tree $T$ of an annotated tree $(T,r,h,{\bf R},f)$ to remove the leaves that are not in $\im(f)$.
By item 6 of the $\diamondsuit_{\sf intr}$ operation, the complete introduce procedure may only add a new node in $T$ above a node in $\im(f)$, or add a leaf to $T$ with a parent node in $\im(f)$.
Hence, the new node may only be added to the cropped tree.
Since the new vertex is also in the boundary, it implies that, given a set $\Ccal$ of annotated trees, $\crop\circ\intr^*(\Ccal,I)=\intr^*(\crop(\Ccal),I)$.

The representation operation acts on ${\bf R}$ to replace it by a boundaried graph whose connected components are representatives of the connected components of ${\bf R}$.
The complete introduce procedure adds a new vertex in the boundary of ${\bf R}$ that is adjacent to boundary vertices only.
Hence, $\rep\circ\intr^*(\Ccal,I)=\rep\circ\intr^*(\rep(\Ccal),I)$.

Moreover, the complete introduce procedure can only increase the height of an annotated tree.
Therefore, ${\sf filter}_k\circ\intr^*(\Ccal,I)={\sf filter}_k\circ\intr^*({\sf filter}_k(\Ccal),I)$.

Since the three operations are commutative, we have that $\Mcal_k\circ\intr^*(\Ccal,I)={\sf filter}_k\circ\rep\circ\intr^*(\Mcal_k(\Ccal),I)=\intr_k(\Mcal_k(\Ccal),I)$.
Therefore, $\Acal=\Mcal_k(\Dcal)=\Mcal_k\circ\intr^*(\Dcal',I)=\intr_k(\Mcal_k(\Dcal'),I)=\intr_k(\Acal',I)$.
\end{proof}

For the introduce procedure, we finally prove that it can generate a bounded number of annotated trees.

\begin{lemma}\label{@equivalent}
Let $w$ and $k$ be two positive integers, let ${\bf G}$ be a $w$-boundaried graph, let $I\subseteq[w]$, and let $\hat{T}\in\Char_k({\bf G})$.
Then $|\intr_k(\hat{T},I)|=\Ocal(w\cdot k)$.
\end{lemma}

\begin{proof}
Let $\hat{T}=(T,r,h,{\bf R},f)$.
Let us show that $|(T,r,f)\diamondsuit_{\sf intr} I|=\Ocal(w\cdot k)$.
Since $\leaf(T,r)\subseteq\im(f)$, $T$ has at most $w$ leaves.
Thus, $T$ has at most $w\cdot(k+1)$ nodes.
The $\diamondsuit_{\sf intr}$ operation consists in adding a new node to $T$, either as the new parent of a node, or as a new leaf of $T$.
Therefore, $|(T,r,f)\diamondsuit_{\sf intr} I|\leq 2w\cdot(k+1)$.

Let $(T',r',f')\in(T,r,f)\diamondsuit_{\sf intr} I$.
In the introduction procedure, we obtain at most two annotated trees from $(T',r',f')$.
Therefore, $|\intr_k(\hat{T},I)|\leq 4w\cdot(k+1)$.
\end{proof}

\subsubsection{Join procedure}

With the join procedure, given the characteristic of two boundaried graphs ${\bf G}_1$ and ${\bf G}_2$ that are compatible, we want to compute the characteristic of ${\bf G}_1\bigoplus{\bf G}_2$.

\subparagraph{Diamond-join operation.}
Let $(T_1,r_1)$ and $(T_2,r_2)$ be two rooted trees, $w$ be an integer, and $f_1:[w]\to V(T_1)$ and $f_2:[w]\to V(T_2)$ be two functions such that $\{f_1^{-1}(u)\mid u\in \im(f_1)\}=\{f_2^{-1}(u)\mid u\in \im(f_2)\}$.
Thanks to this equality, we can identify $f_1$ with $f_2$ and say that $V(T_1)\cap V(T_2)=\im(f_1)$.
We define $(T_1,r_1,f_1)\diamondsuit_{\sf join} (T_2,r_2,f_2)$ as the set of all pairs $(T,r,f)$ such that:
\begin{enumerate}
	\item $f=f_1=f_2$,
	\item $(T,r)$ is a rooted tree,
	\item $V(T)=V(T_1)\cup V(T_2)$,
	\item for $i\in\{1,2\}$, if $u\in V(T_i)$ and $v\in \anc_{T_i,r_i}(u)$, then $v\in \anc_{T,r}(u)$, and
	\item for every $v\in \leaf(T,r)$ and every $w\in\anc_{T,r}(v)$, if $V(vTw)\cap \im(f)=\emptyset$, then there is an $i\in\{1,2\}$ such that $vTw=vT_iw$.
\end{enumerate}

The last item, which states that a branch of $T$ that does not intersect $\im(f)$ either belongs to $T_1$ or $T_2$, is a property needed to ensure connectivity and allows the application of the crop operation in \autoref{@virtuously} and \autoref{@authenticity}.

Let $(T,r)$ be a rooted tree, and for $i\in\{1,2\}$, let $K_i\subseteq V(T)$ and $h_i:K_i\to\bN$.\\
$\update_{T,r,K_1,K_2}(h_1,h_2)$ is the function that maps $v\in V(T)$ to the maximum of 
$$\mbox{$h_1(v)$, $h_2(v)$, and $\max_{c\in \ch(v)}\{1+\update_{T,r,K_1,K_2}(h_1,h_2)(c)\}$,}$$
when they are defined.
For $i\in\{1,2\}$, let $(T_i,r_i)$ be a rooted tree with $V(T_i)=K_i$ and such that the ancestor-descendant relationship between the nodes of $K_i$ is the same in $T$ and $T_i$.
Then we can observe that $\update_{T,r,K_1,K_2}(\height_{T_1,r_1},\height_{T_2,r_2})=\height_{T,r}$.

\subparagraph{Complete join procedure.}
The join procedure corresponds to ``merging'' two annotated trees whose intersection is exactly their boundary.
More formally, given two annotated trees $\hat{T}_1=(T_1,r_1,h_1,{\bf R}_1,f_1)$ and $\hat{T}_2=(T_2,r_2,h_2,{\bf R}_2,f_2)$, the \emph{complete join procedure}, denoted by $\join(\hat{T}_1,\hat{T}_2)$, outputs a set $\Acal$ of annotated trees constructed as follows.
Initially, $\Acal$ is empty.
If ${\bf R}_1$ and ${\bf R}_2$ are compatible (i.e., such that $\{f_1^{-1}(u)\mid u\in \im(f_1)\}=\{f_2^{-1}(u)\mid u\in \im(f_2)\}$) and such that $\Ker(h_1\circ f_1)=\Ker(h_2\circ f_2)$, then for $(T,r,f)\in(T_1,r_1,f_1)\diamondsuit_{\sf join}(T_2,r_2,f_2)$, let $h:=\update_{T,r,V(T_1),V(T_2)}(h_1,h_2)$.
Let $(R,X,\rho):={\bf R}_1\bigoplus{\bf R}_2$. If each connected component of $R$ belongs to $\exc(\Fcal)$, then we add $(T,r,h,{\bf R}_1\bigoplus{\bf R}_2,f)$ to $\Acal$.
Given two sets $\Acal_1$ and $\Acal_2$ of annotated trees, we set $\join^*(\Acal_1,\Acal_2):=\bigcup_{\hat{T}_1\in \Acal_1,\hat{T}_2\in \Acal_2}\join^*(\hat{T}_1,\hat{T}_2)$.

\begin{lemma}\label{@virtuously}
Let $G$ be a graph, $({\sf T},\beta,{\sf r})$ be a nice tree decomposition of $G$, $v$ be a join node of ${\sf T}$ with children $v_1$ and $v_2$, and $\rho:\beta(v)\to[|\beta(v)|]$ be a bijection.
Let $\Acal:=\fullchar(G_v,\beta(v),\rho)$, $\Acal_1:=\fullchar(G_{v_1},\beta(v_1),\rho)$, and $\Acal_2:=\fullchar(G_{v_2},\beta(v_2),\rho)$.
Then $\Acal=\join^*(\Acal_1,\Acal_2)$.
\end{lemma}

\begin{proof}
Let $\hat{T}=(T,r,h,(R,X,\phi),f)\in \join^*(\Acal_1,\Acal_2)$.
There is $\hat{T}_1=(T_1,r_1,h_1,{\bf R}_1,f_1)\in \Acal_1$ and $\hat{T}_2=(T_2,r_2,h_2,{\bf R}_2,f_2)\in \Acal_2$ such that $\hat{T}\in\join^*(\hat{T}_1,\hat{T}_2)$ with witness pair $(\chi_1,\sigma_1)$ and $(\chi_2,\sigma_2)$, respectively, such that $\chi_{1|\im(f)}=\chi_{2|\im(f)}$ and $\sigma_{1|X}=\sigma_{2|X}$.
Note that $$h:=\update_{T,r,V(T_1),V(T_2)}(h_1,h_2)=\update_{T,r,V(T_1),V(T_2)}(\height_{T_1,r_1},\height_{T_2,r_2})=\height_{T,r}.$$
Let $(\chi,\sigma)=(\chi_1\cup\chi_2,\sigma_1\cup\sigma_2)$.

It is easy to see that $\sigma$ is an isomorphism between ${\bf R}$ and $(\bigcup_{w\in\im(f)}G[\chi(w)],\beta(v),\rho)$.
If $uv\in E(G_v)$, since $G_{v_1}\setminus\beta(v)$ and $G_{v_2}\setminus\beta(v)$ are not connected, then there is $i\in\{1,2\}$ such that $\chi_i(u)\in \anc_{T,r}(\chi_i(v))\cup\desc_{T,r}(\chi_i(v))$ holds, so $\chi(u)\in \anc_{T,r}(\chi(v))\cup\desc_{T,r}(\chi(v))$ due to item 4 of the $\diamondsuit_{\sf join}$ operation.

Moreover, item 5 of the $\diamondsuit_{\sf join}$ operation ensures the connectivity of $G_v[\chi(T_w)]$ for each $w\in V(T)$.
Indeed, let $i\in\{1,2\}$ be such that $w\in V(T_i)$. Suppose towards a contradiction that $G_v[\chi(T_w)]$ is not connected.
Note that it implies that $w\notin\im(f)$, because $\im(f)\subseteq V(T_1)\cup V(T_2)$ and $(T_1,\chi_1,r_1)$ and $(T_2,\chi_2,r_2)$ are $\Fcal$-elimination trees of $G_{v_1}$ and $G_{v_2}$, respectively, so $G_{v_1}[\chi_1((T_1)_w)]$ and $G_{v_2}[\chi_2((T_2)_w)]$ are connected and therefore $G_v[\chi(T_w)]$ would be connected.
We assume that $w$ is a minimal node such that $G_v[\chi(T_w)]$ is not connected, i.e., for every $u\in V(T_w)\setminus \{w\}$, $G_v[\chi(T_u)]$ is connected.
Hence, there is $u\in\ch_{T,r}(w)$ such that $\chi(w)$ is not connected to $G_v[\chi(T_u)]$.
So $V(T_u)\cup V(T_i)=\emptyset$, since otherwise, the connectivity of $G_{v_i}[\chi_i((T_i)_w)]$ would imply the connectivity of $\chi(w)$ with $G_v[\chi(T_u)]$.
Thus, there is an $x\in\leaf(T,r)$ such that $u\in\anc_{T,r}(x)$, and therefore $w\in\anc_{T,r}(x)$, and $V(xTw)\cap\im(f)=\emptyset$.
So, according to item 5 of the $\diamondsuit_{\sf join}$ operation, there is $i\in\{1,2\}$ such that $xTw=xT_iw$.
This contradicts the fact that $w\in V(T_i)\setminus\im(f)$ and $V(T_u)\subseteq V(T_j)\setminus\im(f)$ where $\{i,j\}=\{1,2\}$.
Thus, $(T,\chi,r)$ is an $\Fcal$-elimination tree of $G_v$.
Therefore, $(\chi,\sigma)$ witnesses that $\hat{T}\in \Acal$.

Conversely, let $\hat{T}\in \Acal$ with witness pair $(\chi,\sigma)$.
Let $(\chi_1,\sigma_1)$ and $(\chi_2,\sigma_2)$ be the co-restrictions of $(\chi,\sigma)$ to $G_{v_1}$ and $G_{v_2}$, respectively.

Let $T_1$ be the tree obtained from $T$ by removing the nodes not in $\im(\chi_1)$ and adding edges between the parent and children of each removed node.
Let $r_1$ be the root of $T_1$ and $h_1:=\height_{T_1,r_1}$.
Let ${\bf R}_1$ be obtained from ${\bf R}$ by removing the vertices not in $\im(\sigma_1)$.
Then it is easy to see that $\hat{T}_1=(T_1,r_1,h_1,{\bf R}_1,f)$ belongs to $B_1$ with witness pair $(\chi_1,\sigma_1)$.
We construct similarly $\hat{T}_2=(T_2,r_2,h_2,{\bf R}_2,f)\in B_2$ with witness pair $(\chi_2,\sigma_2)$.

Moreover, we claim that $\hat{T}\in \join^*(\hat{T}_1,\hat{T}_2)$.
To prove this claim, the less trivial part
is to show that $T$ respects item 5 of the $\diamondsuit_{\sf join}$ operation.
Let $x\in\leaf(T,r)$ and $w\in\anc_{T,r}(x)$ such that $V(xTw)\cap\im(f)=\emptyset$.
Suppose towards a contradiction that $V(xTw)\cap V(T_1)\neq\emptyset$ and $V(xTw)\cap V(T_2)\neq\emptyset$.
Thus, there exist $u_1\in V(xTw)\cap V(T_1)\setminus\im(f)$ and $u_2\in V(xTw)\cap V(T_2)\setminus\im(f)$.
Without loss of generality, suppose that $x\in V(T_1)$.
We take such a node $u_2$ that is  closest to $x$ and $u_1$ to be the node just before $u_2$ in the path from $x$ to $u_2$.
Since $V(xTu_1)\subseteq V(T_1)\setminus\im(f)$ and $u_2\in V(T_2)\setminus\im(f)$, $\chi(u_2)$ and $\chi(xTu_1)$ are not connected.
This contradicts the fact that $G_v[\chi(T_{u_2})]$ is connected.
Therefore, $\hat{T}\in \join^*(\hat{T}_1,\hat{T}_2)$, which implies that $\hat{T}\in \join^*(\Acal_1,\Acal_2)$.
\end{proof}

\subparagraph{Join procedure.}
Given two annotated trees $\hat{T}_1$ and $\hat{T}_2$, the \emph{join procedure}, denoted by $\join(\hat{T}_1,\hat{T}_2)$, outputs ${\sf filter}_k\circ\rep\circ\join^*(\hat{T}_1,\hat{T}_2)$.
See \autoref{@characterized} for an example.
Note that we do not apply the crop operation since joining two cropped annotated trees gives a cropped annotated tree.
Given two sets $\Acal_1$ and $\Acal_2$ of annotated trees, $\join(\Acal_1,\Acal_2)$ outputs $\bigcup_{\hat{T}_1\in \Acal_1,\hat{T}_2\in \Acal_2}\join(\hat{T}_1,\hat{T}_2)$.

\begin{figure}[ht]
	\centering
\includegraphics[width=0.8\textwidth]{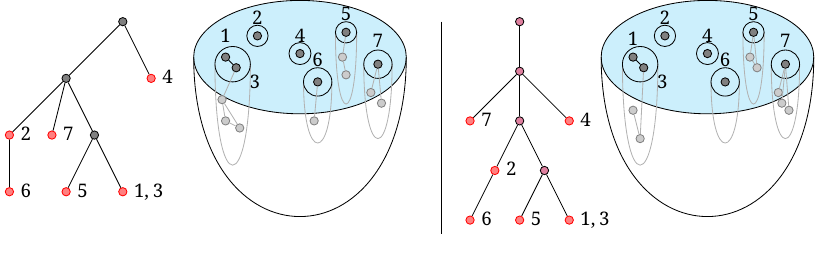}
\includegraphics[width=0.4\textwidth]{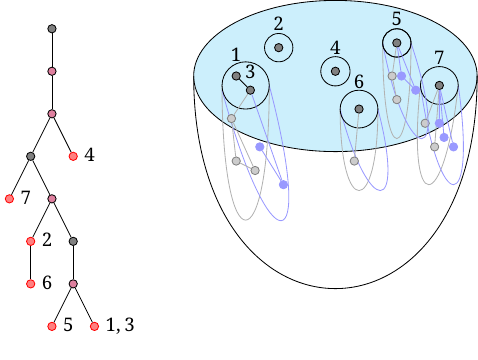}
	\caption{An annotated tree (below) obtained from the join procedure applied to the two annotated trees above.}
	\label{@characterized}
\end{figure}

\begin{lemma}\label{@authenticity}
Let $G$ be a graph, $({\sf T},\beta,{\sf r})$ be a nice tree decomposition of $G$, $v$ be a join node of ${\sf T}$ with children $v_1$ and $v_2$, and $\rho:\beta(v)\to[|\beta(v)|]$ be a bijection.
Let $k$ be an integer, $\Acal:=\Char_k(G_v,\beta(v),\rho)$, $\Acal_1:=\Char_k(G_{v_1},\beta(v_1),\rho)$, and $\Acal_2:=\Char_k(G_{v_2},\beta(v_2),\rho)$.
Then $\Acal=\join_k(\Acal_1,\Acal_2)$.
\end{lemma}

\begin{proof}
Let $\Bcal:=\fullchar(G_v,\beta(v),\rho)$, $\Bcal_1:=\fullchar(G_{v_1},\beta(v_1),\rho)$, and $\Bcal_2:=\fullchar(G_{v_2},\beta(v_2),\rho)$.
According to \autoref{@virtuously}, $\Bcal=\join^*(\Bcal_1,\Bcal_2)$.

Let $\hat{T}=(T,r,h,{\bf R},f)\in \Bcal$, $\hat{T}_1=(T_1,r_1,h_1,{\bf R}_1,f_1)\in \Bcal_1$, and $\hat{T}_2=(T_2,r_2,h_2,{\bf R}_2,f_2)\in \Bcal_2$, such that $\hat{T}\in\join^*(\hat{T}_1,\hat{T}_2)$.
The complete join procedure joins ${\bf R}_1$ and ${\bf R}_2$ to obtain~${\bf R}$, without modifying their boundary nor deleting vertices or edges, so given two sets of annotated trees $\Ccal_1$ and $\Ccal_2$, $\rep\circ\join^*(\Ccal_1,\Ccal_2)=\rep\circ\join^*(\rep(\Ccal_1),\rep(\Ccal_2))$.

The procedure can only increase the height of annotated trees, so ${\sf filter}_k\circ\join^*(\Ccal_1,\Ccal_2)={\sf filter}_k\circ\join^*({\sf filter}_k(\Ccal_1),{\sf filter}_k(\Ccal_2))$.

Moreover, item 6 of the $\diamondsuit_{\sf join}$ operation implies that, if for $w\in V(T)$, $V(T_w)\cap\im(f)=\emptyset$, then there is $i\in\{1,2\}$ such that $T_w=(T_i)_w$.
Therefore, each cropped subtree of $T$ is exactly a cropped subtree of $T_1$ or a cropped subtree of $T_2$.
Thus, $\crop\circ\join^*(\Ccal_1,\Ccal_2)=\join^*(\crop(\Ccal_1),\crop(\Ccal_2))$.

Since these three operations are commutative, we have $\Acal=\Mcal_k(\Bcal)={\sf filter}_k\circ\rep\circ\join^*(\Mcal_k(\Bcal_1),\Mcal_k(\Bcal_2))=\join_k(\Acal_1,\Acal_2)$.
\end{proof}

For the join procedure, we finally prove that it can generate a bounded number of annotated trees.

\begin{lemma}\label{@leistungsfiihigkeit}
Let $w$ and $k$ be two positive integers, let ${\bf G}_1$ and ${\bf G}_2$ be two compatible $w$-boundaried graphs, let $\hat{T}_1\in\Char_k({\bf G}_1)$, and let $\hat{T}_2\in\Char_k({\bf G}_2)$.
Then $|\join_k(\hat{T}_1,\hat{T}_2)|=2^{\Ocal(w\cdot k)}$.
\end{lemma}

\begin{proof}
Let $\hat{T}_1=(T_1,r_1,h_1,{\bf R}_1,f_1)$ and $\hat{T}_2=(T_2,r_2,h_2,{\bf R}_2,f_2)$.
We  will first show that $$|(T_1,r_1,f_1)\diamondsuit_{\sf join}(T_2,r_2,f_2)|=2^{\Ocal(w\cdot k)}.$$
Let $(T,r,f)\in (T_1,r_1,f_1)\diamondsuit_{\sf join}(T_2,r_2,f_2)$.
Notice that $(T,r)$ has at most $w$ leaves and height at most $k$.
Each leaf of $(T,r)$ is in $\im(f)=\im(f_1)=\im(f_2)$, so it corresponds to both a leaf of $(T_1,r_1)$ and a leaf of $(T_2,r_2)$.
Also note that $T$ is obtained by choosing, for each path from a leaf $v$ to $r$, a subset that corresponds to the path $v T_1 r$ (the rest is the path $v T_2 r$).
There are at most $w$ such paths from a leaf to $r$, and each of them has length at most $k+1$.
So $|(T_1,r_1,f_1)\diamondsuit_{\sf join}(T_2,r_2,f_2)|\leq 2^{w(k+1)}$.
Then, to construct an annotated tree $(T,r,h,{\bf R},f)$, the function $h$ and the boundaried graph ${\bf R}$ are totally
determined by $T$, $h_1$, $h_2$, ${\bf R}_1$, and ${\bf R}_2$.
So we obtain the desired result.
\end{proof}

\subsection{The algorithm}\label{@philanthropy}

We finally present a recursive algorithm (\autoref{@millennium}) that computes the elimination distance to $\exc(\Fcal)$ of a graph of bounded treewidth and proves \autoref{@unquestioned}. More precisely, given a boundary graph ${\bf G}$, a nice tree decomposition, and an integer $k$, the algorithm outputs $\Char_k({\bf G})$.

\begin{algorithm}[!ht]
\caption{\edrec$({\bf G},\Tcal,k,v)$}\label{@millennium}
\DontPrintSemicolon

  \KwInput{A connected boundaried graph ${\bf G}=(G,X,\rho)$, a nice tree decomposition $\Tcal=({\sf T},\beta,{\sf r})$ of $G$, an integer $k$, and a node $v\in V({\sf T})$ such that $X=\beta(v)$.}
  \KwOutput{The characteristic $\Char_k({\bf G})$ of $\bf G$.}
  $\Acal \gets \emptyset$\;
	$w\gets|\beta(v)|$\;
  \If{$v$ is a leaf}{$\Acal\gets\{\bf\un\}$}
  \ElseIf{$v$ is a forget node with child $v'$ and forgotten vertex $x$}
    {
    	$\rho'\gets\rho\oplus[x\mapsto w+1]$\;
			$\Acal'\gets\edrec((G,\beta(v'),\rho'),\Tcal,k,v')$\;
			$\Acal\gets \forget(\Acal')$
    }
  \ElseIf{$v$ is an introduce node with child $v'$ and introduced vertex $x$}
    {
    	$\tau\gets (x\leftrightarrow\rho^{-1}(w))$\;
			$\rho'\gets\rho\circ\tau$\;
			$\Acal'\gets\edrec((G,\beta(v'),\rho'|_{\beta(v')}),\Tcal,k,v')$\;
			$N\gets N_{G[\beta(v)]}(x)$\;
			$\Acal\gets\tau(\intr_k(\Acal',\rho'(N)))$
    }
  \ElseIf{$v$ is a join node with children $v_1$ and $v_2$}
   {
   		$\Acal_1\gets\edrec((G,\beta(v),\rho),\Tcal,k,v_1)$\;
			$\Acal_2\gets\edrec((G,\beta(v),\rho),\Tcal,k,v_2)$\;
			$\Acal\gets \join_k(\Acal_1,\Acal_2)$
   }
	\Return{$\Acal$}
\end{algorithm}

Note that using backtracking in \autoref{@millennium}, we can easily construct an annotated tree of minimum height in $\fullchar({\bf G})$ as well as its witness pair. In other words, given a connected graph $G$ such that $\ed_{\exc(\Fcal)}(G)\leq k$, we can construct an $\Fcal$-elimination tree of $G$ of height $\ed_{\exc(\Fcal)}(G)$ using \autoref{@millennium}.

\begin{lemma}\label{@accomplice}
Given a connected graph $G$, an integer $k$, a nice tree decomposition $\Tcal=({\sf T},\beta,{\sf r})$ of $G$ of width $w$, a bijection $\rho:\beta({\sf r})\to[|\beta({\sf r})|]$, and $v\in V({\sf T})$, \edrec$((G,\beta(v),\rho),\Tcal,k,v)$ outputs $\Char_k(G_v,\beta(v),\rho)$.
Moreover, \edrec$((G,\beta({\sf r}),\rho),\Tcal,k,{\sf r})$ outputs $\Char_k(G,\beta({\sf r}),\rho)$ in time $2^{\Ocal_{\ell_\Fcal}(w\cdot k+w\log w)}\cdot n$.
\end{lemma}

\begin{proof}
We prove that for every $v\in V({\sf T})$, by induction on the height of ${\sf T}_v$,
the algorithm \edrec\ of \autoref{@millennium} with input $((G,\beta(v),\rho),\Tcal,k,v)$ returns $\Char_k(G_v,\beta(v),\rho)$.
Indeed, if $v$ is a leaf, then $\Char_k(G_v,\beta(v),\rho)=\{\hat{\un}\}=\edrec((G,\beta(v),\rho),\Tcal,k,v)$.
Otherwise, $v$ is either a forget node, or an introduce node, or a join node, and the correctness of the algorithm is implied from the induction hypothesis and \autoref{@socialization}, \autoref{@philosopher}, and \autoref{@authenticity}, respectively.
Finally, since $G=G_{\sf r}$, we have \edrec$((G,\beta({\sf r}),\rho),\Tcal,k,{\sf r})=\Char_k(G,\beta({\sf r}),\rho)$.

We now analyze the running time.
A nice tree decomposition of width $w$ constructed by \autoref{@estclusire} has $\Ocal(w\cdot n)$ bags, hence the linear dependence follows.

Let us first analyze the join procedure.
During this procedure, we recursively obtain two characteristics of size $2^{\Ocal_{\ell_\Fcal}(w\cdot k+w\log w)}$ according to \autoref{@belligerent}.
Each pair of annotated trees can be joined in $2^{\Ocal(w\cdot k)}$ ways, according to \autoref{@leistungsfiihigkeit}.
Let ${\bf R}_1$ and ${\bf R}_2$ be the boundaried graphs of such a pair of annotated trees. There is an integer $z\leq w$ such that they belong to $\Rcal_{\ell_\Fcal}^{z}$.
By \autoref{@iinelstaai}, ${\bf R}_1$ and ${\bf R}_2$ have size $\Ocal_{\ell_\Fcal}(z)$.
So ${\bf R}:={\bf R}_1\bigoplus{\bf R}_2$ has size $\Ocal_{\ell_\Fcal}(z)$ as well.
The representation operation applied to ${\bf R}$ during the join operation for those two
annotated trees finds the representative of $l\leq z$ boundaried graphs of respective boundaries of sizes $z_1,...,z_l$ with $\sum_{i=1}^l z_i\leq z$ and with $\Ocal_{\ell_\Fcal}(z_i)$ vertices in the underlying graph due to \autoref{@iinelstaai}.
So by \autoref{@objectives}, the representation operation in the join procedure takes time $\sum_{i=1}^l 2^{\Ocal_{\ell_\Fcal}(w_i\log w_i)}=2^{\Ocal_{\ell_\Fcal}(w\log w)}$.
Checking that this is an annotated tree, that its height is at most $d$, and that we did not already create it also takes time $2^{\Ocal_{\ell_\Fcal}(w\cdot k+w\log w)}$.
Hence, the total running time of the join procedure is $2^{\Ocal_{\ell_\Fcal}(w\cdot k+w\log w)}$.

It is easy to see that the forget procedure applied to an annotated tree creates at most one annotated tree.
Moreover, the introduce procedure applied to an annotated tree creates $\Ocal(w\cdot k)$ annotated trees according to \autoref{@equivalent}.
Similarly, the representation, crop, and filter operations in these procedures take time $2^{\Ocal_{\ell_\Fcal}(w\cdot k + w \log w)}$.
So the lemma follows.
\end{proof}

We can finally prove \autoref{@unquestioned}.

\begin{proof}[Proof of Theorem \ref{@unquestioned}]
If $\Fcal=\{K_1\}$, the algorithm of \autoref{@polytheism} outputs the desired result in time $2^{\Ocal(\tw\cdot k)}\cdot n$.
So let us assume that $\Fcal$ is non-trivial.

Suppose first that $G$ is connected.
By \autoref{@naturalism} and \autoref{@estclusire}, we can obtain a nice tree decomposition $\Tcal=({\sf T},\beta,{\sf r})$ of width $2\tw+1$ in time $2^{\Ocal(\tw)}\cdot n$.
Let $\rho:\beta({\sf r})\to[|\beta({\sf r})|]$ be an arbitrary ordering on the vertices of $\beta({\sf r})$.
We apply {\edrec} with input $((G,\beta({\sf r}),\rho),\Tcal,k,{\sf r})$.
By \autoref{@accomplice}, this gives the desired result in time $2^{\Ocal_{\ell_\Fcal}(\tw(k+\log \tw))}\cdot n$.

If $G$ is not connected, we apply the same procedure on each connected component. The running time is the same as in the above case.
\end{proof}

\subsection{Exchangeability of boundaried graphs with the same characteristic}\label{@metropoltheater}

We give here a simple technical lemma on characteristics that will be used in \autoref{@participant}.
We show that boundaried graphs with the same characteristic can be exchanged, i.e., give graphs of the same elimination distance to $\Fcal$ when ``glued'' to the same boundaried graph.

Given a positive integer $k$ and a (possibly disconnected) boundaried graph ${\bf G}$, we define $\Char_k({\bf G})$ as $(\Char_k({\bf C}))_{{\bf C}\in\cc({\bf G})}$. Note that we still have $|\Char_k({\bf G})|=2^{\Ocal_{\ell_\Fcal}(w\cdot k+w\log w)}$. Therefore, we can extend $\funref{@caricatures}$ so that $|\Char_k({\bf G})|\leq\funref{@caricatures}(w,k)$ with $\funref{@caricatures}(w,k)=2^{\Ocal_{\ell_\Fcal}(w\cdot k+w\log w)}$ for any boundaried graph ${\bf G}$.

\begin{lemma}\label{@manuscripts}
Let $\bf G$, $\bf G'$, and $\bf G''$ be three compatible boundaried graphs and let $k$ be an integer such that $\ed_{\exc(\Fcal)}({\bf G}\oplus{\bf G''})\leq k$ and $\Char_k({\bf G})=\Char_k(\bf G')$.
Then $\ed_{\exc(\Fcal)}({\bf G}\oplus{\bf G''})=\ed_{\exc(\Fcal)}({\bf G'}\oplus{\bf G''})$.
\end{lemma}

\begin{proof}
We suppose without loss of generality that ${\bf G}\oplus{\bf G''}$, and therefore ${\bf G'}\oplus{\bf G''}$ as well, is connected. Indeed, if this is not the case, we may apply the following proof to each one of the connected components separately.

Let $\cc({\bf G})=\{{\bf C}_1,...,{\bf C}_l\}$ and $\cc({\bf G'})=\{{\bf C}_1',...,{\bf C}_l'\}$, such that $\Char_k({\bf C}_i)=\Char_k({\bf C}_i')$ for $i\in[l]$. Let $i\in[l]$. We write ${\bf C}_i=(C_i,B_i,\rho_i)$.
Let $\Tcal_i=({\sf T}_i,\beta_i,{\sf r}_i)$ be a nice tree decomposition of ${\bf C}_i$, i.e., such that $\beta_i({\sf r}_i)=B_i$.
Since the $B_i$'s are pairwise disjoint, there is a rooted tree decomposition $\Tcal^*=({\sf T}^*,\beta^*,{\sf r})$ of $G''$, where $G''$ is the underlying graph of ${\bf G''}$, such that, for $i\in[l]$, there is $v_i\in\leaf({\sf T}^*,{\sf r})$ with $\beta^*(v_i)=B_i$.
Let $\Tcal=({\sf T},\beta,{\sf r})$ be the tree decomposition obtained from $\Tcal^*$ and the $\Tcal_i$'s by identifying $v_i$ with ${\sf r}_i$ for $i\in[l]$ and adding nodes in $\Tcal^*$ using \autoref{@estclusire}, so that $\Tcal$ is a nice tree decomposition of ${\bf G}\oplus{\bf G''}$.

Let $\Tcal'=({\sf T}',\beta',{\sf r})$ be a nice tree decomposition of the graph ${\bf G'}\oplus{\bf G''}$ obtained from $\Tcal$ by replacing $\Tcal_i$ by a nice tree decomposition $\Tcal_i'=({\sf T}_i',\beta_i',{\sf r}_i')$ of ${\bf C}_i'$ for $i\in[l]$.
Observe that, for every $i\in[l]$, according to \autoref{@accomplice}, \edrec$(({\bf G}\oplus{\bf G''},\beta(v_i),\rho_{v_i}),\Tcal,k,v_i)=\Char_k({\bf G}\oplus{\bf G''},\beta(v_i),\rho)=\Char_k({\bf C}_i)$.
Similarly, \edrec$(({\bf G'}\oplus{\bf G''},\beta'(v_i),\rho_{v_i}),\Tcal',k,v_i)=\Char_k({\bf C}_i')$.
Thus, \autoref{@millennium} applied with input $(({\bf G}\oplus{\bf G''},\beta(v_i),\rho_{v_i}),\Tcal,k,v_i)$ and $(({\bf G'}\oplus{\bf G''},\beta'(v_i),\rho_{v_i}),\Tcal',k,v_i)$ outputs the same result.

We next set $U:=V({\sf T}\setminus\bigcup_{i\in[l]}V({\sf T}_i))$.
Note that  $U=V({\sf T}'\setminus\bigcup_{i\in[l]}V({\sf T}_i'))$.
Therefore, in each node $u$ of $U$,
\autoref{@millennium} applied with input $(({\bf G}\oplus{\bf G''},\beta(u),\rho_u),\Tcal,k,u)$ and $(({\bf G'}\oplus{\bf G''},\beta'(u),\rho_u),\Tcal',k,u)$ outputs the same result.
Thus, \autoref{@millennium} applied with input $(({\bf G}\oplus{\bf G''},\beta({\sf r}),\rho_{\sf r}),\Tcal,k,{\sf r})$ and $(({\bf G'}\oplus{\bf G''},\beta'({\sf r}),\rho_{\sf r}),\Tcal',k,{\sf r})$ outputs the same result.
So $\Char_k({\bf G}\oplus{\bf G''},\beta({\sf r}),\rho_{\sf r})=\Char_k({\bf G'}\oplus{\bf G''},\beta'({\sf r}),\rho_{\sf r})$.
Therefore, according to \autoref{@ineluctable}, $\ed_{\exc(\Fcal)}({\bf G}\oplus{\bf G''})=\ed_{\exc(\Fcal)}({\bf G'}\oplus{\bf G''})$.
\end{proof}

\section{Elimination distance to a minor-closed graph class}\label{@oeconomicus}

We finally present our main result for {\sc $\Fcal$-M-Elimination Distance}. The following theorem is a restatement of \autoref{@communicated}.

\begin{theorem}\label{@successors}
For every finite collection of graphs $\Fcal$, there exists an algorithm that, given a graph $G$ and an integer $k$,
decides whether $\ed_{\exc(\Fcal)}(G)\leq k$ in time $2^{2^{2^{k^{\Ocal_{\ell_\Fcal}(1)}}}}\cdot n^2$.
In the particular case when $\Fcal$ contains an apex-graph, this algorithm runs in time $2^{2^{\Ocal_{\ell_\Fcal}(k^2\log k)}}\cdot n^2$.
\end{theorem}

\sugar{This algorithm is very similar to the one in \autoref{@conditioned}.
We use the same propositions and lemmata,  {except from the ones that were} tailored for {\sc $\Fcal$-M-Deletion} and for which we need to find an analogous result for {\sc $\Fcal$-M-Elimination Distance}.
In \autoref{@fanaticism}, we described an algorithm to find the elimination distance when the treewidth is bounded, which is the analogue of \autoref{@calculated} for {\sc $\Fcal$-M-Elimination Distance}. In \autoref{@hierarchical}, we present an analogue of \autoref{@transforma}, which either reports an upper bound on the treewidth of the input graph, or finds a wall, or reports that we deal with a \no-instance.

In contrast to the algorithm for {\sc $\Fcal$-M-Deletion}, we cannot use branching the same way because we do not have a bound in $k$ on the number of vertices in a $k$-elimination set, like we had for a $k$-apex set. Thus, while Step~3 of \autoref{@conditioned} could be applied at most $a_\Fcal^k$ times, such a Step~3 would now be applied at most $a_\Fcal^n$ times, which is not efficient.
Since we do not use branching anymore to reduce the size of the apex set of a flatness pair, when we apply \autoref{@disreputable}, the size $a$ of the apex set in the input depends on $k$. Hence, we need to find in the beginning a wall whose size dependence in $k$ is triple-exponential.
This explains the triple-exponential in $k$ that appears in the running time of \autoref{@successors}.
The above algorithm is described in~\autoref{@memorindum} and in~\autoref{@preanimism} we present the proof of its correctness.}

\subsection{Quickly finding a wall}\label{@hierarchical}

We first prove \autoref{@transforma} in the case of {\sc $\Fcal$-M-Elimination Distance}.
The proof is very similar to the one given in \cite{SauST21kapiII} for \autoref{@transforma} in the case of {\sc $\Fcal$-M-Deletion} and is achieved using the following result from Perkovic and Reed\cite{PerkovicR00anim}. \sugar{The main difference with respect to the proof of \autoref{@transforma} given in \cite{SauST21kapiII} is that we need to use two new ingredients tailored for {\sc $\Fcal$-M-Elimination Distance}, namely \autoref{@unquestioned} and \autoref{@barbarians}.}

\begin{proposition}[\cite{PerkovicR00anim}]\label{@antisthenes}
There exists an algorithm with the following specifications:\medskip

	\noindent{\textbf{Input}:}	A graph $G$ and $t\in\bN$ such that $|V(G)|\geq 12t^{3}$.\\
	\noindent{\textbf{Output}:} A graph $G^*$ such that $|V(G^*)|\leq (1-\frac{1}{16t^{2}}) \cdot |V(G)|$ and,
	\begin{itemize}
		\item either $G^*$ is a subgraph of $G$ such that $\tw(G)=\tw({G^*})$, or
		\item $G^*$ is obtained from $G$ after contracting the edges of a matching in $G$.
	\end{itemize}
	Moreover, the algorithm runs in time $2^{\Ocal(t)} \cdot n$.
\end{proposition}

\begin{proof}[Proof of Proposition \ref{@transforma} in the case of {\sc $\Fcal$-M-Elimination Distance}]\phantom{.ddd\\}

\noindent Let $c:=\funref{@relentlessly}(s_\Fcal)\cdot r+k$.\medskip

Suppose that $|V(G)| < 12c^{3}$.
Run the algorithm of \cite{ArnborgCP87comp} that, in time $\Ocal(|V(G)|^{c+2})=2^{\Ocal_{\ell_\Fcal} ((r+k)\cdot \log (r+k))}$, checks whether ${\tw}(G)\leq c$. If this is the case, report the same and stop.
If not, we aim to find an $r$-wall of $G$ or conclude that we are dealing with a {\sf no}-instance.
First consider an arbitrary ordering $(v_1,\dots, v_{|V(G)|})$ of the vertices of $G$.
For each $i\in[|V(G)|]$, set $G_i$ to be the graph induced by the vertices $v_1, \dots, v_i$.
Iteratively run the algorithm of \autoref{@naturalism} on $G_i$ and $c$ for increasing values of $i$.
This algorithm runs in time $2^{\Ocal(c)}\cdot |V(G)|=2^{\Ocal_{\ell_\Fcal} (r+k)}$.
Let $j\in[|V(G)|]$ be the smallest integer such that the above algorithm outputs a report that ${\tw}(G_j)>c$ (it exists since ${\tw}(G)>c$) and notice that there exists a tree decomposition $({\sf T}_j, \beta_j)$ of $G_j$ of width at most $2c+2$, obtained by the one of $G_{j-1}$ by adding the vertex $v_j$ to all the bags.
Thus, we can call the algorithm of \autoref{@unquestioned} with input $(G_j, 2c+2,k)$, which runs in time $2^{\Ocal_{\ell_\Fcal}(c \cdot (k+\log c))}\cdot |V(G_j)|=2^{\Ocal_{\ell_\Fcal}((r+k)\cdot (k+\log (r+k)))}$, in order to find, if it exists, a $k$-elimination set $S_j$ of $G_j$ for $\exc(\Fcal)$.

	\begin{itemize}		\item[$\bullet$] If such a set $S_j$ does not exist, then safely report that $(G,k)$ is a \no-instance.
		\item[$\bullet$] If such a set $S_j$ exists, then call the algorithm of \autoref{@camouflage} for $G_j\setminus S_j$, (and the decomposition of $G_j\setminus S_j$ obtained from $({\sf T}_j, \beta_j)$ by removing the vertices of $S_j$ from all the bags) in order to check whether it contains an elementary $r$-wall $W$ as a minor.
		   This algorithm runs in time $2^{\Ocal(c \cdot \log c)}\cdot r^{\Ocal(c)}\cdot 2^{\Ocal(r^{2})}\cdot |V(G_j\setminus S_j)|=2^{\Ocal_{\ell_\Fcal}((r+k) \cdot \log (r+k))}\cdot r^{\Ocal_{\ell_\Fcal}(r+k)}\cdot 2^{\Ocal(r^2)}=2^{\Ocal_{\ell_\Fcal}(r^2+(r+k) \cdot \log (r+k))}$, since $|E(W)|=\Ocal(r^2)$.
		   Since all connected components of $G_j\setminus S_j$ are in ${\exc(\Fcal)}$, $G_j\setminus S_j$ does not contain $K_{s_\Fcal}$ as a minor. By \autoref{@barbarians}, $\tw(G_j\setminus S_j)\geq c-k=\funref{@relentlessly}(s_\Fcal)\cdot r$. So because of \autoref{@overwhelmingly}, the algorithm of \autoref{@camouflage} will output an elementary $r$-wall $W$ of $G_j \setminus S_j$.
		   We also return $W$ as a wall of $G$.
	\end{itemize}
	Therefore, in the case where $|V(G)|< 12c^3$, we obtain one of the three possible outputs in time $2^{\Ocal_{\ell_\Fcal}(r^2+k^2)}$.
	\medskip

	If $|V(G)|\geq 12c^{3}$, then call the algorithm of \autoref{@antisthenes} with input $(G,c)$, which outputs a graph $G^*$ such that $|V(G^*)|\leq (1-\frac{1}{16c^{2}}) \cdot |V(G)|$ and
	\begin{itemize}
		\item either $G^*$ is a subgraph of $G$ such that $\tw(G)=\tw(G^*)$, or
		\item $G^*$ is obtained from $G$ after contracting the edges of a matching in $G$.
	\end{itemize}

	In both cases, recursively call the algorithm on $G^*$ and distinguish the following two cases.\medskip

	\noindent{\em Case 1}: $G^*$ is a subgraph of $G$ such that $\tw(G)=\tw(G^*)$.

	\begin{itemize}
		\item[(a)] 	If the recursive call on $G^*$ reports that $\tw(G^*)\leq c$, then return that $\tw(G)\leq c$.

		\item[(b)] 	If the recursive call on $G^*$ outputs an $r$-wall $W$ of $G^*$, then return $W$ as a wall of $G$.

		\item[(c)] 	If $(G^*,k)$ is a \no-instance, then report that $(G,k)$ is also a {\sf no}-instance.

	\end{itemize}
	\medskip

	\noindent{\em Case 2}: $G^*$ is obtained from $G$ after contacting the edges of a matching in $G$.

	\begin{itemize}
		\item[(a)] If the recursive call on $G^*$ reports that $\tw(G^*)\leq c$, then do the following.
		      First notice that the fact that $\tw(G^*)\leq c$ implies that $\tw(G)\leq 2c$,
		      since we can obtain a tree decomposition $({\sf T},\beta)$ of $G$ from a tree decomposition $({\sf T}^*,\beta^*)$ of $G^*$,
		      by replacing, in every $t\in V({\sf T}^*)$, every occurrence of a vertex of $G^*$ that is a result of an edge contraction by its endpoints in $G$.
		      Thus, we can call the algorithm of \autoref{@unquestioned} with input $(G, 2c,k)$, which runs in time $2^{\Ocal_{\ell_\Fcal}(c(k+\log c))}\cdot n$, in order to find, if it exists, a $k$-elimination set $S$ of $G$ for $\exc(\Fcal)$. We distinguish again two cases.

		      \begin{itemize}
			      \item[$\bullet$] If such a set $S$ does not exist, then the algorithm reports that $(G,k)$ is a \no-instance.
			      \item[$\bullet$] If such a set $S$ exists, then apply the algorithm of \autoref{@naturalism} with input $(G\setminus S,2c)$, which runs in time $2^{\Ocal(c)}\cdot n$, and obtain a tree decomposition of $G\setminus S$ of width at most $4c+1$.
			            Using this decomposition, call the algorithm of \autoref{@camouflage} for $G\setminus S$ in order to check whether it contains an elementary $r$-wall $W$ as a minor.
			            This algorithm runs in time $2^{\Ocal(c \cdot \log c)}\cdot r^{\Ocal(c)}\cdot 2^{\Ocal(r^{2})}\cdot n=2^{\Ocal_{\ell_\Fcal}((r+k) \cdot \log (r+k))}\cdot r^{\Ocal_{\ell_\Fcal}(r+k)}\cdot 2^{\Ocal(r^2)}\cdot n=2^{\Ocal_{\ell_\Fcal}(r^2+(r+k) \cdot \log (r+k))}\cdot n$, since $|E(G\setminus S)|=\Ocal(n)$ and $|E(W)|=\Ocal(r^2)$.
			            If this algorithm outputs an elementary $r$-wall $W$ of $G\setminus S$, then output $W$. Otherwise, safely report, because of \autoref{@overwhelmingly} and \autoref{@barbarians}, that $\tw(G)\leq \funref{@relentlessly}(s_\Fcal)\cdot r+k=c$.
		      \end{itemize}

		\item[(b)] 	If the recursive call on $G^*$ outputs an $r$-wall $W^*$ of $G^*$, then by uncontracting the edges of $M$ in $W^*$, we can return an $r$-wall of $G$.

		\item[(c)] 	If $(G^*,k)$ is a \no-instance, then report that $(G,k)$ is also a \no-instance.
	\end{itemize}
	It is easy to see that the running time of the above algorithm is given by the function $$T(n,k,r)\ \leq\  T\left((1-\frac{1}{16c^{2}})\cdot n,k,r\right)+ 2^{\Ocal_{\ell_\Fcal}(r^2+k^2)}\cdot n,$$
	which implies that $T(n,k,r)=2^{\Ocal_{\ell_\Fcal}(r^2+k^2)}\cdot n$, as claimed.
\end{proof}

\subsection{Description of the algorithm for {\sc \texorpdfstring{$\Fcal$}{F}-M-Elimination Distance}}
\label{@memorindum}

We define the following constants.
\begin{align*}
	a =				&	\ \funref{@collaboration}(s_\Fcal+k), &
	q =				& \ \funref{@categories}(a_\Fcal,s_\Fcal,k(k+1)/2), \\
	p = 			& \ \funref{@provincial}(a_\Fcal,s_\Fcal,k(k+1)/2), &
	l = 			& \ (q-1)\cdot a, \\
	d = 			& \ \funref{@deliberation}(a_\Fcal-1,\ell_\Fcal), &
	r_4 =			& \ \funref{@differences}(a_\Fcal-1,\ell_\Fcal,3,k(k+1)/2), \\
	r_3 =			& \ \funref{@philistines}(r_4,a_\Fcal-1,a,d), &
	r_2 =     & \ \odd(\max\{\funref{@idealistic}(l+1,r_3,p),\funref{@unaffected}(a_\Fcal,s_\Fcal,k(k+1)/2)\}), \\
	r_1 =     & \ \odd(\funref{@corollaries}(s_\Fcal+k)\cdot r_2), &
\end{align*}

Note that $r_4=\Ocal_{\ell_\Fcal}(k^2)$, $r_3=\Ocal_{\ell_\Fcal}(k^{2\cdot c})$, $r_2=\Ocal_{\ell_\Fcal}(k^{2\cdot c+15})$, and $r_1=2^{\Ocal_{\ell_\Fcal}(k^2\log k+c\log k)}$, where $c=\funref{@withdrawing}(a_\Fcal-1,a,d)=2^{\Ocal_{\ell_\Fcal}(k^{24\cdot(a_\Fcal-1)})}$.
Recall from \autoref{@graphically} that we assume that $G$ has $\Ocal_{s_\Fcal}(k\sqrt{\log k}\cdot n)$ edges.\bigskip

Run the algorithm {\tt Find-Wall-Ed} from \autoref{@transforma} with input $(G,r_1,k)$
and, in time $2^{\Ocal_{\ell_\Fcal}(r_1^2+k^2)}\cdot n=2^{2^{\Ocal_{\ell_\Fcal}(k^2\log k+c\log k)}}\cdot n$,
\begin{itemize}
	\item either report a \no-instance, or
	\item conclude that $\tw(G)\leq \funref{@veneration}(s_\Fcal)\cdot r_1+k$ and solve {\sc $\Fcal$-M-Elimination Distance} in time $2^{\Ocal_{\ell_\Fcal}((r_1+k)k + (r_1+k)\log(r_1+k))}\cdot n=2^{2^{\Ocal_{\ell_\Fcal}(k^2\log k+c\log k)}}\cdot n$ using the algorithm of \autoref{@unquestioned}, or
	\item obtain an $r_1$-wall $W_1$ of $G$.
\end{itemize}

If the output of \autoref{@transforma} is a wall $W_1$,
then run the algorithm {\tt Clique-or-twFlat} of \autoref{@unimportant} with input $(G,r_2,s_\Fcal+k)$.
This takes time $2^{2^{\Ocal_{\ell_\Fcal}(k^2\log k)}\cdot r_2^3\log r_2}\cdot n = 2^{2^{\Ocal_{\ell_\Fcal}(k^2\log k+c\log k)}}\cdot n$.
If the result is a set $A$ of size at most $a$ and a regular flatness pair $(W_2,\mathfrak{R}_2)$ of $G\setminus A$ of height $r_2$ whose $\mathfrak{R}_2$-compass has treewidth at most $r_1$, then proceed, otherwise output a \no-answer.\smallskip

Compute a $W_2$-canonical partition $\tilde{\Qcal}$ of $G\setminus A$.
Compute the set $B$ of vertices of $A$ that are adjacent to at least $q$ $p$-internal bags of $\tilde{\Qcal}$.
As in \autoref{@prohibitions}, compute a collection $\Wcal=\{W^1,...,W^{l+1}\}$ of $r_3$-subwalls of $W_2$ such that for every $i \in [l+1]$, $\bigcup \influence_{\mathfrak{R}_2}(W^i)$ is a subgraph of $\bigcup \{Q\mid Q \text{ is a $p$-internal bag of }\tilde{\Qcal}\}$ and for every $i,j\in[l+1]$, with $i\neq j$, there is no internal bag of $\tilde\Qcal$ that contains vertices of both $V(\bigcup \influence_{\mathfrak{R}_2}(W^i))$ and $V(\bigcup \influence_{\mathfrak{R}_2}(W^j))$.
By the choice of $l$, there is an $i\in [l+1]$ such that no vertex of $\bigcup \influence_{\mathfrak{R}_2}(W^i)$ is adjacent to a vertex of $A\setminus B$.\smallskip

Run the algorithm from \autoref{@expurgated} with input $(G\setminus B,W_2,\mathfrak{R}_2,W^i)$ to obtain a $W^i$-tilt $(W_3,\mathfrak{R}_3)$ of $(W_2,\mathfrak{R}_2)$ in time $\Ocal_{s_\Fcal}(k\sqrt{\log k}\cdot n)$.
\smallskip

As a next step, we apply the algorithm {\tt Homogeneous} of \autoref{@disreputable} with input $(r_4, a_\Fcal-1, a,d,r_1,G,B,W_3,\mathfrak{R}_3)$, which, in time $2^{\Ocal(c\cdot r_4\log r_4+r_1\log r_1)}\cdot (n+m)=2^{2^{\Ocal_{\ell_\Fcal}(k^2\log k+c\log k)}}\cdot n$, outputs a flatness pair $(W_4,\mathfrak{R}_4)$ of $G\setminus B$ of height $r_4$ that is $d$-homogeneous with respect to $\binom{B}{< a_\Fcal}$ and is a $W^*$-tilt of $(W_3,\mathfrak{R}_3)$ for some subwall $W^*$ of $W_3$.
\smallskip

Finally, apply the algorithm {\tt Find-Irrelevant-Vertex} of \autoref{@civilizing} with input $(k(k+1)/2,a_\Fcal-1,G, B, W_4,\mathfrak{R}_4)$, which outputs, in time $\Ocal_{s_\Fcal}(k\sqrt{\log k}\cdot n)$, an irrelevant vertex $v$ such that $(G,k)$ and $(G\setminus v,k)$
are equivalent instances of {\sc $\Fcal$-M-Elimination Distance}.
Then the algorithm runs recursively on the equivalent instance $(G\setminus v,k)$.\bigskip

Since each run takes time $2^{2^{\Ocal_{\ell_\Fcal}(k^2\log k+c\log k)}}\cdot n$ and there are at most $n$ runs,  the algorithm indeed runs in time $2^{2^{\Ocal_{\ell_\Fcal}(k^2\log k+c\log k)}}\cdot n^2$.

Note that $c=2^{\Ocal_{\ell_\Fcal}(k^{24\cdot(a_\Fcal-1)})}$, so if $\Fcal$ contains an apex-graph, i.e., if $a_\Fcal=1$, then $c=\Ocal_{\ell_\Fcal}(1)$.
Thus, the running time is $2^{2^{2^{\Ocal_{\ell_\Fcal}(k^{24\cdot(a_\Fcal-1)})}}}\cdot n^2$ in the general case and $2^{2^{\Ocal_{\ell_\Fcal}(k^2\log k)}}\cdot n^2$ in the case where $\Fcal$ contains an apex-graph.

\subsection{Correctness of the algorithm}\label{@preanimism}

Let $(G,k)$ be a \yes-instance and let $S$ be a $k$-elimination set of $G$ for $\exc(\Fcal)$.
By running \autoref{@transforma} with input $(G,r_1,k)$, the algorithm should either get a report that $\tw(G)\leq \funref{@veneration}(s_\Fcal)\cdot r_1+k$ or find an $r_1$-wall. The correctness of the former is obvious, so we will focus on the latter. \smallskip

Let $W_1$ be an $r_1$-wall of $G$.
According to \autoref{@inoperative}, $K_{s_\Fcal+k}$ is not a minor of $G$.
Moreover, since $W_1$ is a wall of $G$ of height $r_1$, $\tw(G)\geq\tw(W_1)\geq r_1\geq\funref{@corollaries}(s_\Fcal+k)\cdot r_2$.
Hence, if the algorithm runs {\tt Clique-or-twFlat} of \autoref{@unimportant} with input $(G,r_2,s_\Fcal+k)$, it should obtain a set $A$ of size at most $a$ and a regular flatness pair $(W_2,\mathfrak{R}_2)$ of $G\setminus A$ of height $r_2$ whose $\mathfrak{R}_2$-compass has treewidth at most $r_1$.\smallskip

As described in the algorithm, due to \autoref{@prohibitions} and the fact that $r_2\geq\funref{@idealistic}(l+1,r_3,p)$, there is an $r_3$-wall $W^i$ that is a subwall of $W_2$ such that no vertex of $\bigcup \influence_{\mathfrak{R}_2}(W^i)$ is adjacent to a vertex of $A\setminus B$, where $B$ is the set of vertices of $A$ adjacent to at least $q$ $p$-internal bags of a $W_2$-canonical partition $\tilde{\Qcal}$ of $G\setminus A$.\smallskip

When the algorithm applies \autoref{@expurgated} with input $(G\setminus B,W_2,\mathfrak{R}_2,W^i)$, it obtains a $W^i$-tilt $(W_3,\mathfrak{R}_3)$ of $(W_2,\mathfrak{R}_2)$. Due to \autoref{@unimagined} and \autoref{@successively}, $(W_3,\mathfrak{R}_3)$ is a regular flatness pair of $G\setminus B$ whose $\mathfrak{R}_5$-compass has treewidth at most $r_1$.
Thus, since $r_3=\funref{@philistines}(r_4,a_\Fcal-1,a,d)$, the algorithm can apply {\tt Homogeneous} of \autoref{@disreputable} with input $(r_4,a_\Fcal-1,a,d,r_1,G,B,W_3,\mathfrak{R}_3)$ to obtain a flatness pair $(W_4,\mathfrak{R}_4)$ of $G\setminus B$ of height $r_4$ that is $d$-homogeneous with respect to $\binom{B}{< a_\Fcal}$ and is a $W^*$-tilt of $(W_3,\mathfrak{R}_3)$ for some subwall $W^*$ of $W_3$. According to \autoref{@successively}, $(W_4,\mathfrak{R}_4)$ is regular.\smallskip

Let $S'$ be a $k$-elimination set of $G$ for $\exc(\Fcal)$.
\autoref{@lucinatory} implies that there is a set $X_{S'}\supseteq S'$ such that $G\setminus X_{S'}\in{\exc(\Fcal)}$ and $\bid_{G\setminus A,W_2}(X_{S'})\leq k(k+1)/2$.\smallskip

Since $r_2\geq \funref{@unaffected}(a_\Fcal,s_\Fcal,k(k+1)/2)$, every subset of $B$ of size $a_\Fcal$ intersects $X_{S'}$ according to \autoref{@proclamation}.
Hence, $|B\setminus X_{S'}|\leq a_\Fcal-1$.\smallskip

Moreover, note that $(W_3,\mathfrak{R}_3)$ is a $W^i$-tilt of $(W_2,\mathfrak{R}_2)$, $(W_4,\mathfrak{R}_4)$ is a $W^*$-tilt of $(W_3,\mathfrak{R}_3)$, and $(W_4,\mathfrak{R}_4)$ is a flatness pair of $G\setminus B$ with $B\subseteq A$.
Thus, given a $W_4$-canonical partition $Q_1$ of $G\setminus B$, there is a $W_2$-canonical partition $Q_2$ of $G\setminus A$ such that each internal bag of $Q_1$ is an internal bag of $Q_2$.
Thus, $\bid_{G\setminus B,W_4}(X_{S'})\leq\bid_{G\setminus A,W_2}(X_{S'})$.\smallskip

Hence, the algorithm can apply {\tt Find-Irrelevant-Vertex} of \autoref{@civilizing} with input $(k(k+1)/2,a_\Fcal-1,G, B, W_4,\mathfrak{R}_4)$ and obtain a vertex $v$ such that, for any $k$-elimination set $S'$ of $G$ for $\exc(\Fcal)$, $G\setminus X_{S'}\in{\exc(\Fcal)}$ if and only if $G\setminus (X_{S'}\setminus v)\in{\exc(\Fcal)}$.
Thus, there is a $k$-elimination set of $G$ for $\exc(\Fcal)$ if and only if there is a $k$-elimination set of $G\setminus v$ for $\exc(\Fcal)$.
It follows that $(G,k)$ and $(G\setminus v,k)$ are equivalent instances of {\sc $\Fcal$-M-Elimination Distance}.\smallskip

Suppose now that $(G,k)$ is a \no-instance.
Note that as long as \autoref{@unimportant} outputs a flatness pair $(W_2,\mathfrak{R}_2)$, what follows in the proof of correctness works even if $(G,k)$ is a \no-instance. Therefore, we will find an irrelevant vertex.
Otherwise, we would have declared a \no-instance beforehand.
Thus, \autoref{@successors} follows.

\section{Elimination distance when excluding an apex-graph}\label{@stubbornly}

In the case where $\Fcal$ contains an apex-graph, we obtain an alternative algorithm whose complexity is single-exponential in $k$ and cubic in $n$. The following theorem is a restatement of \autoref{@swineherds}.

\begin{theorem}\label{@overflowing}
For every finite collection of graphs $\Fcal$ that contains an apex-graph, there exists an algorithm that, given a graph $G$ and an integer $k$,
decides whether $\ed_{\exc(\Fcal)}(G)\leq k$ in time $2^{k^{\Ocal_{\ell_\Fcal}(1)}}\cdot n^3$.
\end{theorem}

Contrary to the previous section, since $a_\Fcal=1$, any vertex fulfilling the criteria of \autoref{@proclamation} belongs to every $k$-elimination set $S$ of the input graph for $\exc(\Fcal)$.
In the case where $\Fcal$ contains an apex-graph in
\cite{SauST21kapiII}, a Step~3 similar to the one of \autoref{@conditioned} can be applied, where a vertex belonging to every $k$-apex set can be found. Hence, after running Step~3 $k$ times, a $k$-apex set is found and the algorithm stops.
However, a $k$-elimination set may have size $\Omega(n)$, so we may run Step~3 $\Omega(n)$ times.
Since our Step~3 below runs in quadratic time, this gives the cubic dependence of this algorithm.
Fortunately, since we apply Step~3, just like in \autoref{@conditioned} and contrary to \autoref{@oeconomicus}, we manage to find a flatness pair along with an apex set whose size does not depend on $k$. Hence, when applying \autoref{@disreputable}, we do not get a triple-exponential dependence on $k$ anymore for the size of the wall we need to find originally.\smallskip

In order to remember the vertices that are found to belong to every $k$-elimination set, since they do not decrease $k$ like in \autoref{@conditioned}, we need to distinguish them in the input. Hence, we actually give here an algorithm to solve a more general problem with annotations described in \autoref{@abolishing}.

\subsection{Generalization to annotated elimination distance}\label{@abolishing}

Contrary to the previous section, since $a_\Fcal=1$, when applying \autoref{@proclamation}, we find a vertex that belongs to every $k$-elimination set $S$. Such vertices are taken into account by considering the following generalization of {\sc $\Fcal$-M-Elimination Distance}.

\begin{center}
	\fbox{
		\begin{minipage}{13.2cm}
			\noindent{\sc Annotated $\Fcal$-M-Elimination Distance}\\
			\noindent\textbf{Input:}~~A graph $G$, a set $S_0\subseteq V(G)$, and a non-negative integer $k$.\\
			\textbf{Objective:}~~Find, if it exists, a $k$-elimination set $S$ of $G$ for the class $\exc(\Fcal)$ such that $S_0\subseteq S$.
		\end{minipage}
	}
\end{center}

$S_0$ is a set of annotated vertices that corresponds to the vertices identified as vertices of every $k$-elimination set $S$ while running the algorithm.
Clearly, {\sc $\Fcal$-M-Elimination Distance} is the particular case of {\sc Annotated $\Fcal$-M-Elimination Distance} when $S_0$ is empty.\smallskip

In the following lemma, we generalize \autoref{@unquestioned} to its ``annotated'' version.
More precisely, we present a simple trick to reduce the above problem to its ``unannotated'' version while not changing the treewidth of the input graph so much.

\begin{lemma}\label{@tualization}
Let $\Fcal$ be a finite collection of graphs.
There is an algorithm that, given a graph $G$, a set $S_0\subseteq V(G)$, and two integers $k$ and $\tw$ such that the treewidth of $G$ is bounded by $\tw$, decides whether $(G,S_0,k)$ is a \yes-instance of {\sc Annotated $\Fcal$-M-Elimination Distance} in time $2^{\Ocal_{\ell_\Fcal}(\tw\cdot k+\tw\log \tw))}\cdot n$.
\end{lemma}

\begin{proof}
Given a minor-closed graph class $\Gcal$, let ${\Ccal(\Gcal)}:=\{G \mid \forall C\in\cc(G), C\in\Gcal\}$.
Bulian and Dawar~\cite{BulianD17fixe} showed that if $H\in\obs(\Gcal)$ has $l$ connected components, then each graph $\bar{H}$ obtained from $H$ by adding $l-1$ edges to obtain a connected graph belongs to $\obs(\Ccal(\Gcal))$.
Thus, let $H_\Fcal\in\obs(\Ccal{\exc(\Fcal)}))$ obtained in such a way. As said above, $H_\Fcal$ is connected.

Let $G$ be a graph of treewidth at most $\tw$ and $S_0$ be a subset of $V(G)$.
Let $G'$ be a graph obtained from $G$ by gluing a graph $H_v$ isomorphic to $H_\Fcal$ to each vertex $v$ of $S_0$, where $v$ is identified with an arbitrarily chosen vertex of $H_v$.

Let us show that $(G,S_0,k)$ is a \yes-instance of {\sc Annotated $\Fcal$-M-Elimination Distance} if and only if $(G',k)$ is a \yes-instance of {\sc $\obs(\Ccal\exc(\Fcal)))$-M-Elimination distance}.
If $S_0=\emptyset$, the proof is trivial, so we suppose $S_0\neq\emptyset$.

Let $(F',\chi',R')$ be a $\Ccal{\exc(\Fcal)})$-elimination forest of $G'$ of height at most $k$ with associated $k$-elimination set $S'$.
For each $v\in S_0$, the fact that $H_v\notin \Ccal{\exc(\Fcal)})$ implies that $V(H_v)\cap S'\neq\emptyset$.
Let $y\in V(F')$ be the least common ancestor of $\chi'^{-1}(H_v)$ in $F'$.
Since $H_v$ is connected, according to \autoref{@unreflectingly}, $y$ exists and belongs to $\chi'^{-1}(H_v)$.
Moreover, since $V(H_v)\cap S'\neq\emptyset$, $y\in\Int(F',R')$.

Let $(F'',\chi'',R'')$ be the $\Fcal$-elimination forest of $G$ obtained from $(F',\chi',R')$ as follows.
For every $v\in S_0$ and every $t\in\chi'^{-1}(H_v\setminus v)$
if $t\in\Int(F',R')$, remove $t$ from $F''$ and add edges between the parent and the children of $t$
and if $t\in\leaf(F',R')$,
remove $H_v\setminus v$ from $\chi''(t)$.
If $G[\chi''(t)]$ is not connected, then we update $F''$ by replacing $t$ by $|{\sf cc}(G[\chi''(t)])|$ nodes, each one associated with a connected component of $G[\chi''(t)]$.

Thus, we have an $\Fcal$-elimination forest of $G$ associated with a $k$-elimination set $S$ with $S_0\subseteq S$ and with height at most $k$, implying that $(G,S_0,k)$ is a \yes-instance of {\sc Annotated $\Fcal$-M-Elimination Distance}.

Conversely, given an $\Fcal$-elimination forest $(F,\chi,R)$ of $G$ of height at most $k$, and $S_0\subseteq S:=\chi(\Int(F,R))$, we can obtain a $\obs(\Ccal\exc(\Fcal)))$-elimination forest $(F',\chi',R')$ of $G'$ of height at most $k$ by adding to each node $\chi^{-1}(v)$ for $v\in S_0\subseteq S$ a leaf associated with $H_v\setminus v\in \Ccal{\exc(\Fcal)})$. The height of $(F',\chi',R')$ is indeed still at most $k$.
Therefore, $(G,S_0,k)$ is a \yes-instance for {\sc Annotated $\Fcal$-M-Elimination Distance} if and only if $(G',k)$ is a \yes-instance for {\sc $\obs(\Ccal\exc(\Fcal)))$-M-Elimination distance}.

Thus, if we apply the algorithm of \autoref{@unquestioned} to $(G',k)$ to solve {\sc $\obs(\Ccal\exc(\Fcal)))$-M-Elimination distance}, we solve {\sc Annotated $\Fcal$-M-Elimination Distance} on instance $(G,S_0,k)$ in time $2^{\Ocal_{\ell_\Fcal}(\tw(G')\cdot k+\tw(G')\log \tw(G')))}\cdot |V(G')|$.

Given a tree decomposition $\Tcal=({\sf T},\beta)$ of $G$ of width $t$, we can obtain a tree decomposition $\Tcal'$ of $G'$ of width at most $t+|V(H_\Fcal)|$ by adding a node $x_v$ for each $v\in S_0$ such that $\beta(v)=V(H_v)$, adjacent to a node $y$ of $T$ such that $v\in\beta(y)$. Thus, $\tw(G')\leq\tw(G)+|V(H_\Fcal)|=\tw(G)+\Ocal_{\ell_\Fcal}(1)$.
Moreover, $|V(G')|=|V(G)|+(|V(H_\Fcal)|-1)\cdot|S_0|=\Ocal_{\ell_\Fcal}(|V(G)|)$.
Therefore, we can solve {\sc Annotated $\Fcal$-M-Elimination Distance} on instance $(G,S_0,k)$ in time $2^{\Ocal_{\ell_\Fcal}(\tw\cdot k+\tw\log \tw))}\cdot n$, and the lemma follows.
\end{proof}

\subsection{Description of the algorithm for {\sc \texorpdfstring{$\Fcal$}{F}-M-Elimination Distance} when \texorpdfstring{$a_\Fcal=1$}{aF=1}}

We now describe the algorithm to solve {\sc Annotated $\Fcal$-M-Elimination Distance}, and hence {\sc $\Fcal$-M-Elimination Distance}, when $a_\Fcal=1$, i.e., when $\Fcal$ contains an apex-graph. Note that, similarly to this algorithm, the one from \autoref{@oeconomicus} in the general case can also very easily be generalized to its ``annotated'' version.
\sugar{We stress that the reason for the better parametric dependence of this algorithm compared to the algorithm of
\autoref{@communicated}
is that we pursue homogeneous flat walls where homogeneity is asked for subsets of size {\sl not} depending on $k$.}
%\gstam{This is new sugar! \ig{delicious sugar!}}\sed{With marmalade!}
\smallskip

We define the following constants. %!
\begin{align*}
	a =				& \ \funref{@collaboration}(s_\Fcal), &
	q = 			& \ \funref{@categories}(1,s_\Fcal,k(k+1)/2), \\
	p =       & \ \funref{@provincial}(1,s_\Fcal,k(k+1)/2), &
	l =       & \ (q-1)\cdot (k+a), \\
	d =				& \ \funref{@deliberation}(a,\ell_\Fcal) &
	r_4 =     & \ \funref{@differences}(a,\ell_\Fcal,3,k(k+1)/2), \\
	r_3 =     & \ \funref{@philistines}(r_4,a,a,d), &
	t =				& \ \funref{@corollaries}(s_\Fcal)\cdot r_3, \\
	r_2 = 		& \ \odd(t+3), &
	r_2' =    & \ \odd(\max\{ \funref{@unaffected}(1,s_\Fcal,k(k+1)/2), \funref{@idealistic}(l+1, r_2,p)\}), \\
	r_1' =		& \ \odd(\funref{@classifications}(s_\Fcal)\cdot r_2'), &
	r_1 =     & \ \odd(r_1'+k).
\end{align*}

Note that $r_4=\Ocal_{\ell_\Fcal}(k^2)$, $r_3,r_2=\Ocal_{\ell_\Fcal}(k^{2c})$, and $r_2',r_1',r_1=\Ocal_{\ell_\Fcal}(k^{2c+7/2})$, where $c=\funref{@withdrawing}(a,a,d)=\Ocal_{\ell_\Fcal}(1)$.
Recall from \autoref{@graphically} that we assume that $G$ has $\Ocal_{s_\Fcal}(k\sqrt{\log k}\cdot n)$ edges.

The input of this algorithm is a graph $G$, a set $S_0\subseteq V(G)$, and an integer $k$.

\subparagraph{Step 1.} Run the algorithm {\tt Find-Wall-Ed} from \autoref{@transforma} with input $(G\setminus S_0,r_1,k)$
and, in time $2^{\Ocal_{\ell_\Fcal}(r_1^2+k^2)}\cdot n=2^{\Ocal_{\ell_\Fcal}(k^{4c+7})}\cdot n$,
\begin{itemize}
	\item either report a \no-instance, or
	\item conclude that $\tw(G\setminus S_0)\leq \funref{@veneration}(s_\Fcal)\cdot r_1+k$ and solve {\sc Annotated $\Fcal$-M-Elimination Distance} with input $(G,S_0,\funref{@veneration}(s_\Fcal)\cdot r_1+2k,k)$ in time $2^{\Ocal_{\ell_\Fcal}((r_1+k)\cdot k+(r_1+k)\log (r_1+k)))}\cdot n=2^{\Ocal_{\ell_\Fcal}(k^{2c+9/2})}\cdot n$ using the algorithm of \autoref{@tualization}, or
	\item obtain an $r_1$-wall $W_1$ of $G$.
\end{itemize}

If the output of \autoref{@transforma} is a wall $W_1$, consider all the $\binom{r_1}{r_2}^2=2^{\Ocal_{\ell_\Fcal}(k^{2c}\log k)}$ $r_2$-subwalls of $W_1$
and for each one of them, say $W_2$, let $W_2^*$ be the central $(r_2-2)$-subwall of $W_2$ and let $D_{W_2}$ be the graph obtained from $G\setminus S_0$ after removing the perimeter of $W_2$ and taking the connected component containing $W_2^*$.
Run the algorithm {\tt Clique-or-twFlat} of \autoref{@unimportant} with input $(D_{W_2},r_3,s_\Fcal)$.
This takes time $2^{\Ocal_{\ell_\Fcal}(r_3^2)}\cdot n=2^{\Ocal_{\ell_\Fcal}(k^{4c})}\cdot n$.
If for one of these subwalls the result is a set $A$ of size at most $a$ and a regular flatness pair $(W_3,\mathfrak{R}_3)$ of $D_{W_2}\setminus A$ of height $r_3$ whose $\mathfrak{R}_3$-compass has treewidth at most $t$, then we proceed to Step~2, otherwise proceed to Step~3.

\subparagraph{Step 2.}
We obtain a 7-tuple $\mathfrak{R}_3'$ by adding all vertices of $G\setminus (S_0\cup V({\sf Compass}_{\mathfrak{R}_3}(W_3)))$ to the set in the first coordinate of $\mathfrak{R}_3$, such that $(W_3,\mathfrak{R}_3')$ is a regular flatness pair of $G\setminus (S_0\cup A)$.\smallskip

We first apply the algorithm {\tt Homogeneous} of \autoref{@disreputable}
with input $(r_4,a,a,d,t,G\setminus S_0,A,W_3,\mathfrak{R}_3')$, which outputs, in time $2^{\Ocal_{\ell_\Fcal}(r_4\log r_4 + t\log t)}\cdot(n+m)=2^{\Ocal_{\ell_\Fcal}(k^{2c}\log k)}\cdot n$ a flatness pair $(W_4,\mathfrak{R}_4)$ of $G\setminus (S_0\cup A)$ of height $r_4$ that is $d$-homogeneous with respect to $2^A$ and is a $W^*$-tilt of $(W_3,\mathfrak{R}_3')$ for some subwall $W'$ of $W$.
{At this point, we stress that the reason for the better parametric dependence of this algorithm compared to the previous one comes from the fact that the third input parameter $a$ in {\tt Homogeneous} does {\sl not} depend of $k$.}
We apply the algorithm {\tt Find-Irrelevant-Vertex} of \autoref{@civilizing}
with input $(k(k+1)/2,a,G\setminus S_0,A,W_4,\mathfrak{R}_4)$, which outputs, in time $\Ocal_{s_\Fcal}(k\sqrt{\log k}\cdot n)$, a vertex $v$ such that $(G,S_0,k)$ and $(G\setminus v,S_0,k)$
are equivalent instances of {\sc Annotated $\Fcal$-M-Elimination Distance}.
Then the algorithm runs recursively on the equivalent instance $(G\setminus v,S_0,k)$.

\subparagraph{Step 3.}
Consider all the $r_2'$-subwalls of $W_1$, which are $\binom{r_1}{r_2'}^2=2^{\Ocal_{\ell_\Fcal}(k^{2c+7/2}\log k)}$ many,
and for each of them, say $W_2'$,
compute its canonical partition $\Qcal$.
Then, contract each bag $Q$ of $\Qcal$ to a single vertex $v_Q$, remove the vertices $v_Q$ where $Q$ is not a $p$-internal bag of $\Qcal$, and add a new vertex $v_{\rm all}$ and make it adjacent to all remaining $v_{Q}$'s. In the resulting graph $G'$, for every vertex $y$ of $(G\setminus S_0)\setminus V(W_2')$, check, in time $\Ocal(q\cdot m)=\Ocal_{\ell_\Fcal}(k^7\sqrt{\log k}\cdot n)$, using a flow augmentation algorithm~\cite{Diestel10grap}, whether there are $q$ internally vertex-disjoint paths from $v_{\rm all}$ to $y$.
Let $\tilde{A}$ be the set of such $y$'s.\smallskip

If $\tilde{A}=\emptyset$, then report a \no-instance.\smallskip

If $1\leq|\tilde{A}|\leq k+a$, then each vertex of $\tilde{A}$ should intersect every $k$-elimination set $S$ of $G$ for $\exc(\Fcal)$.
The algorithm runs recursively on $(G,S_0\cup\tilde{A},k)$.\smallskip

If, for every wall, $|\tilde{A}|>k+a$, then report that $(G,S_0,k)$ is a \no-instance of {\sc Annotated $\Fcal$-M-Elimination Distance}.\bigskip

After Step~2, the size of $G$ decreases by one, so Step~2 can be applied at most $n$ times.
After Step~3, the size of $S_0$ increases by at least one, so Step~3 can also be applied at most $n$ times.
Note that, if $S_0=V(G)$, then $\tw(G\setminus S_0)=0$, so the algorithm stops.
Thus, the algorithm finishes.
Notice also that Step~3, when applied, takes time $2^{\Ocal_{\ell_\Fcal}(k^{2c+7/2}\log k)}\cdot n^2$, because we apply a flow algorithm for each of the $2^{\Ocal_{\ell_\Fcal}(k^{2c+7/2}\log k)}$ $r_2'$-subwalls and for each vertex of $G$.
Since Step~1 and Step~2 run in time $2^{\Ocal_{\ell_\Fcal}(k^{4c+7})}\cdot n$ and $2^{\Ocal_{\ell_\Fcal}(k^{2c}\log k)}\cdot n$, respectively, and both may be applied at most $n$ times, the claimed time complexity follows: the algorithm runs in time $2^{\Ocal_{\ell_\Fcal}(k^{4c+7})}\cdot n^3$.

\subsection{Correctness of the algorithm}\label{@characteristic}

Let $(G,S_0,k)$ be a \yes-instance and let $S$ be a $k$-elimination set of $G$ for $\exc(\Fcal)$ with $S_0\subseteq S$.
By running \autoref{@transforma} with input $(G\setminus S_0,r_1,k)$, the algorithm should either get a report that $\tw(G\setminus S_0)\leq \funref{@veneration}(s_\Fcal)\cdot r_1+k$ or find an $r_1$-wall.\smallskip

If $\tw(G\setminus S_0)\leq \funref{@veneration}(s_\Fcal)\cdot r_1+k$, then since $S_0\subseteq S$, $\tw(G\setminus S)\leq \funref{@veneration}(s_\Fcal)\cdot r_1+k$.
Hence, according to \autoref{@barbarians}, $\tw(G)\leq \funref{@veneration}(s_\Fcal)\cdot r_1+2k$.
\medskip

Otherwise, let $W_1$ be an $r_1$-wall of $G\setminus S_0$.
According to \autoref{@unchangeable}, since $r_1\geq r_1'+k$, there is an $r_1'$-subwall $W_1'$ of $W_1$ that is a subwall of $G\setminus S$.
Let $H$ be the connected component of $G\setminus S$ containing $W_2$.
The fact that $H$ belongs to ${\exc(\Fcal)}$ implies that it has no $K_{s_\Fcal}$-minor.
Therefore, by \autoref{@possession}, since $r_1'\geq \funref{@classifications}(s_\Fcal)\cdot r_2'$, there is a set
$B\subseteq V(H)$, with $|B|\leq a$,
and a flatness pair $(W_2',\mathfrak{R}_2')$ of $H\setminus B$ of height $r_2'$.\medskip

Let $\Qcal$ be the canonical partition of $W_2'$.
Let $G'$ be the graph obtained after contracting every bag $Q$ of $\Qcal$ to a single vertex $v_Q$, removing the vertices $v_Q$ where $Q$ is not a $p$-internal bag of $\Qcal$, and adding a new vertex $v_{\rm all}$ and making it adjacent to all remaining $v_{Q}$'s.
Let $\tilde{A}$ be the set of vertices $y$ of $G\setminus V(W_2')$ such that there are $q$ internally vertex-disjoint paths from $v_{\rm all}$ to $y$ in $G'$.
Since $S$ is a $k$-elimination set of $G$ for $\exc(\Fcal)$,
there is a set $P\subseteq S$ of size at most $k$ so that $(L,R):=(V(G)\setminus V(H),V(H)\cup P)$ is a separation of $G$ with $P=L\cap R$.\smallskip

Note that $\tilde{A}\subseteq P\cup B$.
To show this, we first prove that, for every $y\notin P\cup B$, the maximum number of internally vertex-disjoint paths from $v_{\sf all}$ to $y$ in $G'$ is
less than $q$.
Indeed,
if $y$ is a vertex in $(V(G)\setminus V(H))\setminus P$, then every path from $y$ to a vertex of $W_2'$ intersects $P$.
Therefore, there are at most $k<q$ internally vertex-disjoint paths from $v_{\sf all}$ to such a $y\in (V(G)\setminus V(H))\setminus P$ in $G'$.
If $y\in V(H)\setminus B$, then we distinguish two cases.
First, if $y$ is a vertex in the $\mathfrak{R}_2'$-compass of $W_2'$,
there are at most $k+a$ such paths that intersect the set $P\cup B$ and
at most four paths that do not intersect $P\cup B$ (in the graph $G'\setminus (P\cup B)$)
due to the flatness of $W_2'$.
If $y$ is in $V(H)$ but not a vertex in the $\mathfrak{R}_2'$-compass of $W_2'$, then, since by the definition of flatness pairs the perimeter of $W_2'$ together with the set $P\cup B$ separate $y$ from the $\mathfrak{R}_2'$-compass of $W_2'$,
every collection of internally vertex-disjoint paths from $v_{\rm all}$ to $y$ in $G'$ should intersect the set $\{v_{Q_{\rm ext}}\}\cup P\cup B$, where $Q_{\rm ext}$ is the external bag of $\Qcal$. Therefore, in all cases, if $y\notin P\cup B$, the maximum number of internally vertex-disjoint paths from $v_{\sf all}$ to $y$ in $G'$ is at most $k+a+4<q$.
Therefore, $y\notin\tilde{A}$.
Hence, $|\tilde{A}|\leq k+a$.\smallskip

Let $\mathfrak{R}_2''$ be the 7-tuple obtained by adding all vertices of $((G\setminus S_0)\setminus P)\setminus H$ to the set in the first coordinate of $\mathfrak{R}_2'$. Notice that since every path between $G\setminus H$ and $H$ intersects $P$, $(W_2',\mathfrak{R}_2'')$ is a flatness pair of $G\setminus (P\cup B)$.\medskip

If $\tilde{A}=\emptyset$, then let $\tilde{\Qcal}$ be an enhancement of $\Qcal$ on $G\setminus (P\cup B)$.
No vertex of $(P\cup B)\setminus S_0$ is adjacent to vertices of at least $q$ $p$-internal bags of $\tilde\Qcal$.
This means that the $p$-internal bags
of $\tilde\Qcal$ that contain vertices adjacent to some vertex of $P\cup B$ are at most $(q-1)\cdot(k+a)=l$.\smallskip

Consider a family $\Wcal=\{W^1, \ldots, W^{l+1}\}$ of $l+1$ $r_2$-subwalls of $W_2'$ such that for every $i \in [l+1]$, $\bigcup \influence_{\mathfrak{R}_2''}(W^i)$ is a subgraph of $\bigcup \{Q\mid Q \text{ is a $p$-internal bag of }\tilde\Qcal\}$ and for every $i,j\in[l+1]$, with $i\neq j$, there is no internal bag of $\tilde\Qcal$ that contains vertices of both $V(\bigcup \influence_{\mathfrak{R}_2''} (W^i))$ and $V(\bigcup \influence_{\mathfrak{R}_2''} (W^j))$. The existence of $\Wcal$ follows from \autoref{@prohibitions} and the fact that $r_2'\geq \funref{@idealistic}(l+1,r_2,p)$.\smallskip

The fact that the $p$-internal bags
of $\tilde\Qcal$ that contain vertices adjacent to some vertex of $(P\cup B)\setminus S_0$ are at most $l$ implies that
there exists an $i\in[l+1]$ such that
no vertex of $V(\bigcup \influence_{\mathfrak{R}_2''}({W^i}))$ is adjacent, in $G$, to a vertex in $(P\cup B)\setminus S_0$.\smallskip

Let $W_2:=W^i$, let $W_2^*$ be the central $(r_2-2)$-subwall of $W_2$, and let $D_{W_2}$ be the graph obtained from $G\setminus S_0$ after removing the perimeter of $W_2$ and taking the connected component containing $W_2^*$.
Any path going from a vertex in $H$ to a vertex in $G\setminus H$ intersects $P$.
Thus, $D_{W_2}\subseteq H$ and therefore, $K_{s_\Fcal}$ is not a minor of $D_{W_2}$.
Moreover, $W_2^*$ is a wall of $D_{W_2}$ of height $r_2-2\geq t+1$, so $\tw(D_{W_2})>t=\funref{@corollaries}(s_\Fcal)\cdot r_3$.
Therefore, since , if the algorithm runs {\tt Clique-or-twFlat} of \autoref{@unimportant} with input $(D_{W_2},r_3,s_\Fcal)$, it should obtain a set $A$ of size at most $a$ and a regular flatness pair $(W_3,\mathfrak{R}_3)$ of $D_{W_2}\setminus A$ of height $r_3$ whose $\mathfrak{R}_3$-compass has treewidth at most $t$.
Hence, the algorithm then runs Step~2.\medskip

If $\tilde{A}\neq\emptyset$, then recall that for every $y\in\tilde{A}$, $y$ has $q$ internally vertex-disjoint paths $P_1,...,P_q$ to different $p$-internal bags $Q_1,...,Q_q$ of $\Qcal$ in $G$.
Hence, there is an enhancement $\tilde{\Qcal}_y$ of $\Qcal$ on $G\setminus (P\cup B)$ such that $P_i$ belongs to the bag $\tilde{Q}_i$ that extends $Q_i$ for $i\in[q]$.
Therefore, $y$ is adjacent to vertices of at least $q$ $p$-internal bags of $\tilde{\Qcal}_y$.
Let $S'$ be a $k$-elimination set of $G$ for $\exc(\Fcal)$.
According to \autoref{@lucinatory}, there is a set $X_{S'}\subseteq V(G)$ such that $G\setminus X_{S'}\in\exc(\Fcal)$ and $\bid_{G\setminus (P\cup B),W_2'}(X_{S'})\leq k(k+1)/2$. Therefore, $y\in X_{S'}$ due to \autoref{@proclamation} and the fact that $r_2'\geq \funref{@unaffected}(1,s_\Fcal,k(k+1)/2)$.
Let $C_{S'}:=G\setminus X_{S'}$.
Recall that $y$ is adjacent to $q>k(k+1)/2$ $p$-internal bags of $\tilde{\Qcal}_y$.
However, $\bid_{\tilde{\Qcal}_y}(X_{S'})\leq\bid_{G\setminus (P\cup B),W_2'}(X_{S'})\leq k(k+1)/2$.
Therefore, $y$ is adjacent to $C_{S'}$, so $y\in S'$.
Since for every $y\in\tilde{A}$, for every $k$-elimination set $S'$, we have $y\in S'$, it implies that $\tilde{A}$ is included in every $k$-elimination set of $G$ for $\exc(\Fcal)$.
Hence, if the algorithm runs Step~3, it then recursively runs on the equivalent instance $(G,S_0\cup\tilde{A},k)$.\medskip

We do not suppose that $(G,k)$ is a \yes-instance anymore.
Let us show the correctness of Step~2.
Suppose that we obtained the wanted flatness pair $(W_3,\mathfrak{R}_3)$ in Step~1.
We obtain a 7-tuple $\mathfrak{R}_3'$ by adding all vertices of $G\setminus (S_0\cup V({\sf Compass}_{\mathfrak{R}_3}(W_3)))$ to the set in the first coordinate of $\mathfrak{R}_3$. Since $(W_3,\mathfrak{R}_3)$ is a regular flatness pair of $D_{W^i}\setminus A$ whose $\mathfrak{R}_3$-compass has treewidth at most $t$ and since the vertices added in $\mathfrak{R}_3'$ are only adjacent to the perimeter of $W^i$, it follows that $(W_3,\mathfrak{R}_3')$ is a regular flatness pair of $G\setminus (S_0\cup A)$ whose $\mathfrak{R}_3'$-compass has treewidth at most $t$.\smallskip

If the algorithm applies the algorithm {\tt Homogeneous} of \autoref{@disreputable} with  $(r_4,a,a,d,t,G\setminus S_0,A,W_3,\mathfrak{R}_3')$ as input, it obtains a flatness pair $(W_4,\mathfrak{R}_4)$ of $G\setminus (S_0\cup A)$ of height $r_4$ that is $d$-homogeneous with respect to $2^A$ and is a $W^*$-tilt of $(W_3,\mathfrak{R}_3')$ for some subwall $W'$ of $W$.
According to \autoref{@successively}, $(W_4,\mathfrak{R}_4)$ is regular.\smallskip

\autoref{@lucinatory} implies that for every $k$-elimination set $S'\supseteq S_0$, there is a set $X_{S'}\supseteq S'$ with $\bid_{G\setminus (S_0\cup A),W_4}(X_{S'})\leq k(k+1)/2$ and $G\setminus X_{S'}\in\exc(\Fcal)$.
We have that $|A\setminus X|\leq|A|\leq a$, so the algorithm can apply {\tt Find-Irrelevant-Vertex} of \autoref{@civilizing} with input $(k(k+1)/2,a,G\setminus S_0,A,W_4,\mathfrak{R}_4)$ to obtain a vertex $v$ such that for every $k$-elimination set $S'\supseteq S_0$, $G\setminus X_{S'}\in{\exc(\Fcal)}$ if and only if $G\setminus (X_{S'}\setminus v)\in{\exc(\Fcal)}$.
It follows that $(G,S_0,k)$ and $(G\setminus v,S_0,k)$ are equivalent instances of {\sc Annotated $\Fcal$-M-Elimination Distance}.\medskip

Suppose now that $(G,S_0,k)$ is a \no-instance.
In Step~1, the algorithm either reports a \no-instance or finds a wall.
In the latter case, the algorithm either goes to Step~2 or to Step~3.
If it runs Step~2, the previous paragraph justifies that the algorithm finds a vertex $v$ such that $(G\setminus v,S_0,k)$ is a \no-instance.
If the algorithm runs Step~3, then it either reports a \no-instance or recursively runs on the instance $(G\setminus y,S_0\cup\tilde{A},k)$.
If $(G\setminus y,S_0\cup\tilde{A},k)$ is \yes-instance, then so is $(G,k)$. Thus, $(G\setminus y,S_0\cup\tilde{A},k)$. is a \no-instance.
Hence, the algorithm always report a \no-instance.
Therefore, \autoref{@overflowing} follows.

\paragraph{Constructing the elimination ordering.}
Notice that the results of \autoref{@successors} and  \autoref{@overflowing} solve the decision version of
 {\sc Elimination Distance to $\Gcal$}. Using the dynamic programming algorithm of \autoref{@fanaticism},
 we may find a $k$-elimination set  $X$ certifying that
 $\ed_{\Gcal}(G)\leq k$. One may further determine, from
 $X$, the way the elimination ordering is applied on the vertices of $X$ as follows. Let ${\sf torso}(G,X)$ be the
 graph obtained from   $G[X]$  if,  for every connected component $C$ of $G-X$,  we make adjacent all pairs of vertices
 in $N_{G}(V(C))$ in $G[X]$.
  Then we know that $\td({\sf torso}(G,X))\leq k$ and the required elimination ordering
 is the same as the one for ${\sf torso}(G,X)$, which can be computed by the algorithm of~\cite{ReidlRSS14afas} in time $2^{\Ocal(k^2)}\cdot n$.

\section{Bounding the obstructions of \texorpdfstring{$\Ecal_k(\exc(\Fcal))$}{Ek(exc(F))}}\label{@participant}

In this section, we prove the following result that provides an upper bound on the size of the graphs in $\obs(\Ecal_k (\exc(\Fcal)))$. The following theorem is a reformulation of \autoref{@reconstructions}.

\begin{theorem}\label{@delectable}
Let $\Fcal$ be a non-empty finite collection of non-empty graphs and $k$ be a positive integer.
Every graph in $\obs(\Ecal_k (\exc(\Fcal)))$ has $2^{2^{2^{2^{k^{\Ocal_{\ell_\Fcal}(1)}}}}}$ vertices.
In the particular case when $\Fcal$ contains an apex-graph, every graph in $\obs(\Ecal_k (\exc(\Fcal)))$ has $2^{2^{k^{\Ocal_{\ell_\Fcal}(1)}}}$ vertices.
\end{theorem}

Recall that when $\Fcal=\{K_{1}\}=\obs(\Gcal_\emptyset)$, it is known \cite{DvorakGT12forb} that every graph in $\obs(\Ecal_k(\Gcal_\emptyset))$ has at most $2^{2^{k-1}}$ vertices.\smallskip

\autoref{@delectable} implies that one can construct an algorithm that receives as input $\obs(\exc(\Fcal))$ and $k$, and outputs $\obs(\Ecal_k (\exc(\Fcal)))$. This is done by enumerating all graphs on at most $f(k)$ vertices, where $f(k)$ is the bound on the number of vertices given by \autoref{@delectable}, and filtering out those that are members of $\Ecal_k (\exc(\Fcal))$ and taking those that are minor-minimal in what is left.
The running time of the algorithm can be bounded by $\Ocal(2^{2^{2^{2^{2^{k^{\Ocal_{\ell_\Fcal}(1)}}}}}}\cdot n^2)$ in the general case and by $\Ocal(2^{2^{2^{k^{\Ocal_{\ell_\Fcal}(1)}}}}\cdot n^2)$ if $\Fcal$ contains an apex-graph.

Note that this brute-force algorithm can be used to solve {\sc $\Fcal$-M-Elimination Distance}. Indeed, to solve {\sc $\Fcal$-M-Elimination Distance}, we can compute $\obs(\Ecal_k (\exc(\Fcal)))$ and then check whether there is a graph in $\obs(\Ecal_k (\exc(\Fcal)))$ that is a minor of the input graph. {Of course, this algorithm is much less efficient than the ones presented in the previous sections.}

The rest of the section is structured as follows: in \autoref{@quantitatively} we bound the treewidth of a minor-minimal obstruction of $\Ecal_k (\exc(\Fcal))$, while in \autoref{@typescript} we bound the size of a minor-minimal obstruction of $\Ecal_k (\exc(\Fcal))$ of small treewidth.
This immediately implies \autoref{@delectable}.

\subsection{Bounding the treewidth of an obstruction}\label{@quantitatively}

In this subsection we aim to prove an upper bound on the treewidth of a minor-minimal obstruction of $\Ecal_k (\exc(\Fcal))$.

\begin{lemma}\label{@libertines}
Let $\Fcal$ be a finite collection of graphs.
There exists a function $\newfun{@deceptively}:\bN^3\to\bN$ such that if $G\in\obs(\Ecal_k ({\exc(\Fcal)}))$,
then $\tw(G)\leq \funref{@deceptively}(k,a_\Fcal,s_\Fcal)$.
Moreover, $\funref{@deceptively}(k,a,s)=2^{\log(k\cdot c)\cdot 2^{k^{a-1}\cdot 2^{\Ocal(s^2\log s)}}}$, where $a=a_\Fcal$, $s=s_\Fcal$, and $c$ is a constant depending on $\ell_\Fcal$.
\end{lemma}

Note that when $a_\Fcal=1$, $\funref{@deceptively}(k,1,s)=\Ocal_{\ell_\Fcal}(k^{2^{2^{\Ocal(s^2\log s)}}})$.

\begin{proof}
For simplicity, we use $s,a$, and $\ell$ instead of $s_{\cal F}, a_{\cal F}$, and $\ell_{\cal F}$, respectively.
We set
\begin{align*}
b= &\ \funref{@collaboration}(s) + k + 1, &
d= &\ \funref{@deliberation}(a-1,\ell), \\
r_4= &\ \funref{@differences}(a-1,\ell,3,k(k+1)/2), &
r_3= &\ \funref{@philistines}(r_4,a-1,b,d), \\
x= &\ \funref{@categories}(a,s,k(k+1)/2), &
p= &\ \funref{@provincial}(a,s,k(k+1)/2), \\
l= &\ (x-1)\cdot b, &
r_2= &\ \odd(\max\{\funref{@idealistic}(l+1,r_3,p),\funref{@unaffected}(a,s,k(k+1)/2)\}), \\
r_1= &\ \funref{@classifications}(s)\cdot r_2, \mbox{ and } &
w= & \ \funref{@relentlessly}(s)\cdot r_1 + k + 1.
\end{align*}
It is easy to verify that $w=2^{\log(k\cdot c)\cdot 2^{k^{a-1}\cdot 2^{\Ocal(s^2\log s)}}}$.\medskip

Suppose towards a contradiction that $\tw(G)>w$.
Since $G\in\obs(\Ecal_k ({\exc(\Fcal)}))$, for each $v\in V(G)$, $G\setminus v\in\Ecal_k ({\exc(\Fcal)})$.
Therefore, there exists a $(k+1)$-elimination set $S$ of $G$ for $\exc(\Fcal)$.
Thus, for each $C\in\cc(G\setminus S)$, $C\in{\exc(\Fcal)}$.
According to \autoref{@barbarians}, $\tw(G\setminus S)>w-k-1=\funref{@relentlessly}(s)\cdot r_1$, so there is $C\in\cc(G\setminus S)$, such that $\tw(C)>\funref{@relentlessly}(s)\cdot r_1$.
Moreover, $K_s$ is not a minor of $C$.
Therefore, according to \autoref{@overwhelmingly}, $C$ contains an $r_1$-wall $W_1$.

Since $r_1=\funref{@classifications}(s)\cdot r_2$, by \autoref{@possession}, there is a set $A\subseteq V(C)$ of size at most $\funref{@collaboration}(s)$ and a flatness pair $(W_2,\mathfrak{R}_2)$ of $C\setminus A$ of height $r_2$ such that $W_2$ is a tilt of a subwall of $W_1$.
Due to \autoref{@philosophic}, there is a regular flatness pair $(W_2',\mathfrak{R}_2')$ of $C\setminus A$ of height $r_2$.

Since $S$ is a $(k+1)$-elimination set of $G$ for $\exc(\Fcal)$ and $C\in\cc(G\setminus S)$, there exists a set $P\subseteq S$ of size at most $k+1$ such that $(L,R):=(V(G)\setminus V(C),V(C)\cup P)$
is a separation of $G$ with $L\cap R=P$.
Thus, if $\mathfrak{R}_2''$ is the 7-tuple obtained by adding the vertices of $G\setminus (C\cup P)$ to the set
in the first coordinate of $\mathfrak{R}_2'$, $(W_2',\mathfrak{R}_2'')$ is a regular flatness pair of $G\setminus (A\cup P)$ of height $r_2$.

Let $\tilde{\Qcal}$ be a $W_2'$-canonical partition of $G\setminus (A\cup P)$.
Let $B$ be the set of vertices of $A\cup P$ adjacent to vertices of at least $x$ $p$-internal bags of $\tilde{\Qcal}$.
Let $\Wcal=\{W^1,...,W^{l+1}\}$ be a family of $l+1$ $r_3$-subwalls of $W_2'$ such that for every $i\in [l+1]$, $\cupall\influence_{\mathfrak{R}_2''}(W^i)$ is a subgraph of $\cupall \{Q\mid Q \text{ is a $p$-internal bag of }\tilde{\Qcal}\}$ and for every $i,j\in[l+1]$ with $i\neq j$, there is no internal bag $Q\in\tilde{\Qcal}$ that contains vertices of both $V(\cupall\influence_{\mathfrak{R}_2''}(W^i))$ and $V(\cupall\influence_{\mathfrak{R}_2''}(W^j))$.
The existence of $\Wcal$ follows from the fact that $r_2\geq\funref{@idealistic}(l+1,r_3,p)$ and \autoref{@prohibitions}.
Notice that the set $N_G((A\cup P)\setminus B)$ intersects the vertex set of at most $(x-1)\cdot |(A\cup P)\setminus B|\leq l$ $p$-internal bags of $\tilde{\Qcal}$.
Thus, there is an $i\in[l+1]$ such that no vertex in $(A\cup P)\setminus B$ is adjacent to vertices of $\cupall\influence_{\mathfrak{R}_2''}(W^i)$.

Let $(W_3,\mathfrak{R}_3)$ be a $W^i$-tilt of $(W_2',\mathfrak{R}_2'')$.
Since $|B|\leq|A\cup P|\leq \funref{@collaboration}(s)+k+1=b$ and $r_3=\funref{@philistines}(r_4,b,a-1,d)$, by \autoref{@disreputable}, there is a flatness pair $(W_4,\mathfrak{R}_4)$ of $G\setminus B$ of height $r_4$ that is $d$-homogeneous with respect to $\binom{B}{<a}$ and is a tilt of a subwall of $(W_3,\mathfrak{R}_3)$.
By \autoref{@unimagined} and \autoref{@successively}, $(W_4,\mathfrak{R}_4)$ is regular.

Recall that $(W_3,\mathfrak{R}_3)$ is a $W^i$-tilt of $(W_2',\mathfrak{R}_2'')$, $(W_4,\mathfrak{R}_4)$ is a tilt of a subwall of $(W_3,\mathfrak{R}_3)$, and $(W_4,\mathfrak{R}_4)$ is a flatness pair of $G\setminus B$ with $B\subseteq A\cup P$. Thus, given a $W_4$-canonical partition $\tilde{\Qcal}_1$ of $G\setminus B$, there is a $W_2'$-canonical partition $\tilde{\Qcal}_2$ of $G\setminus (A\cup P)$ such that each internal bag of $\tilde{\Qcal}_1$ is contained in an internal bag of $\tilde{\Qcal}_2$.
Therefore, for every set $X\subseteq V(G)$, $\bid_{G\setminus B,W_4}(X)\leq\bid_{G\setminus (A\cup P),W_2'}(X)$.

Moreover, since $r_2\geq \funref{@unaffected}(a,s,k(k+1)/2)$, according to \autoref{@proclamation},
 every subset of $B$ of size $a$ intersects every set $X\subseteq V(G)$ such that $G\setminus X\in{\exc(\Fcal)}$ and $\bid_{G\setminus (A\cup P),W_2'}(X)\leq k(k+1)/2$.
Hence, for any such $X$, $|B\setminus X|<a$.

Thus, according to \autoref{@civilizing},
since $r_4=\funref{@differences}(a-1,\ell,3,k(k+1)/2)$, it holds that there is a vertex $v$ such that, for every set $X\subseteq V(G)$ with $G\setminus X\in{\exc(\Fcal)}$ and $\bid_{G\setminus B,W_4}(X)\leq k(k+1)/2$, $G\setminus X\in{\exc(\Fcal)}$ if and only if $G\setminus (X\setminus v)\in{\exc(\Fcal)}$.

\autoref{@lucinatory} implies that for any $k$-elimination set $S'$ of $G\setminus v$ for $\exc(\Fcal)$, there is a set $X\supseteq S'$ such that $G\setminus X\in\exc(\Fcal)$ and $\bid_{G\setminus (A\cup P),W_2'}(X)\leq k(k+1)/2$.
Since $\bid_{G\setminus B,W_4}(X)\leq\bid_{G\setminus (A\cup P),W_2'}(X)$, we also have that
$\bid_{G\setminus B,W_4}(X)\leq k(k+1)/2$.
Thus, $G\setminus X\in{\exc(\Fcal)}$ if and only if $G\setminus (X\setminus v)\in{\exc(\Fcal)}$ and $G\in\Ecal_k({\exc(\Fcal)})$ if and only if $G\setminus v\in\Ecal_k({\exc(\Fcal)})$.
However, since $G\in\obs(\Ecal_k({\exc(\Fcal)}))$, it holds that $G\notin\Ecal_k({\exc(\Fcal)})$ and $G\setminus v\in\Ecal_k({\exc(\Fcal)})$, a contradiction.
\end{proof}

\subsection{Bounding the size of an obstruction of small treewidth}\label{@typescript}

In order to bound the size of an obstruction of small treewidth, we first present some additional notions on tree decompositions on boundaried graphs.

\subparagraph{Treewidth of boundaried graphs.}
Let ${\bf G}=(G,B,\rho)$ be a boundaried graph.
A {\em tree decomposition} of ${\bf G}$ is a rooted tree decomposition $({\sf T},\beta, {\sf r})$ of $G$
such that $\beta({\sf r})=B$.
The {\em width} of $({\sf T},\beta, {\sf r})$ is the width of $({\sf T},\beta)$.
The treewidth of a boundaried graph ${\bf G}$ is the minimum width over all its tree decompositions and is denoted by $\tw({\bf G})$.
A \emph{nice tree decomposition} of ${\bf G}$ is a tree decomposition $({\sf T},\beta, {\sf r})$ of ${\bf G}$ that is also a nice tree decomposition of $G$ rooted at~${\sf r}$.

Let ${\bf G}=(G,B,\rho)$ be a boundaried graph and $\Tcal=({\sf T},\beta, {\sf r})$ be a tree decomposition of ${\bf G}$.
Notice that if $a,b\in V({\sf T})$ and $a\in\anc_{{\sf T},{\sf r}}(b)$, then $G_b$ is a subgraph of $G_a$.
We define the $t_q$-boundaried graph $\bar{\bf G}_q = (\bar{G}_q, \beta(q), \rho_q)$, where $\bar{G}_q = G\setminus (V(G_q)\setminus \beta(q))$.
Notice that ${\bf G}_q$ and $\bar{\bf G}_q$ are compatible and ${\bf G}_q \oplus \bar{\bf G}_q = G$.

\subparagraph{Linked tree decompositions.}

Our next step is to use a special type of tree decompositions, namely {\sl linked tree decompositions}, defined by Robertson and Seymour in~\cite{RobertsonS86V}.
Thomas in~\cite{Thomas90amen} proved that every graph $G$ admits a linked tree decomposition of width $\tw(G)$ (see also \cite{BellenbaumD02twos,Erde18auni}).
By combining the result of \cite{Thomas90amen} and \cite[Lemma 4]{ChatzidimitriouTZ20spar},
we can consider tree decompositions as asserted in the following result.

\begin{proposition}[\cite{ChatzidimitriouTZ20spar}]\label{@irrationality}
	Let $t\in\bN_{\geq 1}$.
	For every boundaried graph ${\bf G}=(G,B,\rho)$ of treewidth $t-1$,
	there exists a tree decomposition $({\sf T},\beta,{\sf r})$ of ${\bf G}$ of width $t-1$ such that
	\begin{enumerate}

		\item $({\sf T},{\sf r})$ is a binary tree,

		\item for every $a,b\in V({\sf T})$
		      where $a$ is a child of $b$ in $({\sf T},{\sf r})$, if $|\beta(a)|=|\beta(b)|$
		      then $G_a$ is a {\sl proper} subgraph of $G_b$,
		      i.e., $|V(G_a)|<|V(G_b)|$,

		\item for every $s\in\bN$ and every pair $u_1, u_2\in V({\sf T})$, where $u_1\in\anc_{{\sf T},{\sf r}}(u_2)$ and $|\beta(u_1)|=|\beta(u_2)|$,
		      either there is an internal vertex $w$ of $u_1 {\sf T} u_2$ such that $|\beta(w)|< s$,
		      or there exists a collection of $s$ vertex-disjoint paths in $G$
		      between $\beta(u_1)$ and $\beta(u_2)$, and

		\item $|V(G)|\leq t\cdot |V({\sf T})|$.
	\end{enumerate}
\end{proposition}

In fact, linked tree decompositions are defined as the tree decompositions satisfying only property $(3)$ \cite{RobertsonS86V,Thomas90amen}.
In our proofs, we will need the extra properties $(1)$, $(2)$, and $(4)$ that are provided by \cite[Lemma 4]{ChatzidimitriouTZ20spar}.
\medskip

We  bound the size of a minor-minimal obstruction of small treewidth in \autoref{@entrepreneurial}.
To do so, we need the following result ({for a proof see e.g. \cite[Lemma 14]{GiannopoulouPRT19cutw}}).

\begin{proposition}[\cite{GiannopoulouPRT19cutw}]\label{@impervious}
	Let $r,m\in\bN_{\geq 1}$ and $w$ be a word of length $m^r$ over the alphabet $[r]$.
	Then there is a number $k\in[r]$ and a subword $u$ of $w$ such that $u$ contains only numbers not smaller than $k$ and $u$ contains the number $k$ at least $m$ times.
\end{proposition}

We are now ready to prove \autoref{@entrepreneurial}.
The idea is to apply the technique of Lagergren \cite{Lagergren98uppe} combined with the bound
on the number of characteristics provided in \autoref{@metropoltheater}. The proof of \autoref{@entrepreneurial} is very similar to the corresponding proof in~\cite{SauST21kapiI} for $\obs(\Acal_k({\exc(\Fcal)}))$.

\begin{lemma}\label{@entrepreneurial}
Let $\Fcal$ be a finite non-trivial collection of graphs.
There exists a function $\newfun{@uncertainty}:\bN^2\to\bN$ such that if $k$ is an integer and $G$ is a graph in $\obs(\Ecal_k({\exc(\Fcal)}))$ of treewidth $\tw$, then $|V(G)|\leq \funref{@uncertainty}(\tw,k)$.
Moreover, $\funref{@uncertainty}(t,k)=2^{2^{\Ocal_{\ell_\Fcal}(t^3+k\cdot t^2)}}$.
\end{lemma}

\begin{proof}
Let $G\in \obs(\Ecal_k({\exc(\Fcal)}))$.
We set $t:=\tw(G)+1$.
For simplicity, we use $\ell$ instead of $\ell_{\cal F}$.
We set
\begin{align*}
d= & \ \funref{@caricatures}(t-1,k)+1, &
m= & \ 2^{\binom{t}{2}}\cdot(d-1)+1, \\
x= & \ m^t, \mbox{ and} &
b= & \ t\cdot 2^x.
\end{align*}
It is easy to verify that $b=2^{2^{\Ocal_{\ell_\Fcal}(t^3+k\cdot t^2)}}$.\medskip

Suppose towards a contradiction that $|V(G)|>b$.
Let $({\sf T},\beta)$ be a tree decomposition of $G$ of width $\tw(G)$ and let ${\sf r}\in V({\sf T})$.
We consider the rooted tree $({\sf T},{\sf r})$ and we set $B:= \beta({\sf r})$ and a bijection $\rho:B\to [|B|]$.
We set ${\bf G}=(G,B,\rho)$ and observe that $({\sf T},\beta,{\sf r})$ is a tree decomposition of ${\bf G}$ of width $\tw(G)$.
Since $\tw({\bf G})= \tw(G) = t-1$, by \autoref{@irrationality},
we can assume that for the tree decomposition $({\sf T},\beta,{\sf r})$ of ${\bf G}$ of width $t-1$,
Properties (1) to (4) are satisfied.

Since $|V(G)|>b= t\cdot 2^x$, Property (4) implies that $|V({\sf T})|> 2^x$.
Also, by Property (1), $({\sf T},{\sf r})$ is a binary tree and therefore there exists a leaf $u$ of ${\sf T}$ such that
$|V(r{\sf T}u)|\geq x$.
We set $l:=|V(r{\sf T}u)|$.

We set $v_1= r$ and for every $i\in [l-1]$, we set $v_{i+1}$ to be the child of $v_i$ in $(T,r)$ that belongs to $V(rTu)$.
Keep in mind that $v_l = u$.
For every $i\in [l]$, we set $c_i:=|\beta(v_i)|$ and observe that, since $({\sf T},\beta,{\sf r})$ has width $t-1$, $c_i\in[t]$.

Let $C$ be the word $c_1\cdots c_x$.
Since $x=m^{t}$ and every $c_i \in [t]$, then, due to \autoref{@impervious},
there is a $t'\in[t]$ and a subword $C'$ of $C$ such that,
for every $c$ in $C'$, $c\geq t'$ and there are at least $m$ numbers in $C'$ that are equal to $t'$.
Therefore, there exists a set $\{z_1,\ldots, z_m\}\subseteq V({\sf T})$ such that for every
$i\in[2,m]$, $z_i$ is a descendant of $z_{i-1}$ in $({\sf T},{\sf r})$, for every $z'\in V(z_1 {\sf T} z_m)$
it holds that $|\beta (z')|\geq t'$, and, for every $i\in[m]$, $|\beta( z_i )| = t'$.
Hence, Property (3) of the tree decomposition $({\sf T},\beta,{\sf r})$ of ${\bf G}$ implies that
there exists a collection ${\cal P}=\{{P}_1,\ldots, {P}_{t'}\}$
of $t'$ vertex-disjoint paths in $G$ between $\beta (z_1)$ and $\beta (z_m)$.

For every $i\in [m]$, let $\rho_i$ be the function mapping a vertex $v$ in $\beta(z_i)$ to the index of the path of ${\cal P}$
it intersects, i.e., for every $j\in [t']$, if $v$ is a vertex in
$V(P_j)\cap \beta(z_i)$, where $P_{j}\in {\cal P}$, then $\rho_i (v)= j$.
Also, for every $i\in[m]$, let ${\bf G}_{z_i}$ be the $t'$-boundaried graph $(G_{z_i}, \beta(z_i), \rho_i)$.
Since, $m=2^{\binom{t}{2}}\cdot(d-1)+1$, there is a set $J\subseteq [m]$ of size $d$ such that for every $i,j\in J$, the graph $G_{z_i}[\beta(z_i)]$ is isomorphic to the graph $G_{z_j}[\beta(z_j)]$.
Therefore, for every $i,j\in J$, ${\bf G}_{z_i}$ and ${\bf G}_{z_j}$ are compatible.
Furthermore, observe that for every $i,j\in J$ with $i\leq j$, ${\bf G}_{z_j} \preceq {\bf G}_{z_i}$.
To see why this holds, for every $i,j\in J$ with $i< j$, let ${\cal P}_{i,j}$ be the collection of subpaths of ${\cal P}$ between the vertices of $\beta(z_i)$ and $\beta(z_j)$ and consider the graph $G_{z_i}[\beta(z_i)] \cup \cupall{\cal P}_{i,j} \cup G_{z_j}$, which is a subgraph of $G_{z_i}$.
By contracting the edges in ${\cal P}_{i,j}$, we obtain a boundaried graph isomorphic to ${\bf G}_{z_j}$.

Recall that $|J| = d= \funref{@caricatures}(t-1,k)+1$.
Thus, by \autoref{@belligerent}, there exist $i,j\in J$ such that $j$ is the smallest element in $J$ that is greater than $i$ and
$\Char_k({\bf G}_{z_i}) = \Char_k({\bf G}_{z_j})$.
For simplicity, we set $a:=z_{i}$ and $b:=z_{j}$.
Notice that, in $G$, by contracting the edges of the paths in $\Pcal_{i,j}$ and removing the vertices of $G_a$ that are not vertices of $G_b$, we obtain a graph isomorphic to $\bar{\bf G}_{a} \oplus {\bf G}_{b}$.
Therefore, $\bar{\bf G}_{a} \oplus {\bf G}_{b}$ is a minor of $G$.
Furthermore, $|V(\bar{\bf G}_a\oplus {\bf G}_b)|<|V(G)|$.
To prove this, we argue that $G_b$ is a proper subgraph of $G_a$.
First recall that for every $y\in V(a{\sf T}b)$, $|\beta(y)|\geq t'$.
If there is a $y\in V(a{\sf T}b)$ such that $|\beta(y)|>t'$, then there is a vertex $v\in \beta(y)$ that is a vertex of $V(G_a)\setminus V(G_b)$ and thus $G_b$ is a proper subgraph of $G_a$, while in the case where for every $y\in V(a{\sf T}b)$, $|\beta(y)| = t'$, Property (2) of \autoref{@irrationality} implies that $G_b$ is a proper subgraph of $G_a$.

Let $G' := \bar{\bf G}_a \oplus {\bf G}_b$. Since $|V(G')|<|V(G)|$, $G'$ is a minor of $G$, and $G\in \obs({\cal E}_k (\exc({\cal F})))$, it holds that $G'\in{\cal E}_k (\exc({\cal F}))$.
By \autoref{@manuscripts}, since $\Char_k({\bf G}_a)=\Char_k({\bf G}_b)$, we have that $\ed_{\exc(\Fcal)}(\bar{\bf G}_a \oplus {\bf G}_a)=\ed_{\exc(\Fcal)}(\bar{\bf G}_a \oplus {\bf G}_b)$ and therefore $\ed_{\exc(\Fcal)}(G)=\ed_{\exc(\Fcal)}(G')$.
This contradicts the fact that $\ed_{\exc(\Fcal)}(G)>k$ and $\ed_{\exc(\Fcal)}(G')\leq k$.
\end{proof}

\section{Concluding remarks}
\label{@entwickltmg}

For a minor-closed graph class $\Gcal$, we proved that {\sc Vertex Deletion to $\Gcal$} can be solved in time $2^{\poly(k)}\cdot n^2$ and that {\sc Elimination Distance to $\Gcal$} can be solved in time $2^{2^{2^{\poly(k)}}}\cdot n^2$, and in time
$2^{2^{c \cdot k^2\log k}}\cdot n^2$ and $2^{\poly(k)}\cdot n^3$ in the case where the obstruction set of $\Gcal$ contains an apex-graph.
Here the degree of {\poly} and $c$ heavily depend on the size of the obstructions of $\Gcal$.
An open question is whether $\poly(k)$ could be replaced by $c \cdot k^d$ for some constant  $c$ depending on $\Gcal$ and
 some {\sl universal} constant $d$ (independent of  $\Gcal$).
We tend to believe that this dependence on $\Gcal$ in the exponent of the polynomial is unavoidable, at least if we want to use the irrelevant vertex technique, and specially our definition of homogeneity.

On the other hand, we are not aware, for any of the two considered problems, of any  lower bound, assuming the Exponential Time Hypothesis~\cite{ImpagliazzoP01whic}, stronger than $2^{o(k)}\cdot n^{\Ocal(1)}$, which follows quite easily from known results for {\sc Vertex Cover}. Proving stronger lower bounds seems to be quite challenging.

Another open problem is whether it is possible to drop the time complexity of {\sc Elimination Distance to $\Gcal$} to $2^{\poly(k)}\cdot n^2$ for {\sl every} minor-closed graph class $\Gcal$.
We tend to believe that this should be possible. However, it seems to require to use branching ingeniously and, in particular,  to find equivalent instances of {\sc Elimination Distance to $\Gcal$} with a decreasing value of $k$.

As for the polynomial running time of our \FPT-algorithms, a priori, nothing prevents the existence of algorithms running in {\sl linear time}, although we are quite far from achieving this. Kawarabayashi~\cite{Kawarabayashi09plan} presented such a linear \FPT-algorithm for the {\sc Planarization} problem, heavily relying on the embedding on the resulting planar graph. Extending this technique to general minor-closed classes would require a very compact encoding of the (entangled) structure of minor-free graphs~\cite{RobertsonS03a} {that would be possible to handle in linear time}.

We also proposed an \XP-algorithm for {\sc Elimination Distance to $\Gcal$} parameterized by the treewidth of the input graph (with running time $n^{\Ocal(\tw^2)}$).
As mentioned in the introduction (see also~\cite{AgrawalKLPRSZ22}), the existence of an \FPT-algorithm for {\sc Elimination Distance to $\Gcal$}, parameterized by treewidth, {remains wide open and this is the case  even in the very special case where ${\cal G}$ contains only the empty graph, where  {\sc Elimination Distance to $\Gcal$} is equivalent to the problem of  computing treedepth.}

{A last direction is to improve the bounds on the size of the obstructions given in \autoref{@reconstructions}. We believe that any substantial improvement should demand novel methodologies that go beyond the irrelevant vertex technique.}

\newpage
\printbibliography

\end{document}